%% file: New_IEEEtran_main.tex
\documentclass[journal]{IEEEtran}
\usepackage{amsmath,amsfonts}
\usepackage{algorithmic}
\usepackage{array}
\usepackage[caption=false,font=normalsize,labelfont=sf,textfont=sf]{subfig}
\usepackage{textcomp}
\usepackage{stfloats}
\usepackage{url}
\usepackage{verbatim}
\usepackage{graphicx}
\def\BibTeX{{\rm B\kern-.05em{\sc i\kern-.025em b}\kern-.08em
    T\kern-.1667em\lower.7ex\hbox{E}\kern-.125emX}}
\usepackage{balance}

\usepackage{amsmath,amsfonts}
\usepackage{algorithm2e}
\SetKwComment{Comment}{$\triangleright$\ }{}

\RestyleAlgo{ruled}
\usepackage{array}
\usepackage{caption}
\usepackage{textcomp}
\usepackage{stfloats}
\usepackage{url}
\newtheorem{claim}{\textsc{Claim}}
\newtheorem{proof}{\textsc{Justification}}

\usepackage{verbatim}
\usepackage{graphicx}
\usepackage{cite}
\usepackage{pgfplots}
\usepackage{mathtools}
\usepackage[T1]{fontenc}
\usepackage[utf8]{inputenc}
\usepackage{fancyhdr}
\usepackage[normalem]{ulem}
\DeclarePairedDelimiter\ceil{\lceil}{\rceil}
\DeclarePairedDelimiter\floor{\lfloor}{\rfloor}

\newenvironment{customlegend}[1][]{%
	\begingroup
	\csname pgfplots@init@cleared@structures\endcsname
	\pgfplotsset{#1}%
}{%
	\csname pgfplots@createlegend\endcsname
	\endgroup
}%
\def\addlegendimage{\csname pgfplots@addlegendimage\endcsname}

\usepackage[colorinlistoftodos]{todonotes}
\newcommand{\suchthat}{\;\ifnum\currentgrouptype=16 \middle\fi|\;}

\newcommand\ngsout{\bgroup\markoverwith{\textcolor{blue}{\rule[0.5ex]{2pt}{0.4pt}}}\ULon}

\newcommand{\changepj}[1]{\textcolor{black}{#1}}
\newcommand{\changeng}[1]{\textcolor{black}{#1}}
\definecolor{LimeGreen}{rgb}{0.2, 0.8, 0.2}
\definecolor{ProcessBlue}{rgb}{0.0, 0.72, 0.92}

\pgfplotsset{compat=1.17}
\begin{document}

\title{Reducing Opinion Echo-Chambers by Intelligent Placement of Moderate-Minded Agents}

\author{

    \IEEEauthorblockN{Prithwish Jana\IEEEauthorrefmark{1},~\IEEEmembership{Member, IEEE}, Romit Roy Choudhury\IEEEauthorrefmark{2},~\IEEEmembership{Fellow, IEEE}, Niloy Ganguly\IEEEauthorrefmark{3},~\IEEEmembership{Senior Member, IEEE}}\\
    \thanks{\IEEEauthorrefmark{1}Prithwish Jana is with the David R. Cheriton School of Computer Science, University of Waterloo, Canada}
    \thanks{\IEEEauthorrefmark{2}Romit Roy Choudhury is with the Department of ECE and CS, University of Illinois at Urbana-Champaign, USA}
    \thanks{\IEEEauthorrefmark{3}Niloy Ganguly is with the Department of Computer Science and Engineering, Indian Institute of Technology Kharagpur, India}
}


\maketitle

\begin{abstract}
In the era of social media, people frequently share their own opinions online on various issues and also in the way, get exposed to others' opinions. Be it for the selective exposure of news feed recommendation algorithms or our own inclination to listen to opinions that support ours, the ultimate result is that we get more and more exposed to opinions that are closer to ours. Further, any population is inherently heterogeneous i.e. people will initially hold a varied range of opinions regarding a topic and showcase a varied range of openness to get influenced by other's opinions. In this paper, we demonstrate the different behavior put forward by open-minded and close-minded agents towards an issue, when they are allowed to freely intermix and communicate. 

We have shown that the intermixing among people leads to formation of opinion echo chambers i.e. a small closed network of people who hold similar opinions and are not affected by opinions of people outside the network. Echo chambers are evidently harmful for a society because it inhibits free healthy communication among all and thus, prevents exchange of opinions, spreads misinformation and increases extremist beliefs. This calls for reduction in echo chambers, because a total consensus of opinion is neither possible in a society nor is welcome. We show that the number of echo chambers depends on the number of close-minded agents and cannot be lessened by increasing the number of open-minded agents. It is in this context, that we identify certain `moderate'-minded agents, who possess the capability of manipulating and reducing the number of echo chambers. The paper proposes an algorithm for intelligent placement of moderate-minded agents in the opinion-time spectrum by which the opinion echo chambers can be maximally reduced. With various experimental setups, we demonstrate that the proposed algorithm fares well when compared to placement of other agents (open- or close-minded) and random placement of `moderate'-minded agents.
\end{abstract}


\section{Introduction}
\label{sec:intro}
\IEEEPARstart{B}{efore} the advent of social media, people  shared their views and opinions only with a limited number of close acquaintance who were mostly located in their geographic vicinity. After the outburst of online platforms like Twitter, Facebook, Pinterest, etc., the whole mode of opinion exchange and the audience for hearing one's views have changed. The audience of a particular social media post depends largely on the \changepj{\textit{content recommendation}} \textit{algorithms} used by that platform. \changepj{These} platform-specific \changepj{algorithms for} personalized recommendation \changepj{analyze our behavioral pattern on the social media platform. Thereby, they } regulate~\cite{etta2022comparing} the posts that are visible and highlighted to us \changepj{on our news feed}. Humans (unconsciously) have a strong desire~\cite{levendusky2013partisan} to listen to opinions that corroborate their own existing views. So, this typically results in two things. First, \changepj{by our own will,} our friend cycle in social media mostly consists of people whose inclination on various topics support ours. Secondly, recommendation algorithms present us with filtered posts (selective exposure~\cite{cinelli2021echo}) that are most likely to increase and enhance our online time on the platform~\cite{perra2019modelling}. 
In effect, both of these combined\changepj{ly} imply that the whole spectrum of opinions is oftentimes obscured from us or ignored by us. 
That is, we are surrounded by like-minded friends who in turn, expose us more to contents closer to our inclinations, effectively forming \textit{echo chambers}~\cite{cota2019quantifying}. 

\changepj{An echo chamber is a closed network of people who share similar beliefs and are largely disconnected from the outside world and their opinions. Researchers~\cite{carpenter2018impact,barbera2015tweeting,bakshy2015exposure} agree with the fact that Internet users on social media can ignore opinions that they do not relate to or dislike and thus, social media has an inarguable role in inculcating close-mindedness and giving rise to the echo chamber effect. It is thus believed that echo chambers reinforce preexisting beliefs and prevent a healthy intermixing of thoughts.} 

\changepj{But how do people's belief/opinion change when they are part of an echo chamber?} In our everyday life, broadly, there are three things that play a key role in the opinion formation process within us: (i) our own past experiences that influence our cognitive reasoning \changepj{towards preference for a topic}, (ii) sources of information like media broadcasts which can be neutral or biased towards a certain agenda and (iii) social interaction with our peers and close ones. Also opinions are transient and dynamic in nature. They can be swayed and manipulated by interpersonal interactions with other individuals or captivating speeches by influential personalities. In this context, let us consider an individual agent $A_{self}$ and a few other agents $A_{1}, A_{2}, \ldots, A_{k}$ who have the potential to influence the opinion of $A_{self}$. Thus, the dynamics of the opinion of $A_{self}$ will majorly depend on: (i) the \textit{openness} of $A_{self}$ to others' opinions i.e. how much is $A_{self}$ interested in listening to others' opinions (ii) the \textit{susceptibility}~\cite{hegselmann2002opinion} of $A_{self}$ towards getting influenced by others' opinions \changepj{i.e. even if $A_{self}$ is open to listen, if there are others with different opinions who can influence her} and (iii) the actual  \textit{influence} exerted by 
the other agents in convincing $A_{self}$ and swaying her   opinion towards their own opinions. Such interpersonal interactions and dynamics of opinion formation are intricately complex in the real world and depends subjectively from person to person. 

\changepj{Formation of echo chambers leads to amplification of held beliefs in a closed circle and thus can lead to spreading of misinformation and hampers our free judgement in an open debate. Further, as pointed out by Botte et al.~\cite{botte2022clustering}, the continuous content filtration on social media news feed is one of the main causes of today's localized clustering of opinions that leads to echo chambers in the long run. With this pretext and keeping in mind the negative effects of an opinion echo chamber, in this paper we focus on the ways by which a regulating body (like government) can restrict formation of multiple small \textit{opinion echo chambers} in today's population.}
\changeng{We observe that this can be achieved if we can bring about local consensus among larger groups. Even if a total consensus of opinions is not possible, we identify the ways by which we decrease the number of echo chambers formed. By reducing the number of such chambers, we can guarantee more people in each chamber. Thus, we can bring about more interaction in the population that is necessary for a healthy community and prevention of extreme opinions within small communities.} 

\section{Related Work}
\label{sec:relwork}

According to the Social Influence Theory~\cite{kelman1958compliance}, when people communicate, they tend to influence one another, effectively making their opinions more similar. In this line, the dynamics of opinion evolution depends on the choice of neighborhood and also the set of protocols that indicate the aggregation function of an agent set's opinions, who are believed to be influencers of one's opinion. With interaction among people in the population, the opinion of a person can take the form of the majoritarian opinion in the neighborhood~\cite{galam2002minority}, a random opinion in the neighborhood~\cite{holley1975ergodic} or opinion of the agent who is `nearest' (based on some pre-decided metric). The DeGroot, Friedkin-Johnsen, Deffuant-Weisbuch and Hegselmann-Krause are some well-known classical protocols/models, each having a set of rules determining how the agents interact. The real-life opinion dynamics among people however is not bound by such strict rules, but for our proposition of reducing echo chambers, we try to proceed with such a model that best mimics the dynamics of opinion on social media platforms.

In the DeGroot model~\cite{degroot1974reaching}, every person has weights corresponding to every other person in the population, that represents how much he/she is influenced by another person's opinions. In this model, the set of neighbors of each agent is fixed throughout and the dynamics follow synchronous updates for every agent in each discrete timestep, on a weighted averaging scheme for individual opinions. The Friedkin-Johnsen model~\cite{friedkin1990social,friedkin1999social} works on a similar line with a fixed network of neighbors but, this model has provision for stubborn agents who prefer to hold on more to their own opinions. In the Deffuant-Weisbuch model~\cite{deffuant2001mixing}, instead of every agent updating opinions at the same time, it is for random pair of agents. The problem with such models is their inadaptability to the modern real-world opinion exchange mechanism. Opinion exchange was in accordance with these model dynamics before the advent of the Internet and social media. At that time, people could only intermix with a fixed set of geographical neighbors. In recent times, people exchange their opinions online over social media platforms, where they can get introduced to anyone even on the opposite side of the planet, just by virtue of their similar trends of opinions~\cite{noorazar2020classical}. As such, our neighbors are no longer fixed and limited by geography, they vary dynamically depending on our current opinions and opinions of others in the population. This is possible by the Hegselmann-Krause (H-K) bounded confidence model~\cite{hegselmann2002opinion} that gives a more flexible view of the neighborhood definition by a time-varying adjacency matrix, as we will elaborate upon in \changepj{our section of problem formulation}. In the later sections of this paper, we will proceed with our propositions on the H-K model, with a dynamic interaction-network of neighbors.

Researchers like Pineda et al.~\cite{pineda2015mass} and Vasca et al.\cite{vasca2021practical} have brought forward the fact that our population is heterogeneous in nature and consists of close-minded and open-minded agents, having disparate behavior towards changing one's opinion. Baumann et al.~\cite{baumann2020modeling} believe that heterogeneous openness of persons, their different mixing patterns and polarized opinions are one of the key causes of echo chamber formation in a population, via interaction/exposure on social media. \changepj{Begby~\cite{begby2022belief} also mentions that people can start holding extremist beliefs as a result of opinion exchange within their close network of friends. Borelli et al.~\cite{borrelli2021quantitative} mentions that when there are two groups in any online discussion, the people focus on the differences between their group rather than the similarities. In the long run, this harbours prejudices, aggressive emotions towards the other group and affective polarization. Although this phenomenon and formation of echo chambers is by itself detrimental, but this is a natural and unavoidable entailment of any socio-cognitive sequence of actions. But, in this paper, we will show that the number of echo chambers can be reduced by introducing some external agents in the population. An interesting observation from the structural hole theory~\cite{lin2021structural} is that, unconnected groups of people in a social media may create a hole in the social structure, which can be occupied by externally-introduced individuals, who can help synthesize harmony between the two groups. With a similar aim, this paper focuses on ways by which the echo chamber effect can be minimized in today's population with a predominantly online presence. }

\section{Problem Formulation}
\label{sec:HKboundedconfidence}

We consider a population of $n$ agents. Each such agent $i\in\{1,2,\ldots,n\}$ possesses an opinion regarding a certain issue.  Opinions are represented in the form of a continuous variable $x_i(t)\in\left[0,1\right] $ that varies dynamically with time. Time instances $t=0,1,\ldots$ are considered to be discrete. At a particular time instance $t$, the opinion $x_i(t)$ of agent $i$ represents her inclination of belief towards two opposite poles $P_{left}$ and $P_{right}$.  Opposite poles may represent $\{P_{left}=$ liberal, $P_{right}=$ conservative$\}$, $\{P_{left}=$ Democrat, $P_{right}=$ Republican$\}$, $\{P_{left}=$ left-leaning partisan media, $P_{right}=$ right-leaning partisan media$\}$ etc. At time instance $t$, agents with belief completely in sync with  $P_{left}$ has $x_i(t)=0$, similarly one in sync with  $P_{right}$ has $x_i(t)=1$ and others have opinion scaled somewhere in the middle between 0 and 1. The \textit{opinion profile} of the population at time $t$ is thus expressed as {\bf x}(t)=$\left[ x_1(t), x_2(t), \ldots, x_n(t)\right] $.

With progress of time, each agent interacts with other agents in the population and her belief gets influenced. In \changepj{
the introduction},  we highlighted the fact that the rise of social media has resulted in the formation of echo chambers within a population. In an echo chamber, each agent is majorly influenced by like-minded friends who in turn, bear like-minded opinions. To emulate this in our framework, we consider the \textit{bounded confidence model}, where each agent $i$ has a fixed \textit{confidence interval} $\epsilon_i$ associated with it. At time instance $t$, the agent is influenced only by the neighboring agents whose opinions are located in the $\epsilon_i$-neighborhood of her opinion $x_i(t)$.  That is, at time $t$, the neighborhood $N_i(t)$ of an agent $i$ is defined by the set of agents:

\begin{equation}
	\label{eqn_HK_neighborsDef}
	N_i(t) = \left\lbrace j \suchthat \;\;\left|x_j(t)-x_i(t)\right|\leq\epsilon_i, \;\;\;j\in[1,n] \right\rbrace 
\end{equation}
	
As such, an agent $i$ confides only upon neighboring agents in its confidence interval. Also, an agent is a neighbor of itself. Thus, confidence interval defines the neighborhood  of a person -- more the interval, higher the number of neighbors for a person and the more open he is towards listening to others' opinions. A population where every agent has the same confidence interval $\epsilon$ is termed as a \textit{homogeneous} population. On the contrary, in real-life scenario, the population comprises people with varied confidence intervals i.e. a \textit{heterogeneous} population where different agents can have different $\epsilon_i$, as exemplified in Figure~\ref{fig:boundedConfPlot}. 

\begin{figure}[h!]
    \centering
    \subfloat[\label{fig:boundedConfPlot}]{
        \includegraphics[width=0.7\columnwidth]{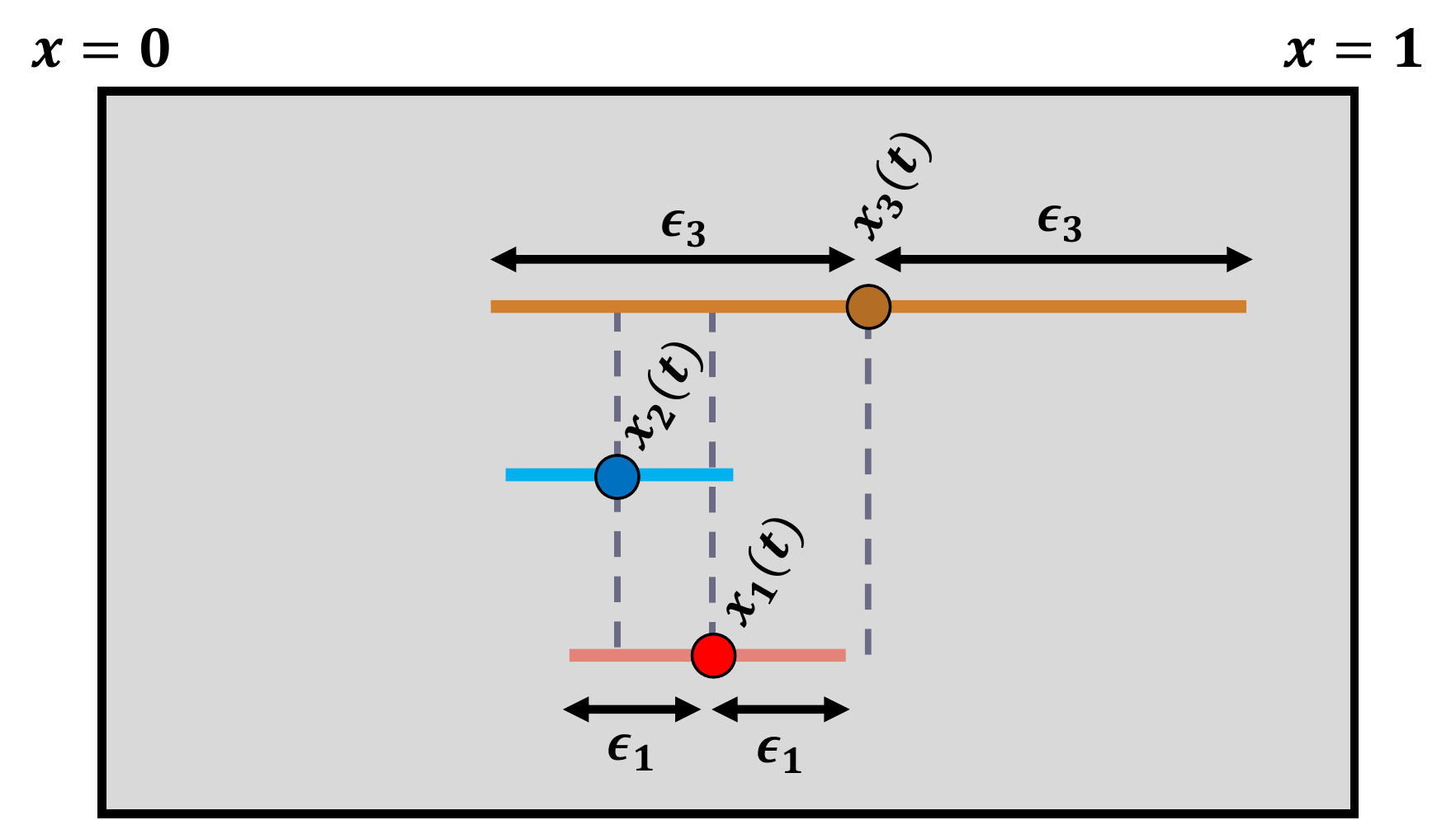}
    }
    \subfloat[\label{fig:boundedConfGraph}]{
        \includegraphics[width=0.25\columnwidth]{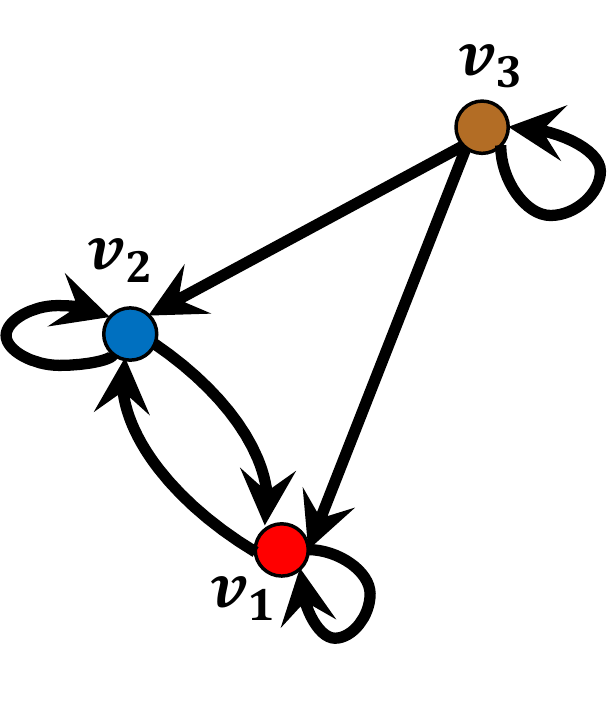}
    }
    \caption{Illustration of the bounded confidence model and graph-theoretic representation with a heterogeneous population. \textbf{(a)} Three agents $\{1,2,3\}$ with opinion profile $\left[x_1(t), x_2(t), x_3(t)\right]$. The $X$-axis denotes the whole $\left[0,1\right]$ opinion spectrum and confidence intervals are shown by colored lines. Agents $1$ and $2$ are mutual neighbors, both being neighbors of Agent $3$. \textbf{(b)} The corresponding graph $\mathcal{G}(\mathcal{V}, E,t)$ where a directed edge $(v_i\to v_j)\in E$ for $v_i,v_j\in\mathcal{V}$ indicates agent $i$ is influenced by agent $j$ at time-instance $t$.}
    \label{fig:boundedConf}
\end{figure}

\subsection{Hegselmann-Krause bounded confidence model}

The echo chamber formation closely relates to the Hegselmann-Krause (H-K) bounded confidence model~\cite{hegselmann2002opinion}. According to this, opinion of an agent $i$ at time instance $t+1$ takes the form of an average of all neighboring agents' opinions in $N_i(t)$. So, the opinion dynamics is expressed as:

\begin{equation}
	\label{eqn_HK}
	x^{HK}_i(t+1)=\frac{1}{|N_i(t)|}\times \sum_{j\in N_i(t)} x_j(t)
\end{equation}

As per the H-K model, in each iteration, the opinion of an agent gets changed to the average of the opinions of all her neighbors including herself. In reality, however, agents tend to hold on to their earlier opinion for a longer time. To state it otherwise, when changing one's opinion, an agent gives more weightage to her own earlier opinion than the weightage given to all other neighbor's opinions combined. This dynamics can be expressed by modifying Eqn.~\ref{eqn_HK} as:

\begin{equation}
	\label{eqn_HK_ownWeighted}
	x^{HK_{mod}}_i(t+1)=w_{own}\times x_i(t) + \frac{1-w_{own}}{|N_i(t)\setminus \{i\}|}\times \sum_{j\,\in\, N_i(t)\setminus \{i\}} x_j(t)
\end{equation}

where, $1-w_{own}<w_{own}\leq 1\implies w_{own}\in (0.5,1]$.

\subsection{Example}

Let's say that the opinion profile of a population of ten agents is $\mathbf{x}(T)=\left[ 0.1, 0.2, 0.4, 0.4, 0.5, 0.7, 0.7, 0.8, 0.8, 1.0\right] $ at time $t=T$. The fifth agent has an opinion of $x_5(T)=0.5$ and let's say it's confidence interval $\epsilon_5=0.25$. Accordingly, by Eqn.~\ref{eqn_HK_neighborsDef}, at this instance the neighbors of the fifth agent are $N_5(T)=\left\lbrace 3,4,5,6,7 \right\rbrace$. Considering the normal H-K opinion dynamics model of Eqn.~\ref{eqn_HK}, the agent will weigh all her neighbors' opinions as well as her own opinion equally by a factor of $\frac{1}{5}=0.2$. That is, the opinion of the fifth agent at time instance $t=T+1$ will be:

\begin{small}
\begin{equation}
	\label{eqn_HK_example_1}
	\begin{split}
	x^{HK}_5(T+1)&=\frac{1}{5}\times\left(x_3(T)+x_4(T)+x_5(T)+x_6(T)+x_7(T)\right)\\
	&=\frac{0.4+0.4+0.5+0.7+0.7}{5}
	~=~0.54
	\end{split}
\end{equation}
\end{small}

On the contrary as per the modified H-K model in Eqn.~\ref{eqn_HK_ownWeighted}, after weighing own opinion by $w_{own}=0.6$ (say), the fifth agent will weigh all her   neighbors' opinions by a factor of $\frac{1-0.6}{4}=0.1$. So, the opinion of the fifth agent at time instance $t=T+1$ will be:

\begin{small}
\begin{equation}
	\label{eqn_modHK_example_1}
	\begin{split}
	x^{HK_{mod}}_5(T+1)&=\left(0.6\right)\times x_5(T)+\\&\;\left(\frac{1-0.6}{4}\right)\times\left(x_3(T)+
    x_4(T)+x_6(T)+x_7(T)\right)\\
	&=0.6\times 0.5+\frac{0.4+0.4+0.7+0.7}{10}
	=0.52\\&\text{(which is closer to }\;x_5(T)\text{, than}\;x^{HK}_5(T+1)\text{)}
	\end{split}
\end{equation}
\end{small}

\subsection{Equilibrium Time.}
We define the \textit{equilibrium time} ($t_{eqm}$) as the time instance $t$ when   there is no further change in the opinion profiles. Mathematically, it can be represented as the minimum $t\in\{0,1,\ldots\}$ for which $\forall j \in \{1,2,\ldots,n\}$,  
$ \left|x_j(t) - x_j(t+1)\right| \leq \delta$, where $\delta$ is a very small positive real number tending to 0. The \textit{opinion clusters} ($C_{eqm}$) at the time of equilibrium ($t_{eqm}$) is defined as the distinct values in the opinion profile, that exist after equilibrium. A single opinion cluster at equilibrium indicates \textit{consensus}, two indicate \textit{polarization} and multiple clusters indicates \textit{fragmentation} of opinions in the population. 

\subsection{Graph-Theoretic Representation.}
We intend to represent the population and their respective neighborhood in the form of a graph network. Let's consider a time-varying graph $\mathcal{G}(\mathcal{V}, E,t)$. Here, each vertex $v_i\in \mathcal{V}$ corresponds to an agent $i\in\{1,2,\ldots,n\}$ in the population. As such, each vertex has an associated time-varying opinion $x_i(t)\in\left[0,1\right] $. There exists a directed edge $(v_i\to v_j)\in E$ if agent $j$ is a neighbor of agent $i$ i.e. opinion $x_j(t)$ lies in the $\epsilon_i$-neighborhood of opinion $x_i(t)$ i.e. $j\in N_i(t)$. To put it otherwise, edge $(v_i\to v_j)\in E$ indicates that agent $i$ is influenced by agent $j$. And the set $\{v \suchthat v_i\to v \in E \}$ forms the neighborhood $N_i(t)$ of agent $i$ at time $t$. From a vertex that corresponds to an agent $i$, there are outgoing edges to all those vertices (agents) that influence agent $i$'s opinion. Edge-weight is $1$ for all edges. This is depicted in Figure~\ref{fig:boundedConfGraph}.

Framing the problem of opinion dynamics by a graph helps in identifying the pull that an agent faces over time (leftward or rightward) while staying in a connected community, where everybody's opinion is evolving with time. The pull that an agent $a$ experiences can be represented towards the position of mean opinions of those agents, to which there are outgoing edges from $a$. In fact, whether opinions of two agents $a_i$ and $a_j$ will ultimately converge (i.e. whether consensus of opinion will be achieved) or not, is related to the presence of paths in the graph. Specifically, the necessary condition for such two agents to reach a consensus at certain stage depends on whether there exists a path (sequence of directed edges) from $a_i$ to $a_j$ or vice-versa. But it may be noted that, the presence of a path is not a sufficient condition for opinion consensus -- it depends on opinions of neighbors also. In the next section, we will categorize the agents based on their openness and will \changepj{then} explain how their respective dynamics of opinion can be explained with the help of such graphical representation.

\section{Agent Classification based on Openness}
\label{sec:agentClassif}
As mentioned in the previous section, any uniform random subset of a real-life social media population will be heterogeneous, with respect to confidence intervals of people. That is, people will have varied range of openness towards others' opinions. In this line, we try to categorize the agents \changepj{
in terms of equilibrium time ($t_{eqm}$) i.e. how fast they reach an equilibrium state of opinion.} 
\changepj{For each openness value i.e. confidence interval ($\delta$), we consider a homogeneous population where all agents have the same openness towards others' opinions -- and study how fast the whole population reaches an opinion equilibrium.} 

We consider a homogeneous population where every agent has the same confidence interval i.e. $\epsilon_i=\epsilon \;\forall i \in \{1,2,\ldots,n\}$. And at $t=0$, the opinions are evenly spaced between 0 and 1 (both inclusive) such that the difference of opinion between two agents of consecutive opinions is $\frac{1}{n-1}$. Effectively, this would mean that $\mathcal{G}(\mathcal{V}, E,t)$ is a $k$-regular graph where each node (barring the agents with extreme opinions closer to $0$ and $1$) has $k=\min(n,\; 2\times\floor{\epsilon*(n-1)}+1)$ outgoing edges and same number of incoming edges. So, as $\epsilon$ gets closer to $0$, the graph becomes a disconnected graph with $n$ vertices and $n$ edges by which every vertex is connected to itself. In absence of interaction between the vertices, it is already in equilibrium state, thus $t_{eqm}=0$. On increasing $\epsilon$, just when few edges start forming, $t_{eqm}$ should shoot up because $\epsilon$ (openness interval) is still low and it takes time for agents with small openness interval to come to a local consensus. In fact, the convergence to consensus becomes exponentially faster~\cite{li2020bounded} as openness increases and more and more edges get introduced. On increasing $\epsilon$ further, when it becomes closer to $1$, every vertex is connected to all other vertices including itself, and thus it becomes a fully connected graph with self-loops. Convergence to equilibrium in such a graph is very fast because every agent is moving towards the global opinion consensus average from the very beginning. So, intuitively the time to reach equilibrium $t_{eqm}$ should decrease monotonically with increase in $\epsilon$ in a homogeneous population, which in turn adds more inter-agent edges in the graph.

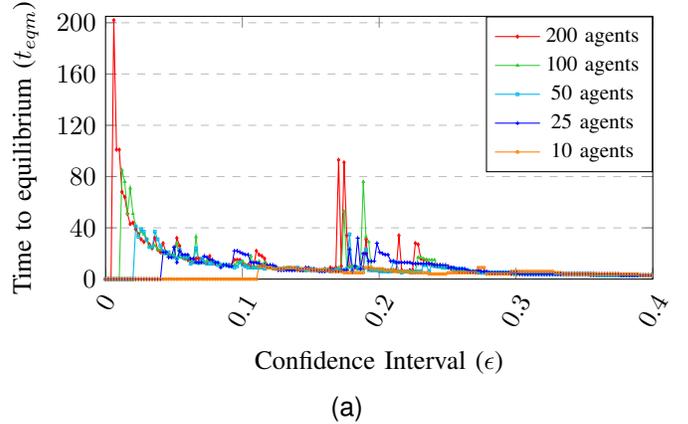
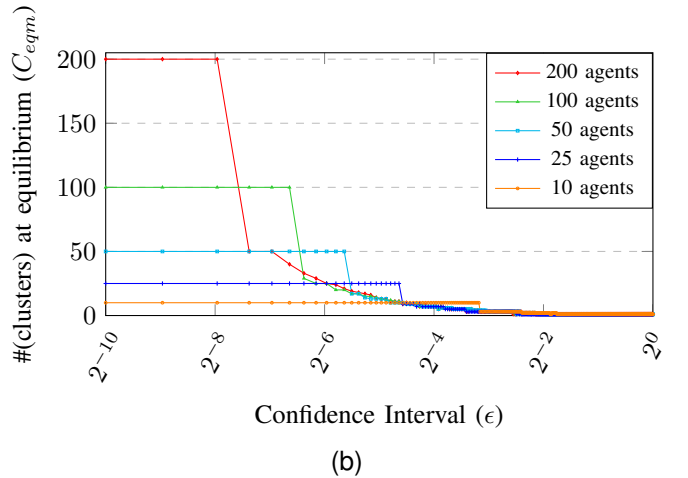
\begin{figure}[h!]
    \centering
    \subfloat[\label{fig:epoch_vs_openness}]{\pgfplotsset{width=\columnwidth,height=0.28\textwidth,compat=1.9}
                \input{Figures/epoch_vs_openness.tex}
    }\\
    \subfloat[\label{fig:cluster_vs_openness}]{\pgfplotsset{width=\columnwidth,height=0.28\textwidth,compat=1.9}
                \input{Figures/cluster_vs_openness}
    }
    \caption{Variation in time to reach equilibrium and number of clusters at equilibrium, for a homogeneous population where every agent has $\epsilon$ confidence interval.}
    \label{fig:epochcluster_vs_openness}
\end{figure}

But there are some important properties to be noted when experimenting with such a setup in practice. As depicted in Figure~\ref{fig:epochcluster_vs_openness}, we experimented with various population sizes. First, $t_{eqm}$ starts at $0$ for lower values of $\epsilon$, irrespective of the population size. It shoots up when edges just start to form in the graph. But as the population size gets larger, a relatively smaller value of $\epsilon$ is required for this sudden upward spike to occur in $t_{eqm}$. This is because, higher the population size, it is more likely to find people closer to one's opinion even if (s)he is open to listen to a very short interval around own. Upon increasing $\epsilon$  beyond $0.1$, the $t_{eqm}$ does not decrease monotonically. Instead, there is a certain region of confidence interval $\epsilon$ where the time to equilibrium becomes higher, as opposed to the trend. This is true for all population sizes and it happens when $\epsilon$ is in $\left[0.17,0.22\right]$, as evident from Figure~\ref{fig:epoch_vs_openness}. Further as in Figure~\ref{fig:cluster_vs_openness}, irrespective of population size, there is a global consensus ($C_{eqm}=1$) at equilibrium for $\epsilon\geq0.25$.

\changepj{Thus, as understandable from Figure~\ref{fig:epochcluster_vs_openness}, with change in confidence interval, there are notable effects over the time to reach equilibrium and also the number of opinion echo-chambers or clusters.} Keeping this in mind, we can identify three types of agents who may play different but significant roles in the formation of echo chambers. First type of agent are the \textit{close-minded} ones, whose openness/confidence interval $\epsilon_{close}$ is very close to zero. These agents interact with a very limited set of neighbors and it is a difficult feat to move them from their held beliefs. Second type comprises the \textit{open-minded agents} whose opinions can get influenced from almost all agents ranging over the whole spectrum of opinions. They can be moulded very easily and they take very less time to reach opinion consensus and equilibrium. Their confidence interval $\epsilon_{open}$ is large, beyond $0.5$ (from Figure~\ref{fig:epochcluster_vs_openness}, behaviors are similar for all $\epsilon\in\left[0.5,1\right]$). The third type of agents are the \textit{moderate-minded} agents whose confidence interval $\epsilon_{moderate}\in\left[0.17,0.22\right]$. They have openness more than close-minded but less than open-minded agents ($0\approx\epsilon_{close}<\epsilon_{moderate}<\epsilon_{open}\approx0.5$).

\section{Role of Different Agents in Cluster Formation}
\label{sec:roleClusterFormn}
As we explained in the previous section, three specific types of agents: \textit{open-}, \textit{close-} and \textit{moderate-minded} are of particular interest to us when analyzing the cluster formation behavior. As per definition, the open-minded agents have large confidence intervals $\epsilon_{open}$ and thus a larger neighborhood. The vertices corresponding to such agents will have high out-degree. On similar terms, vertices corresponding to close-minded agents (with low confidence intervals $\epsilon_{close}$) will have zero or very low out-degree. 

Research~\cite{rosenfeld2021can,liu2022deep} finds that during the pandemic, people have become more conservative of their own held beliefs and even their political ideologies have remained static over time. Also in another context, there was a study~\cite{knevzevic2022relationship} on a population comprising young adults in a conflict-hit society. The authors found that given specific adverse conditions, any person can develop strong attachment to a set of beliefs because these are related to inherent `human proclivities' -- and in this, some people will have substantially more tendency to develop such attachments than others. This serves as a verification of the fact that a society  always harbours some close-minded people who  strongly hold back to their held beliefs. Let's suppose that in a hypothetical situation, there are only close-minded agents who are not totally out of each other's influence. That is, there exists at least some edges in the corresponding graph $\mathcal{G}(\mathcal{V}, E,t)$. As per Figure~\ref{fig:epochcluster_vs_openness}, since this is a homogeneous population with low $\epsilon$, this will lead to a high number of clusters i.e. echo chambers (which is approximately equal to the number of agents) and will take a large time to reach equilibrium, given there are many agents. With this pretext, we present our first claim:

\vspace{2mm}

\begin{claim}
\textit{In a population containing close-minded people, introducing more close-minded people will increase the number of echo-chambers}
\end{claim}

\begin{proof}
This essentially means increasing the number of agents in a homogeneous population of close-minded agents. From Figure~\ref{fig:cluster_vs_openness}, for the same $\epsilon$, increasing the number of agents implies an increase in number of clusters or echo chambers.
\end{proof}

\vspace{2mm}

\begin{claim}
\textit{In a homogeneous population containing only close-minded agents, the intuitive approach of reducing opinion clusters will be to add more and more open-minded agents. But, in actuality, introducing new open-minded people will not be able to interfere with and reduce the existing echo-chambers}
\end{claim}

\begin{figure}[h!]
    \centering
    \subfloat[\label{fig:closeOpenExample_graph}]{\includegraphics[width=0.9\columnwidth]{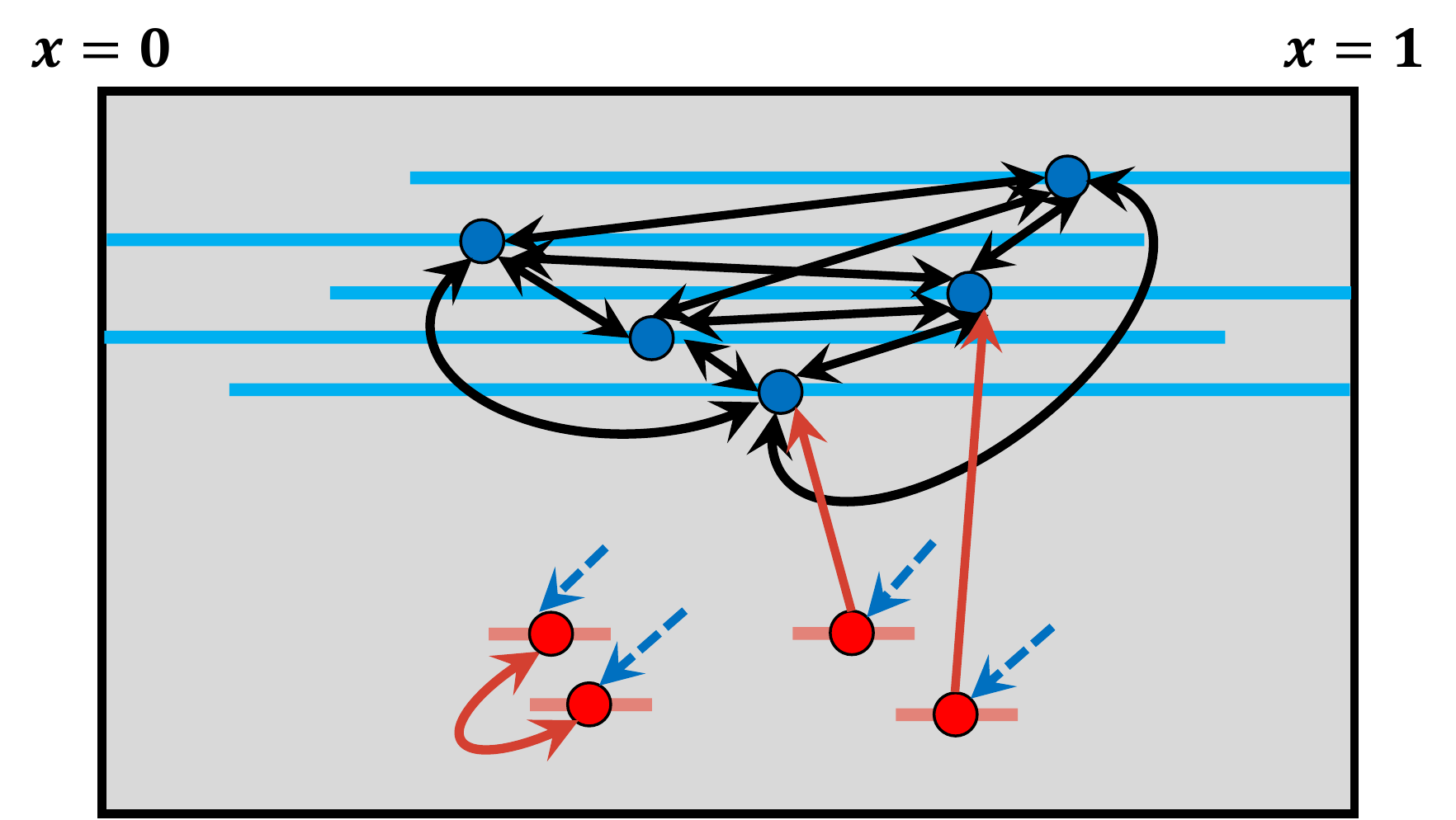}}\\
    \subfloat[\label{fig:closeOpenExample}]{\pgfplotsset{width=0.85\columnwidth,height=0.3\textwidth,compat=1.9}
            \input{Figures/closeOpenExample.tex}}
    \caption{An example of a population of five open-minded (blue) and four close-minded (red) agents, showing corresponding graph and opinion dynamics. In the graph, the black arrows denote the clique formed by open-minded agents and the blue-dotted arrows denote incoming arrows from all the five open-minded vertices.}
    \label{fig:closeOpenExample_plotGraph}
\end{figure}

\begin{proof}
\label{proof_newOpen}
Let's suppose $A_c$ be the set of close-minded agents, each with confidence interval $\epsilon_{c}$. $A_o$ be the set of open-minded agents, each with confidence interval $\epsilon_{o}$. They are randomly distributed over the opinion spectrum $\left[0,1\right]$. The confidence intervals of each open-minded agent $A\in A_o$ (one on each side) combinedly covers almost the whole opinion spectrum. So, it is very likely for another randomly-placed agent to be within the confidence interval of an open-minded agent $A\in A_o$ (high out-degree). Similarly, it is very likely for an open-minded agent $A\in A_o$ to be in another randomly-placed open-minded agent's confidence interval, given that the latter's confidence interval covers almost the whole opinion spectrum (high in-degree). As such, open-minded agents influence and are simultaneously influenced by other open-minded agents. As a result, vertices corresponding to open-minded agents form a strongly-connected component, if not a clique (as for example the blue vertices in Figure~\ref{fig:closeOpenExample_graph}). 

On the other hand, given that $\epsilon_{c}\approx 0$, it is very unlikely for another randomly-placed agent to be within the confidence interval of an agent $A\in A_c$. These vertices will have low out-degrees and will be found as pendant in-vertices attached to the strongly-connected open-minded component (as for example the red vertices in Figure~\ref{fig:closeOpenExample_graph}). Being a strongly-connected open-minded component, the open-minded agents will come to a majoritarian consensus amongst themselves at a exponentially fast pace. On the other hand, the close-minded agents are either already in equilibrium or left behind by the faster converging open-minded component. This is exemplified in Figure~\ref{fig:closeOpenExample}, where the close-minded (red) agents' cluster (echo-chamber) is left as-is and the open-minded (blue) agents form a single cluster at the $t=2$. So, the open-minded agents are not useful in reducing the existing echo-chambers formed by the close-minded agents, because of the huge difference in their corresponding time to reach equilibrium. 

\end{proof}

\vspace{2mm}

\begin{claim}
\textit{In a population containing mixture of close-minded and open-minded people, introducing moderate-minded people will help reduce the existing echo-chambers.} \end{claim}

\begin{proof}
In presence of close- and open-minded agents in a society, we will demonstrate that only moderate-minded agents have the capacity to bring down the overall number of clusters (echo chambers) at equilibrium. Let's suppose $A_c$ is the set of close-minded agents, each with confidence interval $\epsilon_{c}$. $A_o$ is the set of open-minded agents, each with confidence interval $\epsilon_{o}$. $A_m$ is the set of moderate-minded agents, each with confidence interval $\epsilon_{m}$. They are randomly distributed over the opinion spectrum $\left[0,1\right]$. As we discussed in Justification~\ref{proof_newOpen}, open-minded vertices form a strongly-connected component, close to a clique. Close-minded vertices will stay attached as pendant in-vertices to the strongly-connected open-minded components. So, this will not distort the clique structure much, \changepj{as compared to the case where new moderate-minded agents are introduced}. The moderate-minded vertices by virtue of its relatively higher out-degree than their close-minded counterparts, will stay more integrally attached to the strongly-connected open-minded components, \changepj{but not connected to every other agent in the existing clique}. As such, they will distort the clique structure. Thus, the fast-paced convergence among the open-minded agents that was happening by virtue of the clique structure, can no longer be achieved. And now, when the open-minded agents are converging at a slower pace, they will have the opportunity to gradually drag the slowly-converging or already-converged close-minded agents together into a local combined majoritarian consensus (bigger echo chamber) -- effectively giving rise to a lesser number of opinion clusters. 

\changepj{The relationship among the number of clusters at equilibrium and slowing down of the convergence process, is also validated by each of the plots in our experimental section viz., Figures~\ref{fig:closeToModerate_numEpochsClusters}, \ref{fig:openToModerate_numEpochsClusters} and \ref{fig:optimalPlcmnt_numEpochsClusters}. That is, on an average, there is a negative correlation between the time to reach equilibrium for a heterogeneous population and the number of clusters. Thus, slowing down the convergence process can eventually lead to lesser number of echo-chambers. Here, the moderate-minded agents act as a bridge between the slowly-converging/already-converged close-minded agents and the fast-converging open-minded agents, and are capable of slowing down the overall process of reaching opinion equilibrium in the population, resulting in a fewer number of opinion clusters.}

\end{proof}

\section{Intelligent Placement of Moderate-Minded Agents}
\label{sec:optimalModAgentPlacement}

\begin{algorithm}[h!]
\small
\caption{An algorithm to place moderate-minded agents at optimal space-time positions}\label{alg:optPlcmnt}
\KwData{Initial opinion profile $ x(t=0)$, Openness profile $\mathcal{E}$, Budget $\mathcal{B}$}
\KwResult{Space-Time positions of optimal placement of moderate-minded agents}

$t \gets 0$\;
\While{$\mathcal{B} > 0$}{
    Generate graph $\mathcal{G}(\mathcal{V}, E,t)$\; \Comment{initial at t=0, opinion dynamics for t>0}
    Arrange $v_i\in\mathcal{V}$ wrt corresponding opinion $x_i(t)$\;
    \For{$v_i\in\mathcal{V}$}{
        $E^{out}_i\gets$ set of outgoing edges from $v_i$\;
        $N^{L}_i\gets\{x_i(t)-x_k(t)|\; (v_i,v_k)\in E^{out}_i, x_i(t)<x_k(t) \}$\;
        $N^{R}_i\gets\{x_k(t)-x_i(t)|\; (v_i,v_k)\in E^{out}_i, x_i(t)>x_k(t) \}$\;
    }
    \For{$i = 1;\ i < \left|\mathcal{V}\right|;\ i = i + 1$}{
    \If{$\left(v_i,v_{i+1} \;\;\textit{are open-minded}\right)$ \textup{\textbf{and}} $\left(\sum{N^{L}_i} < \sum{N^{R}_i}\right)$ \textup{\textbf{and}} $\left(\sum{N^{L}_{i+1}} > \sum{N^{R}_{i+1}}\right)$}{
      $numNewL \gets \ceil{\frac{\sum{N^{R}_i}-\sum{N^{L}_i}}{\mathcal{E}_i}}$\;
      $numNewR \gets \ceil{\frac{\sum{N^{L}_{i+1}}-\sum{N^{R}_{i+1}}}{\mathcal{E}_{i+1}}}$\;
        \eIf{$\mathcal{B} \geq numNewL$}{
        Introduce $numNewL$ moderate-minded agents with opinion = $x_i(t) - \mathcal{E}_i$, at time = $t$\;
        $\mathcal{B} \gets \mathcal{B}-numNewL$\;
        }{\textbf{break}\;
        }
        \eIf{$\mathcal{B} \geq numNewR$}{
        Introduce $numNewR$ moderate-minded agents with opinion = $x_{i+1}(t) + \mathcal{E}_{i+1}$, at time = $t$\;
        $\mathcal{B} \gets \mathcal{B}-numNewR$\;
        }{\textbf{break}\;
        }      
    }
    }
    $t \gets t+1$\;
}
\end{algorithm}

It is understood that both close- and open-minded agents will not be helpful in reducing existing echo chambers in the community -- only moderate-minded agents have the power to bring this change. They have the potential to slow down the fast convergence of the open-minded agents and thus, bring about the chance for close-minded agents to interact with them. But, to slow down the convergence of the open-minded agents, the questions that arise are: How many external moderate-minded agents to place in the community? Where to place them in the opinion spectrum? And at what time instances during the dynamics of existing opinions, do they need to be introduced?

Let's suppose that there is a budget $\mathcal{B}$ of the maximum number of external moderate-minded agents that we can introduce in the population. Our motive is to place moderate-minded agents in such a way that they slow down (maximally) the reaching of consensus among the open-minded agents. At each timestep, this can be done by identifying a pair of consecutive (no other agent should have opinion in between them) open-minded agents $i, i+1$ who are moving towards each other. Thereafter, we can restrict/reverse this movement by introducing new moderate-minded agents. First, movement towards one another can happen only when the neighborhood of both the agent are not exactly same (if those were exactly same, both would have faced the exact same force and not opposite). So, even if their neighborhood coincide to some extent i.e. they have some common outgoing edges (in graph representation), they cannot be exactly equal.

 In graph representation, for the left agent $i$ of such a pair, the sum of opinions of the right neighbors should be more than the sum of opinions of left neighbors (because of this, it experiences a rightward pull by the opinion dynamics formulation in Eqn.~\ref{eqn_HK}). Similarly, the right agent $i+1$ of the pair will experience a leftward pull because sum of opinions of the left neighbors is more than the sum of opinions of right neighbors. To counter this attractive force, we introduce a leftward (rightward) pull to the left (right) agent of such a pair. To pull the left agent $i$ leftwards, we introduce a minimum number of agents to reverse its movement, at such a position that their $x$ is least to become neighbors of the former. To pull the right agent $i+1$ rightwards, we introduce a minimum number of agents to reverse its movement, at such a position that their $x$ is highest to become neighbors of the former. We continue doing this at each timestep until our budget of moderate-minded agents is covered up. The entire process is elucidated in Algorithm~\ref{alg:optPlcmnt}.

\section{Experimental Results}
\label{sec:experiments}
We now study the impact of introducing moderate agents in heterogeneous populations, both randomly and intelligently. 

\subsection{Effect of Random Placement of Moderate-Minded Agents in the Community}
\label{subsec:randomModAgentPlacement}
Let's suppose there is a population of $200$ people. Complying with real-life, we design it to be a heterogeneous population, with respect to the openness i.e. $\epsilon$-neighborhood of the people. Instead of every person having different openness, for simplicity, we consider that there are two homogeneous subgroups: an open-minded subgroup and a close-minded subgroup. Each person belonging to the close-minded subgroup has the same openness i.e. $\epsilon_c=0.01$, similarly every open-minded person has $\epsilon_o=0.45$. To emulate real-life scenario more closely, we do not differentiate on the initial opinion of close- and open-minded agents. It is not necessary that close-minded agents will hold extremist opinions (i.e. close to $x=0$ or $x=1$); they can very well have opinions somewhere near the central region. On a similar note, open-minded agents also will not always possess centrist opinions. So, we sample the initial opinion profiles, $x_o(t=0)$ and $x_c(t=0)$ of both the open- and close-minded subgroup respectively from a normal distribution $\mathcal{N}(0.5,\frac{0.5}{4})$ centred at $0.5$ and clipped between $0$ and $1$. The initial opinion profile $x(t=0)$ of the overall population is thus a concatenation of the respective profiles, $x_o(t=0)$ and $x_c(t=0)$. We now study the effect of converting existing agents to moderate-minded or introducing new moderate-minded agents, each of whom has $\epsilon_m=0.2$.

\vspace{2mm}

\subsubsection{A majorly close-minded initial population}

First, we consider a population that is comprised of $80\%$ close-minded agents and $20\%$ open-minded agents. Effectively, the average openness of this majorly close-minded population is $(0.8\times\epsilon_c)+(0.2\times\epsilon_o)=(0.8\times0.01)+(0.2\times0.45)=0.098$. With this initial configuration, we systematically go on replacing close-minded agents with moderate-minded agents until all of the former are removed. The initial opinion of all the close-minded agents are fixed. In Figure~\ref{fig:closeToModerate_numEpochsClusters}, we study (a) the time to reach equilibrium and (b) the number of clusters formed, as more and more close-minded people are replaced by moderate-minded. Since the initial positions of the moderate minded-agents are random in the opinion spectrum, we try out five different runs (Run \#1 - Run \#5) and take the average to understand the general trend. In each run, $x_o(t=0)$ is same. But each run corresponds to a different initial opinion profile for the moderate-minded agents $x_m(t=0)\subseteq x_c(t=0)$, and the remaining non-replaced close-minded agents $x_{cRem}(t=0)= x_c(t=0)\setminus x_m(t=0)$. The common trend among all these runs is that there is an increase in the time to reach equilibrium because of more moderate-minded agents in the population, irrespective of their initial position. In fact, when more close-minded people are replaced by moderate-minded agents, the average openness of the population increases. But unlike the general trend in a homogeneous population (Figure~\ref{fig:epoch_vs_openness}, where there's a decrease in equilibrium time with increase in openness), here the time increases with increase in average openness of the population. \changeng{This is with the exception that, in absence of any close-minded agents, the time to reach equilibrium and number of clusters both fall sharply.} In this increased time, larger clusters get formed and the overall number of clusters decreases. 

\begin{figure}[h!]
    \centering
    \subfloat[\label{fig:closeToModerate_numEpochs}]{\pgfplotsset{width=\columnwidth,height=0.28\textwidth,compat=1.9}\input{Figures/closeToModerate_numEpochs.tex}}\\
    \subfloat[\label{fig:closeToModerate_numClusters}]{\pgfplotsset{width=\columnwidth,height=0.28\textwidth,compat=1.9}\input{Figures/closeToModerate_numClusters.tex}}
    \caption{Effect of transforming close-minded agents to moderate-minded, in a population of 200 agents, where initial opinion profile $x(t=0)\in\mathcal{N}(0.5,\frac{0.5}{4})$ and contains $80\%$ close-minded agents and $20\%$ open-minded agents. Irrespective of which close-minded agents are chosen to be transformed, with more and more of such transformations in the heterogeneous population, the equilibrium time is delayed and larger (and lesser number of) opinion clusters get formed.}
    \label{fig:closeToModerate_numEpochsClusters}
\end{figure}

\subsubsection{A majorly open-minded initial population}
Secondly, we consider a majorly open-minded population that is comprised of $80\%$ open-minded agents and $20\%$ close-minded agents. Unlike previous setup, this time we systematically go on replacing open-minded agents with moderate-minded agents until all of the former are removed. Here, the initial opinion of all the close-minded agents are fixed. As more and more open-minded people are replaced by moderate-minded, in Figure~\ref{fig:openToModerate_numEpochsClusters} we plot the time to reach equilibrium and the number of clusters formed. Here also, we try out five different runs (Run \#1 - Run \#5) of the experiment, where $x_c(t=0)$ is same but initial opinion profile for the moderate-minded agents $x_m(t=0)$ and the remaining non-replaced open-minded agents are different. The common trend is again similar to the previous setup. With more randomly-placed moderate-minded agents in the population, the time to reach equilibrium increases and the overall number of clusters decreases.

\begin{figure}[h!]
    \centering
    \subfloat[\label{fig:openToModerate_numEpochs}]{\pgfplotsset{width=\columnwidth,height=0.28\textwidth,compat=1.9}
                \input{Figures/openToModerate_numEpochs.tex}}\\
    \subfloat[\label{fig:openToModerate_numClusters}]{\pgfplotsset{width=\columnwidth,height=0.28\textwidth,compat=1.9}\input{Figures/openToModerate_numClusters.tex}}
    \caption{Effect of transforming open-minded agents to moderate-minded, in a population of 200 agents, where initial opinion profile $x(t=0)\in\mathcal{N}(0.5,\frac{0.5}{4})$ and contains $80\%$ open-minded agents and $20\%$ close-minded agents. Irrespective of which open-minded agents are chosen to be transformed, a higher proportion of moderate-minded agents in the population, on an average, implies delayed equilibrium time and lesser number of opinion clusters.}
    \label{fig:openToModerate_numEpochsClusters}
\end{figure}

In Figure~\ref{fig:openToModerate_simulation}, we observe the opinion dynamics in the same opinion profile with fixed initial position of close-minded agents, but with different proportions of moderate- and open-minded people. It is observed that when there are no moderate-minded agents, it gives rise to more number of opinion clusters most of which comprise either the close-minded agents or the open-minded, but not both. On the contrary, in presence of some moderate-minded agents, the time to reach equilibrium is larger and the number of opinion clusters is reduced. Moreover, now close-minded agents share the same cluster as open- and moderate-minded agents.

\begin{figure}[h!]
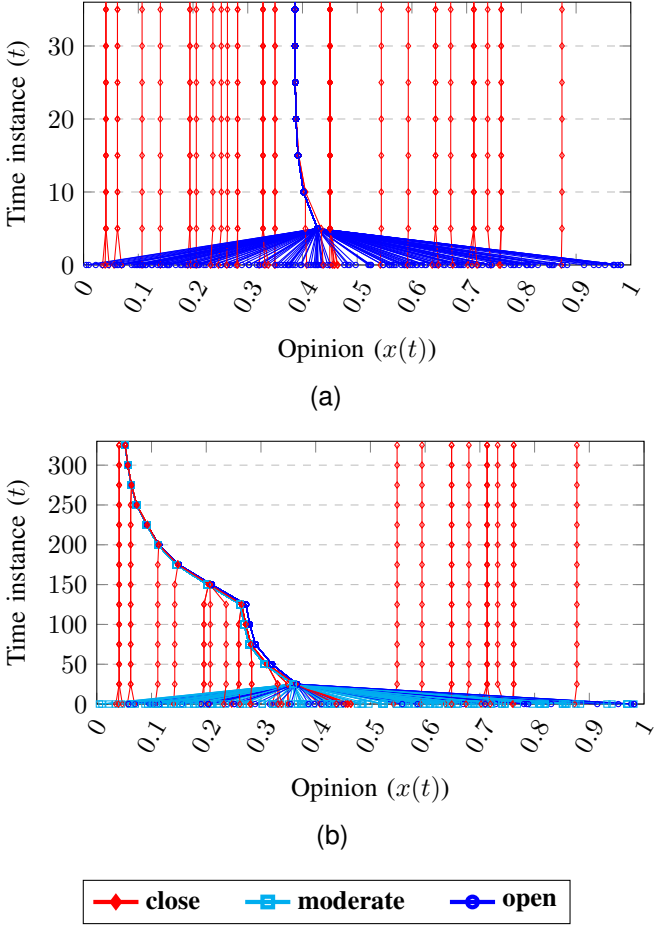

    \centering
    \subfloat[\label{fig:openToModerate_simulationZeroMod}]{\pgfplotsset{width=\columnwidth,height=0.28\textwidth,compat=1.9}
            \input{Figures/openToModerate_simulationZeroMod.tex}}\\
    \subfloat[\label{fig:openToModerate_simulationsixtyPercMod}]{\pgfplotsset{width=\columnwidth,height=0.28\textwidth,compat=1.9}
            \input{Figures/openToModerate_simulationsixtyPercMod.tex}}\\
    \subfloat{\pgfplotsset{width=.7\columnwidth,compat=1.9}
    \begin{tikzpicture}
        \begin{customlegend}[legend entries={{\bf close},{\bf moderate},{\bf open}},legend columns=3,legend style={/tikz/every even column/.append style={column sep=0.5cm}}]
        \addlegendimage{red,mark=diamond,ultra thick,sharp plot}
        \addlegendimage{cyan,mark=square,ultra thick,sharp plot}
        \addlegendimage{blue,mark=o,ultra thick,sharp plot}
        \end{customlegend}
    \end{tikzpicture}}
    \caption{Simulation of opinion dynamics in population of $20\%$ close-minded agents (fixed initial opinions) with \textbf{(a)} $80\%$ open-minded agents and \textbf{(b)} $45\%$ moderate- and $35\%$ open-minded agents, with same initial opinion profile. In \textbf{(a)}, there is one big opinion cluster comprising the open-minded agents and several small ones with the close-minded agents. But in presence of moderate-minded agents in \textbf{(b)}, the equilibrium time is slowed down whereby, larger opinion clusters get formed with the close- and open-minded agents together.}
    \label{fig:openToModerate_simulation}
\end{figure}

\vspace{2mm}
\subsubsection{An equally open- and close-minded initial population}
\label{subsubsec:newRandomModerate}

Thirdly, we consider a population  comprising $50\%$ open-minded and $50\%$ close-minded agents. Introducing new moderate-minded agents in the population, the trend is similar to the earlier setups. Here also, there is an increase in equilibrium time and decrease in the number of opinion clusters.

Overall, it can be observed that in a population consisting of only close- and moderate-minded agents, it is always the case that replacing existing agents with moderate-minded agents or introducing new, slows down the convergence process and in effect, reduces the number of opinion clusters. Even if moderate-minded agents are placed randomly in the opinion spectrum, this phenomenon is always observed in a heterogeneous population. Also, if the population contains some moderate-minded agents already, increasing their number will reduce the number of clusters, until the point where the population is solely comprised of moderate-minded agents.

\if(0)
\begin{figure}[h!]
    \centering
    \subfloat[\label{fig:newModerate_numEpochs}]{\pgfplotsset{width=\columnwidth,height=0.28\textwidth,compat=1.9}
                \input{Figures/newModerate_numEpochs.tex}}\\
    \subfloat[\label{fig:newModerate_numClusters}]{\pgfplotsset{width=\columnwidth,height=0.28\textwidth,compat=1.9}
                \input{Figures/newModerate_numClusters.tex}}
    \caption{Effect of introducing new moderate-minded agents, in a population of 200 agents, where initial opinion profile $x(t=0)\in\mathcal{N}(0.5,\frac{0.5}{4})$ and contains $50\%$ open-minded agents and $50\%$ close-minded agents. With new moderate-minded agents introduced at random positions in the initial opinion spectrum, the equilibrium time is delayed and the number of opinion clusters is reduced.}
    \label{fig:newModerate_numEpochsClusters}
\end{figure}
\fi

\begin{figure}[h!]
    \centering
    \subfloat[\label{fig:optimalPlcmnt_numEpochs}]{\pgfplotsset{width=\columnwidth,height=0.28\textwidth,compat=1.9}
            \input{Figures/optimalPlcmnt_numEpochs.tex}}\\
    \subfloat[\label{fig:optimalPlcmnt_numClusters}]{\pgfplotsset{width=\columnwidth,height=0.28\textwidth,compat=1.9}
            \input{Figures/optimalPlcmnt_numClusters.tex}}
    \caption{Comparison of the effect of introducing new moderate-minded agents, for random (Section~\ref{subsubsec:newRandomModerate}) vs. intelligent (Algorithm~\ref{alg:optPlcmnt}) placement. The population consists of 200 agents, where initial opinion profile $x(t=0)\in\mathcal{N}(0.5,\frac{0.5}{4})$ and contains $50\%$ open-minded and $50\%$ close-minded agents. When moderate-minded agents are placed intelligently at calculated positions over opinion space-time spectrum, the effect is much more pronounced in terms of delaying the equilibrium time and reducing (and enlarging) the opinion clusters.}
    \label{fig:optimalPlcmnt_numEpochsClusters}
\end{figure}

\subsection{Effect of Intelligent Placement of Moderate-Minded Agents}
\label{subsec:optimalModAgentPlacement}

We now evaluate our algorithm for intelligent placement of moderate-minded agents, as described in Algorithm~\ref{alg:optPlcmnt}. The setup is similar to that in random placement of moderate-minded agents. We consider a heterogeneous population of $200$ people, $50\%$ of whom are close-minded and remaining $50\%$ are open-minded. Each person belonging to the close-minded subgroup has $\epsilon_c=0.01$ and every open-minded person has $\epsilon_o=0.45$. Initial opinion profiles of both the open- and close-minded subgroups are sampled from a normal distribution $\mathcal{N}(0.5,\frac{0.5}{4})$ centred at $0.5$ and clipped between $0$ and $1$. We now introduce new moderate-minded agents ($\epsilon_m=0.20$) in the population, such that their placement in space and time is optimal (from Algorithm~\ref{alg:optPlcmnt}) and we study the effect in Figure~\ref{fig:optimalPlcmnt_numEpochsClusters}. It is observed that, in comparison to random placement, intelligent placement slows down the convergence process to a much greater extent. As a result, lesser number of opinion clusters are formed at equilibrium.

\section{Conclusion}
\label{sec:conclusion}

In this paper, we study the process of echo-chamber formation in populations where people have different ranges of openness towards accepting others' views. Modeling the dynamics of opinion in a bounded confidence model, it was observed that in such a population there is a tendency towards formation of a high number of opinion clusters (echo chambers) because of the disparity in convergence speed of different agents. There is a certain section of people having `moderate'-openness who have the power to bridge the gap between this disparity -- they delay the convergence of open-minded agents and allow close-minded agents to interact with the open-minded in this delayed time. It can be observed that in a population consisting of only close- and moderate-minded agents, replacing existing agents with moderate-minded or introducing new moderate-minded agents randomly  slows down the convergence process and in effect, reduces the number of opinion clusters. The paper also proposes an algorithm for intelligent placement of moderate-minded agents in the opinion space and time. This optimally slows down the fast convergence of open-minded agents and reduces the number of opinion echo chambers to a great extent, in comparison to random placement. Thus, external moderate-minded agents bearing specific opinions if introduced intelligently over time, can serve to be beneficial in diminishing the multiple small echo chambers and thus, facilitate more healthy exchange of opinions in the population.

\bibliographystyle{ieeetr}
\bibliography{myreferences}

\vspace{-12mm}

\begin{IEEEbiography}[{\includegraphics[width=1in,height=1.25in,clip,keepaspectratio]{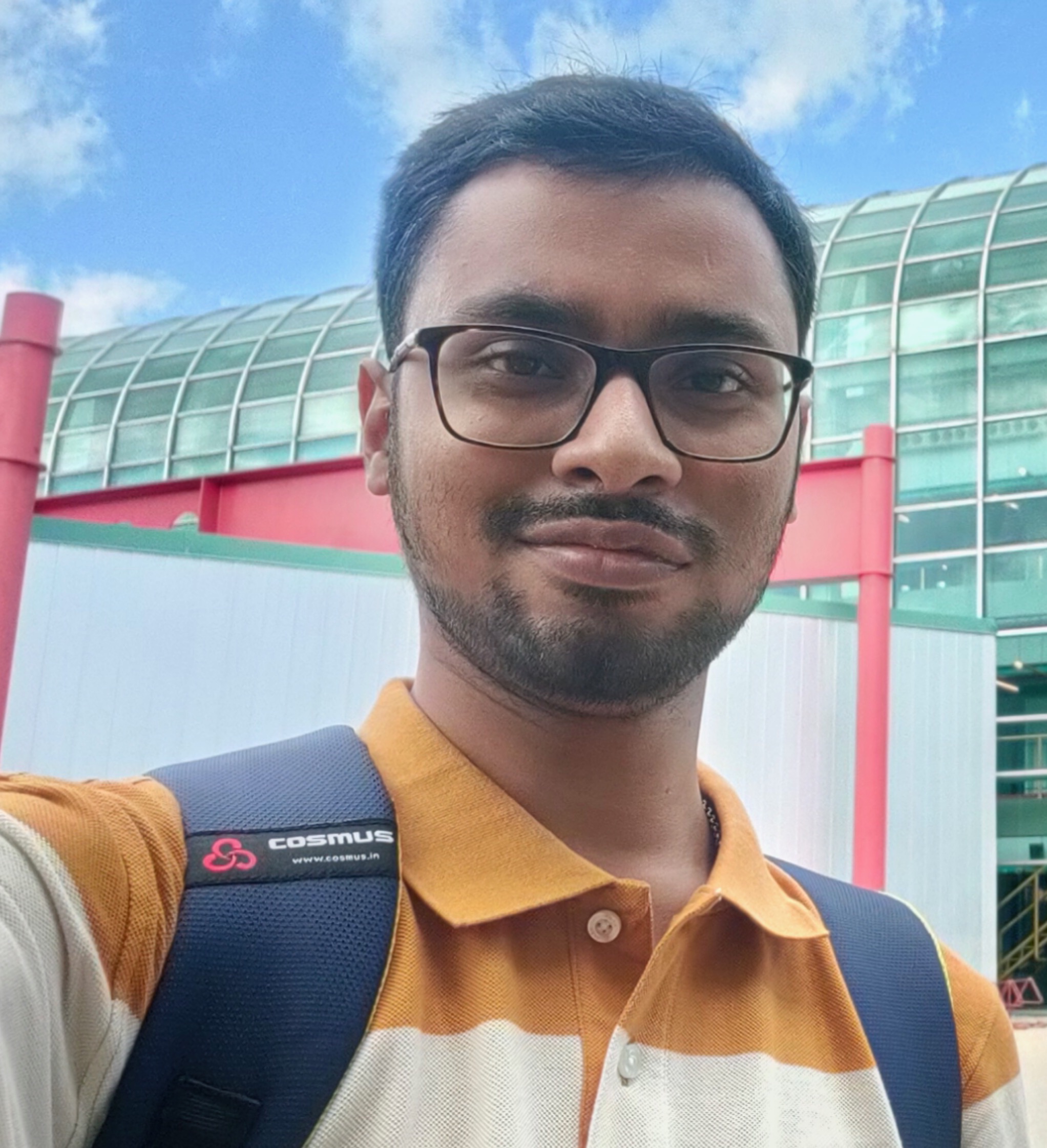}}]{Prithwish Jana} (Member, IEEE)
is a PhD student in Computer Science at the David R. Cheriton School of Computer Science, University of Waterloo, Canada. Prior to joining UW, he did his Masters (MTech) in Computer Science from IIT Kharagpur where he received the Institute Silver Medal and got selected for the Best Masters’ Thesis Award. He did his Bachelors (BE) in Computer Science from Jadavpur University, where he received the University Gold Medal for standing first in the department. He received prestigious accolades and fellowships from the Govt. of India viz., JBNSTS, KVPY and NTSE. His research interests include AI, logic solvers, machine learning, social computing, natural language processing and computer vision. \end{IEEEbiography}

\vspace{-12mm}

\begin{IEEEbiography}[{\includegraphics[width=1in,height=1.25in,clip,keepaspectratio]{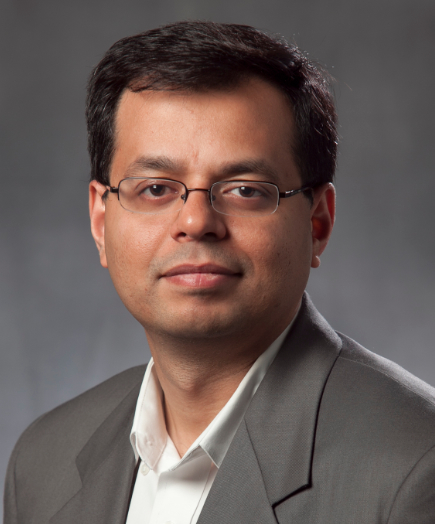}}]{Romit Roy Chowdhury} (Fellow, IEEE)
is a Jerry Sanders III AMD Scholar and Professor of ECE and CS at the University of Illinois at Urbana Champaign (UIUC). He joined UIUC from Fall 2013, prior to which he was an Associate Professor at Duke University. Romit received his PhD in the CS department of UIUC in Fall 2006. His research interests are in wireless networking, embedded sensing, and applied signal processing. Along with his students, he received a few research awards, including the ACM Sigmobile Rockstar Award, the UIUC Distinguished Alumni Award, the 2017 MobiSys Best Paper Award, etc. He was elevated to IEEE Fellow in 2018.\end{IEEEbiography}

\vspace{-12mm}

\begin{IEEEbiography}[{\includegraphics[width=1in,height=1.25in,clip,keepaspectratio]{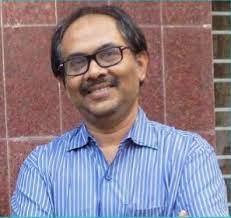}}]{Niloy Ganguly} (Senior Member, IEEE)
is a Professor in the Dept. of Computer Science and Engineering at IIT Kharagpur and a Fellow of Indian Academy of Engineering. He was a Visiting Professor in Leibnitz University of Hannover for two years from 2021 - 2022.  He has done his BTech from IIT Kharagpur and his PhD from IIEST, Shibpur. His research interests lie primarily in Social Computing, Machine Learning, and Network Science. He has regularly published in top international venues such as NeurIPS, KDD, ICDM, CSCW, AAAI, IJCAI, WWW, ACL, EMNLP, CHI, ICWSM,  IEEE and ACM Transaction etc and have received several best paper awards. \end{IEEEbiography}

\end{document}

%% file: Figures/epoch_vs_openness.tex
\begin{tikzpicture}
\begin{axis}[
xlabel={Confidence Interval ($\epsilon$)},
ylabel={Time to equilibrium ($t_{eqm}$)},
xmin=0, xmax=0.4,
ymin=0, ymax=205,
xtick={0,0.1,0.2,0.3,0.4,0.5,0.6,0.7,0.8,0.9,1.0},
ytick={0,40,80,120,160,200},
xticklabel style={rotate=60},
legend style={at={(1,1)},anchor=north east,font=\footnotesize},
ymajorgrids=true,
grid style=dashed
]

\addlegendentry{200 agents}
\addplot[
color=red,
mark=diamond,
mark size=0.6pt
]
coordinates {
(0,0)
(0.002004008,0)
(0.004008016,0)
(0.006012024,202)
(0.008016032,101)
(0.01002004,101)
(0.012024048,68)
(0.014028056,64)
(0.016032064,51)
(0.018036072,43)
(0.02004008,44)
(0.022044088,39)
(0.024048096,35)
(0.026052104,31)
(0.028056112,29)
(0.03006012,31)
(0.032064128,27)
(0.034068136,24)
(0.036072144,32)
(0.038076152,23)
(0.04008016,21)
(0.042084168,28)
(0.044088176,21)
(0.046092184,19)
(0.048096192,18)
(0.0501002,18)
(0.052104208,32)
(0.054108216,26)
(0.056112224,17)
(0.058116232,16)
(0.06012024,15)
(0.062124248,14)
(0.064128257,14)
(0.066132265,16)
(0.068136273,16)
(0.070140281,15)
(0.072144289,13)
(0.074148297,13)
(0.076152305,18)
(0.078156313,13)
(0.080160321,12)
(0.082164329,12)
(0.084168337,11)
(0.086172345,11)
(0.088176353,11)
(0.090180361,11)
(0.092184369,11)
(0.094188377,15)
(0.096192385,15)
(0.098196393,15)
(0.100200401,12)
(0.102204409,11)
(0.104208417,11)
(0.106212425,10)
(0.108216433,9)
(0.110220441,22)
(0.112224449,19)
(0.114228457,18)
(0.116232465,16)
(0.118236473,10)
(0.120240481,10)
(0.122244489,9)
(0.124248497,9)
(0.126252505,8)
(0.128256513,8)
(0.130260521,8)
(0.132264529,8)
(0.134268537,8)
(0.136272545,7)
(0.138276553,8)
(0.140280561,8)
(0.142284569,7)
(0.144288577,7)
(0.146292585,7)
(0.148296593,8)
(0.150300601,7)
(0.152304609,8)
(0.154308617,7)
(0.156312625,7)
(0.158316633,7)
(0.160320641,8)
(0.162324649,7)
(0.164328657,9)
(0.166332665,8)
(0.168336673,9)
(0.170340681,93)
(0.172344689,10)
(0.174348697,91)
(0.176352705,34)
(0.178356713,10)
(0.180360721,10)
(0.182364729,9)
(0.184368737,9)
(0.186372745,8)
(0.188376754,9)
(0.190380762,31)
(0.19238477,8)
(0.194388778,7)
(0.196392786,7)
(0.198396794,7)
(0.200400802,7)
(0.20240481,7)
(0.204408818,7)
(0.206412826,6)
(0.208416834,7)
(0.210420842,7)
(0.21242485,7)
(0.214428858,34)
(0.216432866,7)
(0.218436874,6)
(0.220440882,6)
(0.22244489,7)
(0.224448898,7)
(0.226452906,28)
(0.228456914,27)
(0.230460922,16)
(0.23246493,15)
(0.234468938,15)
(0.236472946,12)
(0.238476954,11)
(0.240480962,11)
(0.24248497,10)
(0.244488978,11)
(0.246492986,10)
(0.248496994,9)
(0.250501002,10)
(0.25250501,8)
(0.254509018,7)
(0.256513026,7)
(0.258517034,6)
(0.260521042,6)
(0.26252505,6)
(0.264529058,6)
(0.266533066,6)
(0.268537074,6)
(0.270541082,5)
(0.27254509,5)
(0.274549098,5)
(0.276553106,5)
(0.278557114,5)
(0.280561122,5)
(0.28256513,5)
(0.284569138,5)
(0.286573146,5)
(0.288577154,5)
(0.290581162,5)
(0.29258517,5)
(0.294589178,5)
(0.296593186,4)
(0.298597194,4)
(0.300601202,4)
(0.30260521,4)
(0.304609218,4)
(0.306613226,4)
(0.308617234,4)
(0.310621242,4)
(0.312625251,4)
(0.314629259,4)
(0.316633267,4)
(0.318637275,4)
(0.320641283,4)
(0.322645291,4)
(0.324649299,4)
(0.326653307,4)
(0.328657315,4)
(0.330661323,4)
(0.332665331,4)
(0.334669339,4)
(0.336673347,4)
(0.338677355,4)
(0.340681363,4)
(0.342685371,4)
(0.344689379,4)
(0.346693387,4)
(0.348697395,4)
(0.350701403,4)
(0.352705411,4)
(0.354709419,3)
(0.356713427,3)
(0.358717435,3)
(0.360721443,3)
(0.362725451,3)
(0.364729459,3)
(0.366733467,3)
(0.368737475,3)
(0.370741483,3)
(0.372745491,3)
(0.374749499,3)
(0.376753507,3)
(0.378757515,3)
(0.380761523,3)
(0.382765531,3)
(0.384769539,3)
(0.386773547,3)
(0.388777555,3)
(0.390781563,3)
(0.392785571,3)
(0.394789579,3)
(0.396793587,3)
(0.398797595,3)
(0.400801603,3)
(0.402805611,3)
(0.404809619,3)
(0.406813627,3)
(0.408817635,3)
(0.410821643,3)
(0.412825651,3)
(0.414829659,3)
(0.416833667,3)
(0.418837675,3)
(0.420841683,3)
(0.422845691,3)
(0.424849699,3)
(0.426853707,3)
(0.428857715,3)
(0.430861723,3)
(0.432865731,3)
(0.434869739,3)
(0.436873747,3)
(0.438877756,3)
(0.440881764,3)
(0.442885772,3)
(0.44488978,3)
(0.446893788,3)
(0.448897796,3)
(0.450901804,3)
(0.452905812,3)
(0.45490982,3)
(0.456913828,3)
(0.458917836,3)
(0.460921844,3)
(0.462925852,3)
(0.46492986,3)
(0.466933868,3)
(0.468937876,3)
(0.470941884,3)
(0.472945892,3)
(0.4749499,3)
(0.476953908,3)
(0.478957916,3)
(0.480961924,3)
(0.482965932,3)
(0.48496994,3)
(0.486973948,3)
(0.488977956,3)
(0.490981964,3)
(0.492985972,3)
(0.49498998,3)
(0.496993988,3)
(0.498997996,3)
(0.501002004,3)
(0.503006012,2)
(0.50501002,2)
(0.507014028,2)
(0.509018036,2)
(0.511022044,2)
(0.513026052,2)
(0.51503006,2)
(0.517034068,2)
(0.519038076,2)
(0.521042084,2)
(0.523046092,2)
(0.5250501,2)
(0.527054108,2)
(0.529058116,2)
(0.531062124,2)
(0.533066132,2)
(0.53507014,2)
(0.537074148,2)
(0.539078156,2)
(0.541082164,2)
(0.543086172,2)
(0.54509018,2)
(0.547094188,2)
(0.549098196,2)
(0.551102204,2)
(0.553106212,2)
(0.55511022,2)
(0.557114228,2)
(0.559118236,2)
(0.561122244,2)
(0.563126253,2)
(0.565130261,2)
(0.567134269,2)
(0.569138277,2)
(0.571142285,2)
(0.573146293,2)
(0.575150301,2)
(0.577154309,2)
(0.579158317,2)
(0.581162325,2)
(0.583166333,2)
(0.585170341,2)
(0.587174349,2)
(0.589178357,2)
(0.591182365,2)
(0.593186373,2)
(0.595190381,2)
(0.597194389,2)
(0.599198397,2)
(0.601202405,2)
(0.603206413,2)
(0.605210421,2)
(0.607214429,2)
(0.609218437,2)
(0.611222445,2)
(0.613226453,2)
(0.615230461,2)
(0.617234469,2)
(0.619238477,2)
(0.621242485,2)
(0.623246493,2)
(0.625250501,2)
(0.627254509,2)
(0.629258517,2)
(0.631262525,2)
(0.633266533,2)
(0.635270541,2)
(0.637274549,2)
(0.639278557,2)
(0.641282565,2)
(0.643286573,2)
(0.645290581,2)
(0.647294589,2)
(0.649298597,2)
(0.651302605,2)
(0.653306613,2)
(0.655310621,2)
(0.657314629,2)
(0.659318637,2)
(0.661322645,2)
(0.663326653,2)
(0.665330661,2)
(0.667334669,2)
(0.669338677,2)
(0.671342685,2)
(0.673346693,2)
(0.675350701,2)
(0.677354709,2)
(0.679358717,2)
(0.681362725,2)
(0.683366733,2)
(0.685370741,2)
(0.687374749,2)
(0.689378758,2)
(0.691382766,2)
(0.693386774,2)
(0.695390782,2)
(0.69739479,2)
(0.699398798,2)
(0.701402806,2)
(0.703406814,2)
(0.705410822,2)
(0.70741483,2)
(0.709418838,2)
(0.711422846,2)
(0.713426854,2)
(0.715430862,2)
(0.71743487,2)
(0.719438878,2)
(0.721442886,2)
(0.723446894,2)
(0.725450902,2)
(0.72745491,2)
(0.729458918,2)
(0.731462926,2)
(0.733466934,2)
(0.735470942,2)
(0.73747495,2)
(0.739478958,2)
(0.741482966,2)
(0.743486974,2)
(0.745490982,2)
(0.74749499,2)
(0.749498998,2)
(0.751503006,2)
(0.753507014,2)
(0.755511022,2)
(0.75751503,2)
(0.759519038,2)
(0.761523046,2)
(0.763527054,2)
(0.765531062,2)
(0.76753507,2)
(0.769539078,2)
(0.771543086,2)
(0.773547094,2)
(0.775551102,2)
(0.77755511,2)
(0.779559118,2)
(0.781563126,2)
(0.783567134,2)
(0.785571142,2)
(0.78757515,2)
(0.789579158,2)
(0.791583166,2)
(0.793587174,2)
(0.795591182,2)
(0.79759519,2)
(0.799599198,2)
(0.801603206,2)
(0.803607214,2)
(0.805611222,2)
(0.80761523,2)
(0.809619238,2)
(0.811623246,2)
(0.813627255,2)
(0.815631263,2)
(0.817635271,2)
(0.819639279,2)
(0.821643287,2)
(0.823647295,2)
(0.825651303,2)
(0.827655311,2)
(0.829659319,2)
(0.831663327,2)
(0.833667335,2)
(0.835671343,2)
(0.837675351,2)
(0.839679359,2)
(0.841683367,2)
(0.843687375,2)
(0.845691383,2)
(0.847695391,2)
(0.849699399,2)
(0.851703407,2)
(0.853707415,2)
(0.855711423,2)
(0.857715431,2)
(0.859719439,2)
(0.861723447,2)
(0.863727455,2)
(0.865731463,2)
(0.867735471,2)
(0.869739479,2)
(0.871743487,2)
(0.873747495,2)
(0.875751503,2)
(0.877755511,2)
(0.879759519,2)
(0.881763527,2)
(0.883767535,2)
(0.885771543,2)
(0.887775551,2)
(0.889779559,2)
(0.891783567,2)
(0.893787575,2)
(0.895791583,2)
(0.897795591,2)
(0.899799599,2)
(0.901803607,2)
(0.903807615,2)
(0.905811623,2)
(0.907815631,2)
(0.909819639,2)
(0.911823647,2)
(0.913827655,2)
(0.915831663,2)
(0.917835671,2)
(0.919839679,2)
(0.921843687,2)
(0.923847695,2)
(0.925851703,2)
(0.927855711,2)
(0.929859719,2)
(0.931863727,2)
(0.933867735,2)
(0.935871743,2)
(0.937875752,2)
(0.93987976,2)
(0.941883768,2)
(0.943887776,2)
(0.945891784,2)
(0.947895792,2)
(0.9498998,2)
(0.951903808,2)
(0.953907816,2)
(0.955911824,2)
(0.957915832,2)
(0.95991984,2)
(0.961923848,2)
(0.963927856,2)
(0.965931864,2)
(0.967935872,2)
(0.96993988,2)
(0.971943888,2)
(0.973947896,2)
(0.975951904,2)
(0.977955912,2)
(0.97995992,2)
(0.981963928,2)
(0.983967936,2)
(0.985971944,2)
(0.987975952,2)
(0.98997996,2)
(0.991983968,2)
(0.993987976,2)
(0.995991984,2)
(0.997995992,2)
(1,1)
};

\addlegendentry{100 agents}
\addplot[
color=LimeGreen,
mark=triangle,
mark size=0.5pt
]
coordinates {
(0,0)
(0.002004008,0)
(0.004008016,0)
(0.006012024,0)
(0.008016032,0)
(0.01002004,0)
(0.012024048,85)
(0.014028056,76)
(0.016032064,51)
(0.018036072,71)
(0.02004008,51)
(0.022044088,42)
(0.024048096,34)
(0.026052104,37)
(0.028056112,35)
(0.03006012,30)
(0.032064128,26)
(0.034068136,25)
(0.036072144,26)
(0.038076152,23)
(0.04008016,23)
(0.042084168,20)
(0.044088176,21)
(0.046092184,18)
(0.048096192,18)
(0.0501002,19)
(0.052104208,28)
(0.054108216,17)
(0.056112224,17)
(0.058116232,17)
(0.06012024,17)
(0.062124248,14)
(0.064128257,14)
(0.066132265,33)
(0.068136273,12)
(0.070140281,14)
(0.072144289,14)
(0.074148297,14)
(0.076152305,15)
(0.078156313,12)
(0.080160321,13)
(0.082164329,11)
(0.084168337,11)
(0.086172345,11)
(0.088176353,11)
(0.090180361,10)
(0.092184369,10)
(0.094188377,11)
(0.096192385,13)
(0.098196393,13)
(0.100200401,14)
(0.102204409,11)
(0.104208417,10)
(0.106212425,18)
(0.108216433,9)
(0.110220441,9)
(0.112224449,14)
(0.114228457,9)
(0.116232465,13)
(0.118236473,10)
(0.120240481,11)
(0.122244489,9)
(0.124248497,9)
(0.126252505,8)
(0.128256513,8)
(0.130260521,9)
(0.132264529,8)
(0.134268537,8)
(0.136272545,8)
(0.138276553,7)
(0.140280561,7)
(0.142284569,7)
(0.144288577,7)
(0.146292585,7)
(0.148296593,7)
(0.150300601,7)
(0.152304609,7)
(0.154308617,7)
(0.156312625,7)
(0.158316633,7)
(0.160320641,8)
(0.162324649,7)
(0.164328657,7)
(0.166332665,7)
(0.168336673,8)
(0.170340681,9)
(0.172344689,8)
(0.174348697,53)
(0.176352705,9)
(0.178356713,10)
(0.180360721,10)
(0.182364729,10)
(0.184368737,9)
(0.186372745,9)
(0.188376754,76)
(0.190380762,33)
(0.19238477,29)
(0.194388778,7)
(0.196392786,7)
(0.198396794,7)
(0.200400802,7)
(0.20240481,7)
(0.204408818,6)
(0.206412826,7)
(0.208416834,6)
(0.210420842,7)
(0.21242485,6)
(0.214428858,6)
(0.216432866,6)
(0.218436874,6)
(0.220440882,6)
(0.22244489,6)
(0.224448898,6)
(0.226452906,6)
(0.228456914,17)
(0.230460922,17)
(0.23246493,16)
(0.234468938,15)
(0.236472946,15)
(0.238476954,15)
(0.240480962,15)
(0.24248497,10)
(0.244488978,9)
(0.246492986,9)
(0.248496994,9)
(0.250501002,9)
(0.25250501,9)
(0.254509018,7)
(0.256513026,7)
(0.258517034,7)
(0.260521042,6)
(0.26252505,6)
(0.264529058,6)
(0.266533066,6)
(0.268537074,6)
(0.270541082,6)
(0.27254509,5)
(0.274549098,5)
(0.276553106,5)
(0.278557114,5)
(0.280561122,5)
(0.28256513,5)
(0.284569138,5)
(0.286573146,5)
(0.288577154,5)
(0.290581162,5)
(0.29258517,5)
(0.294589178,5)
(0.296593186,5)
(0.298597194,4)
(0.300601202,4)
(0.30260521,4)
(0.304609218,4)
(0.306613226,4)
(0.308617234,4)
(0.310621242,4)
(0.312625251,4)
(0.314629259,4)
(0.316633267,4)
(0.318637275,4)
(0.320641283,4)
(0.322645291,4)
(0.324649299,4)
(0.326653307,4)
(0.328657315,4)
(0.330661323,4)
(0.332665331,4)
(0.334669339,4)
(0.336673347,4)
(0.338677355,4)
(0.340681363,4)
(0.342685371,4)
(0.344689379,4)
(0.346693387,4)
(0.348697395,4)
(0.350701403,4)
(0.352705411,4)
(0.354709419,3)
(0.356713427,3)
(0.358717435,3)
(0.360721443,3)
(0.362725451,3)
(0.364729459,3)
(0.366733467,3)
(0.368737475,3)
(0.370741483,3)
(0.372745491,3)
(0.374749499,3)
(0.376753507,3)
(0.378757515,3)
(0.380761523,3)
(0.382765531,3)
(0.384769539,3)
(0.386773547,3)
(0.388777555,3)
(0.390781563,3)
(0.392785571,3)
(0.394789579,3)
(0.396793587,3)
(0.398797595,3)
(0.400801603,3)
(0.402805611,3)
(0.404809619,3)
(0.406813627,3)
(0.408817635,3)
(0.410821643,3)
(0.412825651,3)
(0.414829659,3)
(0.416833667,3)
(0.418837675,3)
(0.420841683,3)
(0.422845691,3)
(0.424849699,3)
(0.426853707,3)
(0.428857715,3)
(0.430861723,3)
(0.432865731,3)
(0.434869739,3)
(0.436873747,3)
(0.438877756,3)
(0.440881764,3)
(0.442885772,3)
(0.44488978,3)
(0.446893788,3)
(0.448897796,3)
(0.450901804,3)
(0.452905812,3)
(0.45490982,3)
(0.456913828,3)
(0.458917836,3)
(0.460921844,3)
(0.462925852,3)
(0.46492986,3)
(0.466933868,3)
(0.468937876,3)
(0.470941884,3)
(0.472945892,3)
(0.4749499,3)
(0.476953908,3)
(0.478957916,3)
(0.480961924,3)
(0.482965932,3)
(0.48496994,3)
(0.486973948,3)
(0.488977956,3)
(0.490981964,3)
(0.492985972,3)
(0.49498998,3)
(0.496993988,3)
(0.498997996,3)
(0.501002004,3)
(0.503006012,3)
(0.50501002,3)
(0.507014028,2)
(0.509018036,2)
(0.511022044,2)
(0.513026052,2)
(0.51503006,2)
(0.517034068,2)
(0.519038076,2)
(0.521042084,2)
(0.523046092,2)
(0.5250501,2)
(0.527054108,2)
(0.529058116,2)
(0.531062124,2)
(0.533066132,2)
(0.53507014,2)
(0.537074148,2)
(0.539078156,2)
(0.541082164,2)
(0.543086172,2)
(0.54509018,2)
(0.547094188,2)
(0.549098196,2)
(0.551102204,2)
(0.553106212,2)
(0.55511022,2)
(0.557114228,2)
(0.559118236,2)
(0.561122244,2)
(0.563126253,2)
(0.565130261,2)
(0.567134269,2)
(0.569138277,2)
(0.571142285,2)
(0.573146293,2)
(0.575150301,2)
(0.577154309,2)
(0.579158317,2)
(0.581162325,2)
(0.583166333,2)
(0.585170341,2)
(0.587174349,2)
(0.589178357,2)
(0.591182365,2)
(0.593186373,2)
(0.595190381,2)
(0.597194389,2)
(0.599198397,2)
(0.601202405,2)
(0.603206413,2)
(0.605210421,2)
(0.607214429,2)
(0.609218437,2)
(0.611222445,2)
(0.613226453,2)
(0.615230461,2)
(0.617234469,2)
(0.619238477,2)
(0.621242485,2)
(0.623246493,2)
(0.625250501,2)
(0.627254509,2)
(0.629258517,2)
(0.631262525,2)
(0.633266533,2)
(0.635270541,2)
(0.637274549,2)
(0.639278557,2)
(0.641282565,2)
(0.643286573,2)
(0.645290581,2)
(0.647294589,2)
(0.649298597,2)
(0.651302605,2)
(0.653306613,2)
(0.655310621,2)
(0.657314629,2)
(0.659318637,2)
(0.661322645,2)
(0.663326653,2)
(0.665330661,2)
(0.667334669,2)
(0.669338677,2)
(0.671342685,2)
(0.673346693,2)
(0.675350701,2)
(0.677354709,2)
(0.679358717,2)
(0.681362725,2)
(0.683366733,2)
(0.685370741,2)
(0.687374749,2)
(0.689378758,2)
(0.691382766,2)
(0.693386774,2)
(0.695390782,2)
(0.69739479,2)
(0.699398798,2)
(0.701402806,2)
(0.703406814,2)
(0.705410822,2)
(0.70741483,2)
(0.709418838,2)
(0.711422846,2)
(0.713426854,2)
(0.715430862,2)
(0.71743487,2)
(0.719438878,2)
(0.721442886,2)
(0.723446894,2)
(0.725450902,2)
(0.72745491,2)
(0.729458918,2)
(0.731462926,2)
(0.733466934,2)
(0.735470942,2)
(0.73747495,2)
(0.739478958,2)
(0.741482966,2)
(0.743486974,2)
(0.745490982,2)
(0.74749499,2)
(0.749498998,2)
(0.751503006,2)
(0.753507014,2)
(0.755511022,2)
(0.75751503,2)
(0.759519038,2)
(0.761523046,2)
(0.763527054,2)
(0.765531062,2)
(0.76753507,2)
(0.769539078,2)
(0.771543086,2)
(0.773547094,2)
(0.775551102,2)
(0.77755511,2)
(0.779559118,2)
(0.781563126,2)
(0.783567134,2)
(0.785571142,2)
(0.78757515,2)
(0.789579158,2)
(0.791583166,2)
(0.793587174,2)
(0.795591182,2)
(0.79759519,2)
(0.799599198,2)
(0.801603206,2)
(0.803607214,2)
(0.805611222,2)
(0.80761523,2)
(0.809619238,2)
(0.811623246,2)
(0.813627255,2)
(0.815631263,2)
(0.817635271,2)
(0.819639279,2)
(0.821643287,2)
(0.823647295,2)
(0.825651303,2)
(0.827655311,2)
(0.829659319,2)
(0.831663327,2)
(0.833667335,2)
(0.835671343,2)
(0.837675351,2)
(0.839679359,2)
(0.841683367,2)
(0.843687375,2)
(0.845691383,2)
(0.847695391,2)
(0.849699399,2)
(0.851703407,2)
(0.853707415,2)
(0.855711423,2)
(0.857715431,2)
(0.859719439,2)
(0.861723447,2)
(0.863727455,2)
(0.865731463,2)
(0.867735471,2)
(0.869739479,2)
(0.871743487,2)
(0.873747495,2)
(0.875751503,2)
(0.877755511,2)
(0.879759519,2)
(0.881763527,2)
(0.883767535,2)
(0.885771543,2)
(0.887775551,2)
(0.889779559,2)
(0.891783567,2)
(0.893787575,2)
(0.895791583,2)
(0.897795591,2)
(0.899799599,2)
(0.901803607,2)
(0.903807615,2)
(0.905811623,2)
(0.907815631,2)
(0.909819639,2)
(0.911823647,2)
(0.913827655,2)
(0.915831663,2)
(0.917835671,2)
(0.919839679,2)
(0.921843687,2)
(0.923847695,2)
(0.925851703,2)
(0.927855711,2)
(0.929859719,2)
(0.931863727,2)
(0.933867735,2)
(0.935871743,2)
(0.937875752,2)
(0.93987976,2)
(0.941883768,2)
(0.943887776,2)
(0.945891784,2)
(0.947895792,2)
(0.9498998,2)
(0.951903808,2)
(0.953907816,2)
(0.955911824,2)
(0.957915832,2)
(0.95991984,2)
(0.961923848,2)
(0.963927856,2)
(0.965931864,2)
(0.967935872,2)
(0.96993988,2)
(0.971943888,2)
(0.973947896,2)
(0.975951904,2)
(0.977955912,2)
(0.97995992,2)
(0.981963928,2)
(0.983967936,2)
(0.985971944,2)
(0.987975952,2)
(0.98997996,2)
(0.991983968,2)
(0.993987976,2)
(0.995991984,2)
(0.997995992,2)
(1,1)
};

\addlegendentry{50 agents}
\addplot[
color=ProcessBlue,
mark=square,
mark size=0.5pt
]
coordinates {
(0,0)
(0.002004008,0)
(0.004008016,0)
(0.006012024,0)
(0.008016032,0)
(0.01002004,0)
(0.012024048,0)
(0.014028056,0)
(0.016032064,0)
(0.018036072,0)
(0.02004008,0)
(0.022044088,41)
(0.024048096,33)
(0.026052104,39)
(0.028056112,37)
(0.03006012,31)
(0.032064128,25)
(0.034068136,25)
(0.036072144,37)
(0.038076152,31)
(0.04008016,26)
(0.042084168,22)
(0.044088176,20)
(0.046092184,21)
(0.048096192,18)
(0.0501002,17)
(0.052104208,16)
(0.054108216,23)
(0.056112224,18)
(0.058116232,16)
(0.06012024,17)
(0.062124248,12)
(0.064128257,13)
(0.066132265,24)
(0.068136273,13)
(0.070140281,16)
(0.072144289,13)
(0.074148297,12)
(0.076152305,12)
(0.078156313,12)
(0.080160321,13)
(0.082164329,11)
(0.084168337,11)
(0.086172345,11)
(0.088176353,11)
(0.090180361,11)
(0.092184369,10)
(0.094188377,9)
(0.096192385,10)
(0.098196393,13)
(0.100200401,11)
(0.102204409,10)
(0.104208417,9)
(0.106212425,9)
(0.108216433,9)
(0.110220441,9)
(0.112224449,9)
(0.114228457,9)
(0.116232465,8)
(0.118236473,9)
(0.120240481,9)
(0.122244489,9)
(0.124248497,8)
(0.126252505,8)
(0.128256513,8)
(0.130260521,8)
(0.132264529,8)
(0.134268537,9)
(0.136272545,8)
(0.138276553,8)
(0.140280561,9)
(0.142284569,8)
(0.144288577,7)
(0.146292585,7)
(0.148296593,7)
(0.150300601,8)
(0.152304609,7)
(0.154308617,7)
(0.156312625,7)
(0.158316633,7)
(0.160320641,7)
(0.162324649,7)
(0.164328657,6)
(0.166332665,7)
(0.168336673,7)
(0.170340681,7)
(0.172344689,7)
(0.174348697,7)
(0.176352705,8)
(0.178356713,35)
(0.180360721,8)
(0.182364729,10)
(0.184368737,8)
(0.186372745,8)
(0.188376754,8)
(0.190380762,23)
(0.19238477,7)
(0.194388778,7)
(0.196392786,7)
(0.198396794,7)
(0.200400802,6)
(0.20240481,6)
(0.204408818,6)
(0.206412826,6)
(0.208416834,6)
(0.210420842,14)
(0.21242485,13)
(0.214428858,6)
(0.216432866,5)
(0.218436874,6)
(0.220440882,6)
(0.22244489,6)
(0.224448898,5)
(0.226452906,6)
(0.228456914,6)
(0.230460922,6)
(0.23246493,11)
(0.234468938,11)
(0.236472946,6)
(0.238476954,11)
(0.240480962,10)
(0.24248497,10)
(0.244488978,10)
(0.246492986,9)
(0.248496994,9)
(0.250501002,9)
(0.25250501,9)
(0.254509018,8)
(0.256513026,8)
(0.258517034,8)
(0.260521042,8)
(0.26252505,7)
(0.264529058,7)
(0.266533066,6)
(0.268537074,6)
(0.270541082,6)
(0.27254509,6)
(0.274549098,5)
(0.276553106,5)
(0.278557114,5)
(0.280561122,5)
(0.28256513,5)
(0.284569138,5)
(0.286573146,5)
(0.288577154,5)
(0.290581162,5)
(0.29258517,5)
(0.294589178,5)
(0.296593186,5)
(0.298597194,5)
(0.300601202,5)
(0.30260521,5)
(0.304609218,5)
(0.306613226,4)
(0.308617234,4)
(0.310621242,4)
(0.312625251,4)
(0.314629259,4)
(0.316633267,4)
(0.318637275,4)
(0.320641283,4)
(0.322645291,4)
(0.324649299,4)
(0.326653307,4)
(0.328657315,4)
(0.330661323,4)
(0.332665331,4)
(0.334669339,4)
(0.336673347,4)
(0.338677355,4)
(0.340681363,4)
(0.342685371,4)
(0.344689379,4)
(0.346693387,4)
(0.348697395,4)
(0.350701403,4)
(0.352705411,4)
(0.354709419,4)
(0.356713427,4)
(0.358717435,3)
(0.360721443,3)
(0.362725451,3)
(0.364729459,3)
(0.366733467,3)
(0.368737475,3)
(0.370741483,3)
(0.372745491,3)
(0.374749499,3)
(0.376753507,3)
(0.378757515,3)
(0.380761523,3)
(0.382765531,3)
(0.384769539,3)
(0.386773547,3)
(0.388777555,3)
(0.390781563,3)
(0.392785571,3)
(0.394789579,3)
(0.396793587,3)
(0.398797595,3)
(0.400801603,3)
(0.402805611,3)
(0.404809619,3)
(0.406813627,3)
(0.408817635,3)
(0.410821643,3)
(0.412825651,3)
(0.414829659,3)
(0.416833667,3)
(0.418837675,3)
(0.420841683,3)
(0.422845691,3)
(0.424849699,3)
(0.426853707,3)
(0.428857715,3)
(0.430861723,3)
(0.432865731,3)
(0.434869739,3)
(0.436873747,3)
(0.438877756,3)
(0.440881764,3)
(0.442885772,3)
(0.44488978,3)
(0.446893788,3)
(0.448897796,3)
(0.450901804,3)
(0.452905812,3)
(0.45490982,3)
(0.456913828,3)
(0.458917836,3)
(0.460921844,3)
(0.462925852,3)
(0.46492986,3)
(0.466933868,3)
(0.468937876,3)
(0.470941884,3)
(0.472945892,3)
(0.4749499,3)
(0.476953908,3)
(0.478957916,3)
(0.480961924,3)
(0.482965932,3)
(0.48496994,3)
(0.486973948,3)
(0.488977956,3)
(0.490981964,3)
(0.492985972,3)
(0.49498998,3)
(0.496993988,3)
(0.498997996,3)
(0.501002004,3)
(0.503006012,3)
(0.50501002,3)
(0.507014028,3)
(0.509018036,3)
(0.511022044,2)
(0.513026052,2)
(0.51503006,2)
(0.517034068,2)
(0.519038076,2)
(0.521042084,2)
(0.523046092,2)
(0.5250501,2)
(0.527054108,2)
(0.529058116,2)
(0.531062124,2)
(0.533066132,2)
(0.53507014,2)
(0.537074148,2)
(0.539078156,2)
(0.541082164,2)
(0.543086172,2)
(0.54509018,2)
(0.547094188,2)
(0.549098196,2)
(0.551102204,2)
(0.553106212,2)
(0.55511022,2)
(0.557114228,2)
(0.559118236,2)
(0.561122244,2)
(0.563126253,2)
(0.565130261,2)
(0.567134269,2)
(0.569138277,2)
(0.571142285,2)
(0.573146293,2)
(0.575150301,2)
(0.577154309,2)
(0.579158317,2)
(0.581162325,2)
(0.583166333,2)
(0.585170341,2)
(0.587174349,2)
(0.589178357,2)
(0.591182365,2)
(0.593186373,2)
(0.595190381,2)
(0.597194389,2)
(0.599198397,2)
(0.601202405,2)
(0.603206413,2)
(0.605210421,2)
(0.607214429,2)
(0.609218437,2)
(0.611222445,2)
(0.613226453,2)
(0.615230461,2)
(0.617234469,2)
(0.619238477,2)
(0.621242485,2)
(0.623246493,2)
(0.625250501,2)
(0.627254509,2)
(0.629258517,2)
(0.631262525,2)
(0.633266533,2)
(0.635270541,2)
(0.637274549,2)
(0.639278557,2)
(0.641282565,2)
(0.643286573,2)
(0.645290581,2)
(0.647294589,2)
(0.649298597,2)
(0.651302605,2)
(0.653306613,2)
(0.655310621,2)
(0.657314629,2)
(0.659318637,2)
(0.661322645,2)
(0.663326653,2)
(0.665330661,2)
(0.667334669,2)
(0.669338677,2)
(0.671342685,2)
(0.673346693,2)
(0.675350701,2)
(0.677354709,2)
(0.679358717,2)
(0.681362725,2)
(0.683366733,2)
(0.685370741,2)
(0.687374749,2)
(0.689378758,2)
(0.691382766,2)
(0.693386774,2)
(0.695390782,2)
(0.69739479,2)
(0.699398798,2)
(0.701402806,2)
(0.703406814,2)
(0.705410822,2)
(0.70741483,2)
(0.709418838,2)
(0.711422846,2)
(0.713426854,2)
(0.715430862,2)
(0.71743487,2)
(0.719438878,2)
(0.721442886,2)
(0.723446894,2)
(0.725450902,2)
(0.72745491,2)
(0.729458918,2)
(0.731462926,2)
(0.733466934,2)
(0.735470942,2)
(0.73747495,2)
(0.739478958,2)
(0.741482966,2)
(0.743486974,2)
(0.745490982,2)
(0.74749499,2)
(0.749498998,2)
(0.751503006,2)
(0.753507014,2)
(0.755511022,2)
(0.75751503,2)
(0.759519038,2)
(0.761523046,2)
(0.763527054,2)
(0.765531062,2)
(0.76753507,2)
(0.769539078,2)
(0.771543086,2)
(0.773547094,2)
(0.775551102,2)
(0.77755511,2)
(0.779559118,2)
(0.781563126,2)
(0.783567134,2)
(0.785571142,2)
(0.78757515,2)
(0.789579158,2)
(0.791583166,2)
(0.793587174,2)
(0.795591182,2)
(0.79759519,2)
(0.799599198,2)
(0.801603206,2)
(0.803607214,2)
(0.805611222,2)
(0.80761523,2)
(0.809619238,2)
(0.811623246,2)
(0.813627255,2)
(0.815631263,2)
(0.817635271,2)
(0.819639279,2)
(0.821643287,2)
(0.823647295,2)
(0.825651303,2)
(0.827655311,2)
(0.829659319,2)
(0.831663327,2)
(0.833667335,2)
(0.835671343,2)
(0.837675351,2)
(0.839679359,2)
(0.841683367,2)
(0.843687375,2)
(0.845691383,2)
(0.847695391,2)
(0.849699399,2)
(0.851703407,2)
(0.853707415,2)
(0.855711423,2)
(0.857715431,2)
(0.859719439,2)
(0.861723447,2)
(0.863727455,2)
(0.865731463,2)
(0.867735471,2)
(0.869739479,2)
(0.871743487,2)
(0.873747495,2)
(0.875751503,2)
(0.877755511,2)
(0.879759519,2)
(0.881763527,2)
(0.883767535,2)
(0.885771543,2)
(0.887775551,2)
(0.889779559,2)
(0.891783567,2)
(0.893787575,2)
(0.895791583,2)
(0.897795591,2)
(0.899799599,2)
(0.901803607,2)
(0.903807615,2)
(0.905811623,2)
(0.907815631,2)
(0.909819639,2)
(0.911823647,2)
(0.913827655,2)
(0.915831663,2)
(0.917835671,2)
(0.919839679,2)
(0.921843687,2)
(0.923847695,2)
(0.925851703,2)
(0.927855711,2)
(0.929859719,2)
(0.931863727,2)
(0.933867735,2)
(0.935871743,2)
(0.937875752,2)
(0.93987976,2)
(0.941883768,2)
(0.943887776,2)
(0.945891784,2)
(0.947895792,2)
(0.9498998,2)
(0.951903808,2)
(0.953907816,2)
(0.955911824,2)
(0.957915832,2)
(0.95991984,2)
(0.961923848,2)
(0.963927856,2)
(0.965931864,2)
(0.967935872,2)
(0.96993988,2)
(0.971943888,2)
(0.973947896,2)
(0.975951904,2)
(0.977955912,2)
(0.97995992,2)
(0.981963928,2)
(0.983967936,2)
(0.985971944,2)
(0.987975952,2)
(0.98997996,2)
(0.991983968,2)
(0.993987976,2)
(0.995991984,2)
(0.997995992,2)
(1,1)
};

\addlegendentry{25 agents}
\addplot[
color=blue,
mark=+,
mark size=0.8pt
]
coordinates {
(0,0)
(0.002004008,0)
(0.004008016,0)
(0.006012024,0)
(0.008016032,0)
(0.01002004,0)
(0.012024048,0)
(0.014028056,0)
(0.016032064,0)
(0.018036072,0)
(0.02004008,0)
(0.022044088,0)
(0.024048096,0)
(0.026052104,0)
(0.028056112,0)
(0.03006012,0)
(0.032064128,0)
(0.034068136,0)
(0.036072144,0)
(0.038076152,0)
(0.04008016,0)
(0.042084168,21)
(0.044088176,21)
(0.046092184,17)
(0.048096192,17)
(0.0501002,25)
(0.052104208,13)
(0.054108216,22)
(0.056112224,19)
(0.058116232,19)
(0.06012024,19)
(0.062124248,16)
(0.064128257,16)
(0.066132265,13)
(0.068136273,13)
(0.070140281,13)
(0.072144289,18)
(0.074148297,17)
(0.076152305,15)
(0.078156313,15)
(0.080160321,13)
(0.082164329,13)
(0.084168337,9)
(0.086172345,11)
(0.088176353,11)
(0.090180361,10)
(0.092184369,10)
(0.094188377,22)
(0.096192385,22)
(0.098196393,21)
(0.100200401,20)
(0.102204409,19)
(0.104208417,19)
(0.106212425,13)
(0.108216433,13)
(0.110220441,13)
(0.112224449,12)
(0.114228457,11)
(0.116232465,11)
(0.118236473,11)
(0.120240481,11)
(0.122244489,11)
(0.124248497,10)
(0.126252505,7)
(0.128256513,7)
(0.130260521,7)
(0.132264529,7)
(0.134268537,7)
(0.136272545,7)
(0.138276553,7)
(0.140280561,7)
(0.142284569,7)
(0.144288577,7)
(0.146292585,9)
(0.148296593,9)
(0.150300601,7)
(0.152304609,7)
(0.154308617,7)
(0.156312625,6)
(0.158316633,6)
(0.160320641,6)
(0.162324649,7)
(0.164328657,7)
(0.166332665,7)
(0.168336673,6)
(0.170340681,6)
(0.172344689,7)
(0.174348697,7)
(0.176352705,7)
(0.178356713,23)
(0.180360721,6)
(0.182364729,10)
(0.184368737,32)
(0.186372745,8)
(0.188376754,20)
(0.190380762,20)
(0.19238477,14)
(0.194388778,14)
(0.196392786,14)
(0.198396794,28)
(0.200400802,21)
(0.20240481,20)
(0.204408818,19)
(0.206412826,19)
(0.208416834,14)
(0.210420842,14)
(0.21242485,13)
(0.214428858,13)
(0.216432866,13)
(0.218436874,13)
(0.220440882,13)
(0.22244489,13)
(0.224448898,12)
(0.226452906,12)
(0.228456914,12)
(0.230460922,12)
(0.23246493,12)
(0.234468938,12)
(0.236472946,12)
(0.238476954,12)
(0.240480962,11)
(0.24248497,11)
(0.244488978,11)
(0.246492986,11)
(0.248496994,11)
(0.250501002,9)
(0.25250501,9)
(0.254509018,9)
(0.256513026,8)
(0.258517034,8)
(0.260521042,8)
(0.26252505,8)
(0.264529058,7)
(0.266533066,6)
(0.268537074,6)
(0.270541082,6)
(0.27254509,7)
(0.274549098,6)
(0.276553106,6)
(0.278557114,6)
(0.280561122,5)
(0.28256513,5)
(0.284569138,5)
(0.286573146,5)
(0.288577154,5)
(0.290581162,5)
(0.29258517,5)
(0.294589178,5)
(0.296593186,5)
(0.298597194,5)
(0.300601202,5)
(0.30260521,4)
(0.304609218,4)
(0.306613226,4)
(0.308617234,4)
(0.310621242,4)
(0.312625251,4)
(0.314629259,4)
(0.316633267,4)
(0.318637275,4)
(0.320641283,4)
(0.322645291,4)
(0.324649299,4)
(0.326653307,4)
(0.328657315,4)
(0.330661323,4)
(0.332665331,4)
(0.334669339,4)
(0.336673347,4)
(0.338677355,4)
(0.340681363,4)
(0.342685371,4)
(0.344689379,4)
(0.346693387,4)
(0.348697395,4)
(0.350701403,4)
(0.352705411,4)
(0.354709419,4)
(0.356713427,4)
(0.358717435,4)
(0.360721443,4)
(0.362725451,4)
(0.364729459,4)
(0.366733467,4)
(0.368737475,4)
(0.370741483,4)
(0.372745491,4)
(0.374749499,4)
(0.376753507,3)
(0.378757515,3)
(0.380761523,3)
(0.382765531,3)
(0.384769539,3)
(0.386773547,3)
(0.388777555,3)
(0.390781563,3)
(0.392785571,3)
(0.394789579,3)
(0.396793587,3)
(0.398797595,3)
(0.400801603,3)
(0.402805611,3)
(0.404809619,3)
(0.406813627,3)
(0.408817635,3)
(0.410821643,3)
(0.412825651,3)
(0.414829659,3)
(0.416833667,3)
(0.418837675,3)
(0.420841683,3)
(0.422845691,3)
(0.424849699,3)
(0.426853707,3)
(0.428857715,3)
(0.430861723,3)
(0.432865731,3)
(0.434869739,3)
(0.436873747,3)
(0.438877756,3)
(0.440881764,3)
(0.442885772,3)
(0.44488978,3)
(0.446893788,3)
(0.448897796,3)
(0.450901804,3)
(0.452905812,3)
(0.45490982,3)
(0.456913828,3)
(0.458917836,3)
(0.460921844,3)
(0.462925852,3)
(0.46492986,3)
(0.466933868,3)
(0.468937876,3)
(0.470941884,3)
(0.472945892,3)
(0.4749499,3)
(0.476953908,3)
(0.478957916,3)
(0.480961924,3)
(0.482965932,3)
(0.48496994,3)
(0.486973948,3)
(0.488977956,3)
(0.490981964,3)
(0.492985972,3)
(0.49498998,3)
(0.496993988,3)
(0.498997996,3)
(0.501002004,2)
(0.503006012,2)
(0.50501002,2)
(0.507014028,2)
(0.509018036,2)
(0.511022044,2)
(0.513026052,2)
(0.51503006,2)
(0.517034068,2)
(0.519038076,2)
(0.521042084,2)
(0.523046092,2)
(0.5250501,2)
(0.527054108,2)
(0.529058116,2)
(0.531062124,2)
(0.533066132,2)
(0.53507014,2)
(0.537074148,2)
(0.539078156,2)
(0.541082164,2)
(0.543086172,2)
(0.54509018,2)
(0.547094188,2)
(0.549098196,2)
(0.551102204,2)
(0.553106212,2)
(0.55511022,2)
(0.557114228,2)
(0.559118236,2)
(0.561122244,2)
(0.563126253,2)
(0.565130261,2)
(0.567134269,2)
(0.569138277,2)
(0.571142285,2)
(0.573146293,2)
(0.575150301,2)
(0.577154309,2)
(0.579158317,2)
(0.581162325,2)
(0.583166333,2)
(0.585170341,2)
(0.587174349,2)
(0.589178357,2)
(0.591182365,2)
(0.593186373,2)
(0.595190381,2)
(0.597194389,2)
(0.599198397,2)
(0.601202405,2)
(0.603206413,2)
(0.605210421,2)
(0.607214429,2)
(0.609218437,2)
(0.611222445,2)
(0.613226453,2)
(0.615230461,2)
(0.617234469,2)
(0.619238477,2)
(0.621242485,2)
(0.623246493,2)
(0.625250501,2)
(0.627254509,2)
(0.629258517,2)
(0.631262525,2)
(0.633266533,2)
(0.635270541,2)
(0.637274549,2)
(0.639278557,2)
(0.641282565,2)
(0.643286573,2)
(0.645290581,2)
(0.647294589,2)
(0.649298597,2)
(0.651302605,2)
(0.653306613,2)
(0.655310621,2)
(0.657314629,2)
(0.659318637,2)
(0.661322645,2)
(0.663326653,2)
(0.665330661,2)
(0.667334669,2)
(0.669338677,2)
(0.671342685,2)
(0.673346693,2)
(0.675350701,2)
(0.677354709,2)
(0.679358717,2)
(0.681362725,2)
(0.683366733,2)
(0.685370741,2)
(0.687374749,2)
(0.689378758,2)
(0.691382766,2)
(0.693386774,2)
(0.695390782,2)
(0.69739479,2)
(0.699398798,2)
(0.701402806,2)
(0.703406814,2)
(0.705410822,2)
(0.70741483,2)
(0.709418838,2)
(0.711422846,2)
(0.713426854,2)
(0.715430862,2)
(0.71743487,2)
(0.719438878,2)
(0.721442886,2)
(0.723446894,2)
(0.725450902,2)
(0.72745491,2)
(0.729458918,2)
(0.731462926,2)
(0.733466934,2)
(0.735470942,2)
(0.73747495,2)
(0.739478958,2)
(0.741482966,2)
(0.743486974,2)
(0.745490982,2)
(0.74749499,2)
(0.749498998,2)
(0.751503006,2)
(0.753507014,2)
(0.755511022,2)
(0.75751503,2)
(0.759519038,2)
(0.761523046,2)
(0.763527054,2)
(0.765531062,2)
(0.76753507,2)
(0.769539078,2)
(0.771543086,2)
(0.773547094,2)
(0.775551102,2)
(0.77755511,2)
(0.779559118,2)
(0.781563126,2)
(0.783567134,2)
(0.785571142,2)
(0.78757515,2)
(0.789579158,2)
(0.791583166,2)
(0.793587174,2)
(0.795591182,2)
(0.79759519,2)
(0.799599198,2)
(0.801603206,2)
(0.803607214,2)
(0.805611222,2)
(0.80761523,2)
(0.809619238,2)
(0.811623246,2)
(0.813627255,2)
(0.815631263,2)
(0.817635271,2)
(0.819639279,2)
(0.821643287,2)
(0.823647295,2)
(0.825651303,2)
(0.827655311,2)
(0.829659319,2)
(0.831663327,2)
(0.833667335,2)
(0.835671343,2)
(0.837675351,2)
(0.839679359,2)
(0.841683367,2)
(0.843687375,2)
(0.845691383,2)
(0.847695391,2)
(0.849699399,2)
(0.851703407,2)
(0.853707415,2)
(0.855711423,2)
(0.857715431,2)
(0.859719439,2)
(0.861723447,2)
(0.863727455,2)
(0.865731463,2)
(0.867735471,2)
(0.869739479,2)
(0.871743487,2)
(0.873747495,2)
(0.875751503,2)
(0.877755511,2)
(0.879759519,2)
(0.881763527,2)
(0.883767535,2)
(0.885771543,2)
(0.887775551,2)
(0.889779559,2)
(0.891783567,2)
(0.893787575,2)
(0.895791583,2)
(0.897795591,2)
(0.899799599,2)
(0.901803607,2)
(0.903807615,2)
(0.905811623,2)
(0.907815631,2)
(0.909819639,2)
(0.911823647,2)
(0.913827655,2)
(0.915831663,2)
(0.917835671,2)
(0.919839679,2)
(0.921843687,2)
(0.923847695,2)
(0.925851703,2)
(0.927855711,2)
(0.929859719,2)
(0.931863727,2)
(0.933867735,2)
(0.935871743,2)
(0.937875752,2)
(0.93987976,2)
(0.941883768,2)
(0.943887776,2)
(0.945891784,2)
(0.947895792,2)
(0.9498998,2)
(0.951903808,2)
(0.953907816,2)
(0.955911824,2)
(0.957915832,2)
(0.95991984,2)
(0.961923848,2)
(0.963927856,2)
(0.965931864,2)
(0.967935872,2)
(0.96993988,2)
(0.971943888,2)
(0.973947896,2)
(0.975951904,2)
(0.977955912,2)
(0.97995992,2)
(0.981963928,2)
(0.983967936,2)
(0.985971944,2)
(0.987975952,2)
(0.98997996,2)
(0.991983968,2)
(0.993987976,2)
(0.995991984,2)
(0.997995992,2)
(1,1)
};

\addlegendentry{10 agents}
\addplot[
color=orange,
mark=o,
mark size=0.5pt
]
coordinates {(0,0)
(0.002004008,0)
(0.004008016,0)
(0.006012024,0)
(0.008016032,0)
(0.01002004,0)
(0.012024048,0)
(0.014028056,0)
(0.016032064,0)
(0.018036072,0)
(0.02004008,0)
(0.022044088,0)
(0.024048096,0)
(0.026052104,0)
(0.028056112,0)
(0.03006012,0)
(0.032064128,0)
(0.034068136,0)
(0.036072144,0)
(0.038076152,0)
(0.04008016,0)
(0.042084168,0)
(0.044088176,0)
(0.046092184,0)
(0.048096192,0)
(0.0501002,0)
(0.052104208,0)
(0.054108216,0)
(0.056112224,0)
(0.058116232,0)
(0.06012024,0)
(0.062124248,0)
(0.064128257,0)
(0.066132265,0)
(0.068136273,0)
(0.070140281,0)
(0.072144289,0)
(0.074148297,0)
(0.076152305,0)
(0.078156313,0)
(0.080160321,0)
(0.082164329,0)
(0.084168337,0)
(0.086172345,0)
(0.088176353,0)
(0.090180361,0)
(0.092184369,0)
(0.094188377,0)
(0.096192385,0)
(0.098196393,0)
(0.100200401,0)
(0.102204409,0)
(0.104208417,0)
(0.106212425,0)
(0.108216433,0)
(0.110220441,0)
(0.112224449,10)
(0.114228457,10)
(0.116232465,10)
(0.118236473,9)
(0.120240481,9)
(0.122244489,8)
(0.124248497,8)
(0.126252505,8)
(0.128256513,8)
(0.130260521,8)
(0.132264529,9)
(0.134268537,9)
(0.136272545,9)
(0.138276553,9)
(0.140280561,7)
(0.142284569,8)
(0.144288577,8)
(0.146292585,8)
(0.148296593,8)
(0.150300601,7)
(0.152304609,7)
(0.154308617,7)
(0.156312625,7)
(0.158316633,7)
(0.160320641,7)
(0.162324649,7)
(0.164328657,6)
(0.166332665,6)
(0.168336673,6)
(0.170340681,6)
(0.172344689,6)
(0.174348697,5)
(0.176352705,5)
(0.178356713,5)
(0.180360721,5)
(0.182364729,5)
(0.184368737,5)
(0.186372745,5)
(0.188376754,5)
(0.190380762,9)
(0.19238477,9)
(0.194388778,9)
(0.196392786,8)
(0.198396794,8)
(0.200400802,8)
(0.20240481,8)
(0.204408818,7)
(0.206412826,7)
(0.208416834,7)
(0.210420842,6)
(0.21242485,6)
(0.214428858,6)
(0.216432866,6)
(0.218436874,6)
(0.220440882,6)
(0.22244489,5)
(0.224448898,5)
(0.226452906,5)
(0.228456914,5)
(0.230460922,5)
(0.23246493,5)
(0.234468938,5)
(0.236472946,4)
(0.238476954,4)
(0.240480962,4)
(0.24248497,4)
(0.244488978,4)
(0.246492986,4)
(0.248496994,4)
(0.250501002,5)
(0.25250501,5)
(0.254509018,5)
(0.256513026,5)
(0.258517034,5)
(0.260521042,5)
(0.26252505,5)
(0.264529058,5)
(0.266533066,5)
(0.268537074,5)
(0.270541082,5)
(0.27254509,9)
(0.274549098,9)
(0.276553106,9)
(0.278557114,4)
(0.280561122,4)
(0.28256513,4)
(0.284569138,4)
(0.286573146,4)
(0.288577154,4)
(0.290581162,4)
(0.29258517,4)
(0.294589178,4)
(0.296593186,6)
(0.298597194,6)
(0.300601202,6)
(0.30260521,6)
(0.304609218,6)
(0.306613226,6)
(0.308617234,6)
(0.310621242,6)
(0.312625251,6)
(0.314629259,6)
(0.316633267,6)
(0.318637275,6)
(0.320641283,6)
(0.322645291,6)
(0.324649299,6)
(0.326653307,6)
(0.328657315,5)
(0.330661323,5)
(0.332665331,5)
(0.334669339,4)
(0.336673347,4)
(0.338677355,4)
(0.340681363,4)
(0.342685371,4)
(0.344689379,4)
(0.346693387,4)
(0.348697395,4)
(0.350701403,4)
(0.352705411,4)
(0.354709419,4)
(0.356713427,4)
(0.358717435,4)
(0.360721443,4)
(0.362725451,4)
(0.364729459,4)
(0.366733467,4)
(0.368737475,4)
(0.370741483,4)
(0.372745491,4)
(0.374749499,4)
(0.376753507,4)
(0.378757515,4)
(0.380761523,4)
(0.382765531,4)
(0.384769539,4)
(0.386773547,4)
(0.388777555,4)
(0.390781563,3)
(0.392785571,3)
(0.394789579,3)
(0.396793587,3)
(0.398797595,3)
(0.400801603,3)
(0.402805611,3)
(0.404809619,3)
(0.406813627,3)
(0.408817635,3)
(0.410821643,3)
(0.412825651,3)
(0.414829659,3)
(0.416833667,3)
(0.418837675,3)
(0.420841683,3)
(0.422845691,3)
(0.424849699,3)
(0.426853707,3)
(0.428857715,3)
(0.430861723,3)
(0.432865731,3)
(0.434869739,3)
(0.436873747,3)
(0.438877756,3)
(0.440881764,3)
(0.442885772,3)
(0.44488978,3)
(0.446893788,3)
(0.448897796,3)
(0.450901804,3)
(0.452905812,3)
(0.45490982,3)
(0.456913828,3)
(0.458917836,3)
(0.460921844,3)
(0.462925852,3)
(0.46492986,3)
(0.466933868,3)
(0.468937876,3)
(0.470941884,3)
(0.472945892,3)
(0.4749499,3)
(0.476953908,3)
(0.478957916,3)
(0.480961924,3)
(0.482965932,3)
(0.48496994,3)
(0.486973948,3)
(0.488977956,3)
(0.490981964,3)
(0.492985972,3)
(0.49498998,3)
(0.496993988,3)
(0.498997996,3)
(0.501002004,3)
(0.503006012,3)
(0.50501002,3)
(0.507014028,3)
(0.509018036,3)
(0.511022044,3)
(0.513026052,3)
(0.51503006,3)
(0.517034068,3)
(0.519038076,3)
(0.521042084,3)
(0.523046092,3)
(0.5250501,3)
(0.527054108,3)
(0.529058116,3)
(0.531062124,3)
(0.533066132,3)
(0.53507014,3)
(0.537074148,3)
(0.539078156,3)
(0.541082164,3)
(0.543086172,3)
(0.54509018,3)
(0.547094188,3)
(0.549098196,3)
(0.551102204,3)
(0.553106212,3)
(0.55511022,3)
(0.557114228,2)
(0.559118236,2)
(0.561122244,2)
(0.563126253,2)
(0.565130261,2)
(0.567134269,2)
(0.569138277,2)
(0.571142285,2)
(0.573146293,2)
(0.575150301,2)
(0.577154309,2)
(0.579158317,2)
(0.581162325,2)
(0.583166333,2)
(0.585170341,2)
(0.587174349,2)
(0.589178357,2)
(0.591182365,2)
(0.593186373,2)
(0.595190381,2)
(0.597194389,2)
(0.599198397,2)
(0.601202405,2)
(0.603206413,2)
(0.605210421,2)
(0.607214429,2)
(0.609218437,2)
(0.611222445,2)
(0.613226453,2)
(0.615230461,2)
(0.617234469,2)
(0.619238477,2)
(0.621242485,2)
(0.623246493,2)
(0.625250501,2)
(0.627254509,2)
(0.629258517,2)
(0.631262525,2)
(0.633266533,2)
(0.635270541,2)
(0.637274549,2)
(0.639278557,2)
(0.641282565,2)
(0.643286573,2)
(0.645290581,2)
(0.647294589,2)
(0.649298597,2)
(0.651302605,2)
(0.653306613,2)
(0.655310621,2)
(0.657314629,2)
(0.659318637,2)
(0.661322645,2)
(0.663326653,2)
(0.665330661,2)
(0.667334669,2)
(0.669338677,2)
(0.671342685,2)
(0.673346693,2)
(0.675350701,2)
(0.677354709,2)
(0.679358717,2)
(0.681362725,2)
(0.683366733,2)
(0.685370741,2)
(0.687374749,2)
(0.689378758,2)
(0.691382766,2)
(0.693386774,2)
(0.695390782,2)
(0.69739479,2)
(0.699398798,2)
(0.701402806,2)
(0.703406814,2)
(0.705410822,2)
(0.70741483,2)
(0.709418838,2)
(0.711422846,2)
(0.713426854,2)
(0.715430862,2)
(0.71743487,2)
(0.719438878,2)
(0.721442886,2)
(0.723446894,2)
(0.725450902,2)
(0.72745491,2)
(0.729458918,2)
(0.731462926,2)
(0.733466934,2)
(0.735470942,2)
(0.73747495,2)
(0.739478958,2)
(0.741482966,2)
(0.743486974,2)
(0.745490982,2)
(0.74749499,2)
(0.749498998,2)
(0.751503006,2)
(0.753507014,2)
(0.755511022,2)
(0.75751503,2)
(0.759519038,2)
(0.761523046,2)
(0.763527054,2)
(0.765531062,2)
(0.76753507,2)
(0.769539078,2)
(0.771543086,2)
(0.773547094,2)
(0.775551102,2)
(0.77755511,2)
(0.779559118,2)
(0.781563126,2)
(0.783567134,2)
(0.785571142,2)
(0.78757515,2)
(0.789579158,2)
(0.791583166,2)
(0.793587174,2)
(0.795591182,2)
(0.79759519,2)
(0.799599198,2)
(0.801603206,2)
(0.803607214,2)
(0.805611222,2)
(0.80761523,2)
(0.809619238,2)
(0.811623246,2)
(0.813627255,2)
(0.815631263,2)
(0.817635271,2)
(0.819639279,2)
(0.821643287,2)
(0.823647295,2)
(0.825651303,2)
(0.827655311,2)
(0.829659319,2)
(0.831663327,2)
(0.833667335,2)
(0.835671343,2)
(0.837675351,2)
(0.839679359,2)
(0.841683367,2)
(0.843687375,2)
(0.845691383,2)
(0.847695391,2)
(0.849699399,2)
(0.851703407,2)
(0.853707415,2)
(0.855711423,2)
(0.857715431,2)
(0.859719439,2)
(0.861723447,2)
(0.863727455,2)
(0.865731463,2)
(0.867735471,2)
(0.869739479,2)
(0.871743487,2)
(0.873747495,2)
(0.875751503,2)
(0.877755511,2)
(0.879759519,2)
(0.881763527,2)
(0.883767535,2)
(0.885771543,2)
(0.887775551,2)
(0.889779559,2)
(0.891783567,2)
(0.893787575,2)
(0.895791583,2)
(0.897795591,2)
(0.899799599,2)
(0.901803607,2)
(0.903807615,2)
(0.905811623,2)
(0.907815631,2)
(0.909819639,2)
(0.911823647,2)
(0.913827655,2)
(0.915831663,2)
(0.917835671,2)
(0.919839679,2)
(0.921843687,2)
(0.923847695,2)
(0.925851703,2)
(0.927855711,2)
(0.929859719,2)
(0.931863727,2)
(0.933867735,2)
(0.935871743,2)
(0.937875752,2)
(0.93987976,2)
(0.941883768,2)
(0.943887776,2)
(0.945891784,2)
(0.947895792,2)
(0.9498998,2)
(0.951903808,2)
(0.953907816,2)
(0.955911824,2)
(0.957915832,2)
(0.95991984,2)
(0.961923848,2)
(0.963927856,2)
(0.965931864,2)
(0.967935872,2)
(0.96993988,2)
(0.971943888,2)
(0.973947896,2)
(0.975951904,2)
(0.977955912,2)
(0.97995992,2)
(0.981963928,2)
(0.983967936,2)
(0.985971944,2)
(0.987975952,2)
(0.98997996,2)
(0.991983968,2)
(0.993987976,2)
(0.995991984,2)
(0.997995992,2)
(1,1)
};

\end{axis}
\end{tikzpicture}

%% file: Figures/cluster_vs_openness.tex
\begin{tikzpicture}
\begin{axis}[
xlabel={Confidence Interval ($\epsilon$)},
ylabel={\#(clusters) at equilibrium ($C_{eqm}$)},
xmin=0.0009765625, xmax=1.0,
ymin=0, ymax=205,
xtick={0.0009765625,0.00390625,0.015625,0.0625,0.25,1.0},
ytick={0,50,100,150,200},
xticklabel style={rotate=60},
legend style={at={(1,1)},anchor=north east,font=\footnotesize},
ymajorgrids=true,
grid style=dashed,
xmode=log,
log basis x={2}
]

\addlegendentry{200 agents}
\addplot[
color=red,
mark=diamond,
mark size=0.6pt
]
coordinates {
(0.0009765625,200)
(0.002004008,200)
(0.004008016,200)
(0.006012024,50)
(0.008016032,50)
(0.01002004,40)
(0.012024048,33)
(0.014028056,29)
(0.016032064,25)
(0.018036072,24)
(0.02004008,21)
(0.022044088,19)
(0.024048096,18)
(0.026052104,17)
(0.028056112,16)
(0.03006012,14)
(0.032064128,13)
(0.034068136,13)
(0.036072144,11)
(0.038076152,11)
(0.04008016,11)
(0.042084168,9)
(0.044088176,10)
(0.046092184,9)
(0.048096192,9)
(0.0501002,9)
(0.052104208,7)
(0.054108216,7)
(0.056112224,8)
(0.058116232,7)
(0.06012024,7)
(0.062124248,7)
(0.064128257,7)
(0.066132265,5)
(0.068136273,7)
(0.070140281,6)
(0.072144289,6)
(0.074148297,6)
(0.076152305,5)
(0.078156313,5)
(0.080160321,5)
(0.082164329,5)
(0.084168337,5)
(0.086172345,5)
(0.088176353,5)
(0.090180361,5)
(0.092184369,5)
(0.094188377,3)
(0.096192385,3)
(0.098196393,4)
(0.100200401,5)
(0.102204409,4)
(0.104208417,4)
(0.106212425,4)
(0.108216433,4)
(0.110220441,3)
(0.112224449,3)
(0.114228457,3)
(0.116232465,3)
(0.118236473,3)
(0.120240481,3)
(0.122244489,3)
(0.124248497,3)
(0.126252505,3)
(0.128256513,3)
(0.130260521,3)
(0.132264529,3)
(0.134268537,3)
(0.136272545,3)
(0.138276553,3)
(0.140280561,3)
(0.142284569,3)
(0.144288577,3)
(0.146292585,3)
(0.148296593,3)
(0.150300601,3)
(0.152304609,3)
(0.154308617,3)
(0.156312625,3)
(0.158316633,3)
(0.160320641,3)
(0.162324649,3)
(0.164328657,3)
(0.166332665,3)
(0.168336673,3)
(0.170340681,1)
(0.172344689,3)
(0.174348697,1)
(0.176352705,1)
(0.178356713,3)
(0.180360721,3)
(0.182364729,2)
(0.184368737,3)
(0.186372745,2)
(0.188376754,3)
(0.190380762,1)
(0.19238477,2)
(0.194388778,2)
(0.196392786,2)
(0.198396794,2)
(0.200400802,2)
(0.20240481,2)
(0.204408818,2)
(0.206412826,2)
(0.208416834,2)
(0.210420842,2)
(0.21242485,2)
(0.214428858,1)
(0.216432866,2)
(0.218436874,2)
(0.220440882,2)
(0.22244489,2)
(0.224448898,2)
(0.226452906,1)
(0.228456914,1)
(0.230460922,1)
(0.23246493,1)
(0.234468938,1)
(0.236472946,1)
(0.238476954,1)
(0.240480962,1)
(0.24248497,1)
(0.244488978,1)
(0.246492986,1)
(0.248496994,1)
(0.250501002,1)
(0.25250501,1)
(0.254509018,1)
(0.256513026,1)
(0.258517034,1)
(0.260521042,1)
(0.26252505,1)
(0.264529058,1)
(0.266533066,1)
(0.268537074,1)
(0.270541082,1)
(0.27254509,1)
(0.274549098,1)
(0.276553106,1)
(0.278557114,1)
(0.280561122,1)
(0.28256513,1)
(0.284569138,1)
(0.286573146,1)
(0.288577154,1)
(0.290581162,1)
(0.29258517,1)
(0.294589178,1)
(0.296593186,1)
(0.298597194,1)
(0.300601202,1)
(0.30260521,1)
(0.304609218,1)
(0.306613226,1)
(0.308617234,1)
(0.310621242,1)
(0.312625251,1)
(0.314629259,1)
(0.316633267,1)
(0.318637275,1)
(0.320641283,1)
(0.322645291,1)
(0.324649299,1)
(0.326653307,1)
(0.328657315,1)
(0.330661323,1)
(0.332665331,1)
(0.334669339,1)
(0.336673347,1)
(0.338677355,1)
(0.340681363,1)
(0.342685371,1)
(0.344689379,1)
(0.346693387,1)
(0.348697395,1)
(0.350701403,1)
(0.352705411,1)
(0.354709419,1)
(0.356713427,1)
(0.358717435,1)
(0.360721443,1)
(0.362725451,1)
(0.364729459,1)
(0.366733467,1)
(0.368737475,1)
(0.370741483,1)
(0.372745491,1)
(0.374749499,1)
(0.376753507,1)
(0.378757515,1)
(0.380761523,1)
(0.382765531,1)
(0.384769539,1)
(0.386773547,1)
(0.388777555,1)
(0.390781563,1)
(0.392785571,1)
(0.394789579,1)
(0.396793587,1)
(0.398797595,1)
(0.400801603,1)
(0.402805611,1)
(0.404809619,1)
(0.406813627,1)
(0.408817635,1)
(0.410821643,1)
(0.412825651,1)
(0.414829659,1)
(0.416833667,1)
(0.418837675,1)
(0.420841683,1)
(0.422845691,1)
(0.424849699,1)
(0.426853707,1)
(0.428857715,1)
(0.430861723,1)
(0.432865731,1)
(0.434869739,1)
(0.436873747,1)
(0.438877756,1)
(0.440881764,1)
(0.442885772,1)
(0.44488978,1)
(0.446893788,1)
(0.448897796,1)
(0.450901804,1)
(0.452905812,1)
(0.45490982,1)
(0.456913828,1)
(0.458917836,1)
(0.460921844,1)
(0.462925852,1)
(0.46492986,1)
(0.466933868,1)
(0.468937876,1)
(0.470941884,1)
(0.472945892,1)
(0.4749499,1)
(0.476953908,1)
(0.478957916,1)
(0.480961924,1)
(0.482965932,1)
(0.48496994,1)
(0.486973948,1)
(0.488977956,1)
(0.490981964,1)
(0.492985972,1)
(0.49498998,1)
(0.496993988,1)
(0.498997996,1)
(0.501002004,1)
(0.503006012,1)
(0.50501002,1)
(0.507014028,1)
(0.509018036,1)
(0.511022044,1)
(0.513026052,1)
(0.51503006,1)
(0.517034068,1)
(0.519038076,1)
(0.521042084,1)
(0.523046092,1)
(0.5250501,1)
(0.527054108,1)
(0.529058116,1)
(0.531062124,1)
(0.533066132,1)
(0.53507014,1)
(0.537074148,1)
(0.539078156,1)
(0.541082164,1)
(0.543086172,1)
(0.54509018,1)
(0.547094188,1)
(0.549098196,1)
(0.551102204,1)
(0.553106212,1)
(0.55511022,1)
(0.557114228,1)
(0.559118236,1)
(0.561122244,1)
(0.563126253,1)
(0.565130261,1)
(0.567134269,1)
(0.569138277,1)
(0.571142285,1)
(0.573146293,1)
(0.575150301,1)
(0.577154309,1)
(0.579158317,1)
(0.581162325,1)
(0.583166333,1)
(0.585170341,1)
(0.587174349,1)
(0.589178357,1)
(0.591182365,1)
(0.593186373,1)
(0.595190381,1)
(0.597194389,1)
(0.599198397,1)
(0.601202405,1)
(0.603206413,1)
(0.605210421,1)
(0.607214429,1)
(0.609218437,1)
(0.611222445,1)
(0.613226453,1)
(0.615230461,1)
(0.617234469,1)
(0.619238477,1)
(0.621242485,1)
(0.623246493,1)
(0.625250501,1)
(0.627254509,1)
(0.629258517,1)
(0.631262525,1)
(0.633266533,1)
(0.635270541,1)
(0.637274549,1)
(0.639278557,1)
(0.641282565,1)
(0.643286573,1)
(0.645290581,1)
(0.647294589,1)
(0.649298597,1)
(0.651302605,1)
(0.653306613,1)
(0.655310621,1)
(0.657314629,1)
(0.659318637,1)
(0.661322645,1)
(0.663326653,1)
(0.665330661,1)
(0.667334669,1)
(0.669338677,1)
(0.671342685,1)
(0.673346693,1)
(0.675350701,1)
(0.677354709,1)
(0.679358717,1)
(0.681362725,1)
(0.683366733,1)
(0.685370741,1)
(0.687374749,1)
(0.689378758,1)
(0.691382766,1)
(0.693386774,1)
(0.695390782,1)
(0.69739479,1)
(0.699398798,1)
(0.701402806,1)
(0.703406814,1)
(0.705410822,1)
(0.70741483,1)
(0.709418838,1)
(0.711422846,1)
(0.713426854,1)
(0.715430862,1)
(0.71743487,1)
(0.719438878,1)
(0.721442886,1)
(0.723446894,1)
(0.725450902,1)
(0.72745491,1)
(0.729458918,1)
(0.731462926,1)
(0.733466934,1)
(0.735470942,1)
(0.73747495,1)
(0.739478958,1)
(0.741482966,1)
(0.743486974,1)
(0.745490982,1)
(0.74749499,1)
(0.749498998,1)
(0.751503006,1)
(0.753507014,1)
(0.755511022,1)
(0.75751503,1)
(0.759519038,1)
(0.761523046,1)
(0.763527054,1)
(0.765531062,1)
(0.76753507,1)
(0.769539078,1)
(0.771543086,1)
(0.773547094,1)
(0.775551102,1)
(0.77755511,1)
(0.779559118,1)
(0.781563126,1)
(0.783567134,1)
(0.785571142,1)
(0.78757515,1)
(0.789579158,1)
(0.791583166,1)
(0.793587174,1)
(0.795591182,1)
(0.79759519,1)
(0.799599198,1)
(0.801603206,1)
(0.803607214,1)
(0.805611222,1)
(0.80761523,1)
(0.809619238,1)
(0.811623246,1)
(0.813627255,1)
(0.815631263,1)
(0.817635271,1)
(0.819639279,1)
(0.821643287,1)
(0.823647295,1)
(0.825651303,1)
(0.827655311,1)
(0.829659319,1)
(0.831663327,1)
(0.833667335,1)
(0.835671343,1)
(0.837675351,1)
(0.839679359,1)
(0.841683367,1)
(0.843687375,1)
(0.845691383,1)
(0.847695391,1)
(0.849699399,1)
(0.851703407,1)
(0.853707415,1)
(0.855711423,1)
(0.857715431,1)
(0.859719439,1)
(0.861723447,1)
(0.863727455,1)
(0.865731463,1)
(0.867735471,1)
(0.869739479,1)
(0.871743487,1)
(0.873747495,1)
(0.875751503,1)
(0.877755511,1)
(0.879759519,1)
(0.881763527,1)
(0.883767535,1)
(0.885771543,1)
(0.887775551,1)
(0.889779559,1)
(0.891783567,1)
(0.893787575,1)
(0.895791583,1)
(0.897795591,1)
(0.899799599,1)
(0.901803607,1)
(0.903807615,1)
(0.905811623,1)
(0.907815631,1)
(0.909819639,1)
(0.911823647,1)
(0.913827655,1)
(0.915831663,1)
(0.917835671,1)
(0.919839679,1)
(0.921843687,1)
(0.923847695,1)
(0.925851703,1)
(0.927855711,1)
(0.929859719,1)
(0.931863727,1)
(0.933867735,1)
(0.935871743,1)
(0.937875752,1)
(0.93987976,1)
(0.941883768,1)
(0.943887776,1)
(0.945891784,1)
(0.947895792,1)
(0.9498998,1)
(0.951903808,1)
(0.953907816,1)
(0.955911824,1)
(0.957915832,1)
(0.95991984,1)
(0.961923848,1)
(0.963927856,1)
(0.965931864,1)
(0.967935872,1)
(0.96993988,1)
(0.971943888,1)
(0.973947896,1)
(0.975951904,1)
(0.977955912,1)
(0.97995992,1)
(0.981963928,1)
(0.983967936,1)
(0.985971944,1)
(0.987975952,1)
(0.98997996,1)
(0.991983968,1)
(0.993987976,1)
(0.995991984,1)
(0.997995992,1)
(1,1)
};

\addlegendentry{100 agents}
\addplot[
color=LimeGreen,
mark=triangle,
mark size=0.5pt
]
coordinates {
(0.0009765625,100)
(0.002004008,100)
(0.004008016,100)
(0.006012024,100)
(0.008016032,100)
(0.01002004,100)
(0.012024048,29)
(0.014028056,25)
(0.016032064,25)
(0.018036072,20)
(0.02004008,20)
(0.022044088,17)
(0.024048096,17)
(0.026052104,15)
(0.028056112,15)
(0.03006012,14)
(0.032064128,13)
(0.034068136,12)
(0.036072144,12)
(0.038076152,11)
(0.04008016,11)
(0.042084168,10)
(0.044088176,9)
(0.046092184,9)
(0.048096192,9)
(0.0501002,9)
(0.052104208,7)
(0.054108216,8)
(0.056112224,7)
(0.058116232,7)
(0.06012024,7)
(0.062124248,7)
(0.064128257,7)
(0.066132265,5)
(0.068136273,7)
(0.070140281,6)
(0.072144289,6)
(0.074148297,6)
(0.076152305,5)
(0.078156313,6)
(0.080160321,5)
(0.082164329,5)
(0.084168337,5)
(0.086172345,5)
(0.088176353,5)
(0.090180361,5)
(0.092184369,5)
(0.094188377,5)
(0.096192385,5)
(0.098196393,5)
(0.100200401,5)
(0.102204409,4)
(0.104208417,4)
(0.106212425,3)
(0.108216433,4)
(0.110220441,4)
(0.112224449,3)
(0.114228457,4)
(0.116232465,3)
(0.118236473,3)
(0.120240481,3)
(0.122244489,3)
(0.124248497,3)
(0.126252505,3)
(0.128256513,3)
(0.130260521,3)
(0.132264529,3)
(0.134268537,3)
(0.136272545,3)
(0.138276553,3)
(0.140280561,3)
(0.142284569,3)
(0.144288577,3)
(0.146292585,3)
(0.148296593,3)
(0.150300601,3)
(0.152304609,3)
(0.154308617,3)
(0.156312625,3)
(0.158316633,3)
(0.160320641,3)
(0.162324649,3)
(0.164328657,3)
(0.166332665,3)
(0.168336673,3)
(0.170340681,3)
(0.172344689,3)
(0.174348697,1)
(0.176352705,3)
(0.178356713,3)
(0.180360721,3)
(0.182364729,3)
(0.184368737,3)
(0.186372745,3)
(0.188376754,1)
(0.190380762,1)
(0.19238477,1)
(0.194388778,2)
(0.196392786,2)
(0.198396794,2)
(0.200400802,2)
(0.20240481,2)
(0.204408818,2)
(0.206412826,2)
(0.208416834,2)
(0.210420842,2)
(0.21242485,2)
(0.214428858,2)
(0.216432866,2)
(0.218436874,2)
(0.220440882,2)
(0.22244489,2)
(0.224448898,2)
(0.226452906,2)
(0.228456914,1)
(0.230460922,1)
(0.23246493,1)
(0.234468938,1)
(0.236472946,1)
(0.238476954,1)
(0.240480962,1)
(0.24248497,1)
(0.244488978,1)
(0.246492986,1)
(0.248496994,1)
(0.250501002,1)
(0.25250501,1)
(0.254509018,1)
(0.256513026,1)
(0.258517034,1)
(0.260521042,1)
(0.26252505,1)
(0.264529058,1)
(0.266533066,1)
(0.268537074,1)
(0.270541082,1)
(0.27254509,1)
(0.274549098,1)
(0.276553106,1)
(0.278557114,1)
(0.280561122,1)
(0.28256513,1)
(0.284569138,1)
(0.286573146,1)
(0.288577154,1)
(0.290581162,1)
(0.29258517,1)
(0.294589178,1)
(0.296593186,1)
(0.298597194,1)
(0.300601202,1)
(0.30260521,1)
(0.304609218,1)
(0.306613226,1)
(0.308617234,1)
(0.310621242,1)
(0.312625251,1)
(0.314629259,1)
(0.316633267,1)
(0.318637275,1)
(0.320641283,1)
(0.322645291,1)
(0.324649299,1)
(0.326653307,1)
(0.328657315,1)
(0.330661323,1)
(0.332665331,1)
(0.334669339,1)
(0.336673347,1)
(0.338677355,1)
(0.340681363,1)
(0.342685371,1)
(0.344689379,1)
(0.346693387,1)
(0.348697395,1)
(0.350701403,1)
(0.352705411,1)
(0.354709419,1)
(0.356713427,1)
(0.358717435,1)
(0.360721443,1)
(0.362725451,1)
(0.364729459,1)
(0.366733467,1)
(0.368737475,1)
(0.370741483,1)
(0.372745491,1)
(0.374749499,1)
(0.376753507,1)
(0.378757515,1)
(0.380761523,1)
(0.382765531,1)
(0.384769539,1)
(0.386773547,1)
(0.388777555,1)
(0.390781563,1)
(0.392785571,1)
(0.394789579,1)
(0.396793587,1)
(0.398797595,1)
(0.400801603,1)
(0.402805611,1)
(0.404809619,1)
(0.406813627,1)
(0.408817635,1)
(0.410821643,1)
(0.412825651,1)
(0.414829659,1)
(0.416833667,1)
(0.418837675,1)
(0.420841683,1)
(0.422845691,1)
(0.424849699,1)
(0.426853707,1)
(0.428857715,1)
(0.430861723,1)
(0.432865731,1)
(0.434869739,1)
(0.436873747,1)
(0.438877756,1)
(0.440881764,1)
(0.442885772,1)
(0.44488978,1)
(0.446893788,1)
(0.448897796,1)
(0.450901804,1)
(0.452905812,1)
(0.45490982,1)
(0.456913828,1)
(0.458917836,1)
(0.460921844,1)
(0.462925852,1)
(0.46492986,1)
(0.466933868,1)
(0.468937876,1)
(0.470941884,1)
(0.472945892,1)
(0.4749499,1)
(0.476953908,1)
(0.478957916,1)
(0.480961924,1)
(0.482965932,1)
(0.48496994,1)
(0.486973948,1)
(0.488977956,1)
(0.490981964,1)
(0.492985972,1)
(0.49498998,1)
(0.496993988,1)
(0.498997996,1)
(0.501002004,1)
(0.503006012,1)
(0.50501002,1)
(0.507014028,1)
(0.509018036,1)
(0.511022044,1)
(0.513026052,1)
(0.51503006,1)
(0.517034068,1)
(0.519038076,1)
(0.521042084,1)
(0.523046092,1)
(0.5250501,1)
(0.527054108,1)
(0.529058116,1)
(0.531062124,1)
(0.533066132,1)
(0.53507014,1)
(0.537074148,1)
(0.539078156,1)
(0.541082164,1)
(0.543086172,1)
(0.54509018,1)
(0.547094188,1)
(0.549098196,1)
(0.551102204,1)
(0.553106212,1)
(0.55511022,1)
(0.557114228,1)
(0.559118236,1)
(0.561122244,1)
(0.563126253,1)
(0.565130261,1)
(0.567134269,1)
(0.569138277,1)
(0.571142285,1)
(0.573146293,1)
(0.575150301,1)
(0.577154309,1)
(0.579158317,1)
(0.581162325,1)
(0.583166333,1)
(0.585170341,1)
(0.587174349,1)
(0.589178357,1)
(0.591182365,1)
(0.593186373,1)
(0.595190381,1)
(0.597194389,1)
(0.599198397,1)
(0.601202405,1)
(0.603206413,1)
(0.605210421,1)
(0.607214429,1)
(0.609218437,1)
(0.611222445,1)
(0.613226453,1)
(0.615230461,1)
(0.617234469,1)
(0.619238477,1)
(0.621242485,1)
(0.623246493,1)
(0.625250501,1)
(0.627254509,1)
(0.629258517,1)
(0.631262525,1)
(0.633266533,1)
(0.635270541,1)
(0.637274549,1)
(0.639278557,1)
(0.641282565,1)
(0.643286573,1)
(0.645290581,1)
(0.647294589,1)
(0.649298597,1)
(0.651302605,1)
(0.653306613,1)
(0.655310621,1)
(0.657314629,1)
(0.659318637,1)
(0.661322645,1)
(0.663326653,1)
(0.665330661,1)
(0.667334669,1)
(0.669338677,1)
(0.671342685,1)
(0.673346693,1)
(0.675350701,1)
(0.677354709,1)
(0.679358717,1)
(0.681362725,1)
(0.683366733,1)
(0.685370741,1)
(0.687374749,1)
(0.689378758,1)
(0.691382766,1)
(0.693386774,1)
(0.695390782,1)
(0.69739479,1)
(0.699398798,1)
(0.701402806,1)
(0.703406814,1)
(0.705410822,1)
(0.70741483,1)
(0.709418838,1)
(0.711422846,1)
(0.713426854,1)
(0.715430862,1)
(0.71743487,1)
(0.719438878,1)
(0.721442886,1)
(0.723446894,1)
(0.725450902,1)
(0.72745491,1)
(0.729458918,1)
(0.731462926,1)
(0.733466934,1)
(0.735470942,1)
(0.73747495,1)
(0.739478958,1)
(0.741482966,1)
(0.743486974,1)
(0.745490982,1)
(0.74749499,1)
(0.749498998,1)
(0.751503006,1)
(0.753507014,1)
(0.755511022,1)
(0.75751503,1)
(0.759519038,1)
(0.761523046,1)
(0.763527054,1)
(0.765531062,1)
(0.76753507,1)
(0.769539078,1)
(0.771543086,1)
(0.773547094,1)
(0.775551102,1)
(0.77755511,1)
(0.779559118,1)
(0.781563126,1)
(0.783567134,1)
(0.785571142,1)
(0.78757515,1)
(0.789579158,1)
(0.791583166,1)
(0.793587174,1)
(0.795591182,1)
(0.79759519,1)
(0.799599198,1)
(0.801603206,1)
(0.803607214,1)
(0.805611222,1)
(0.80761523,1)
(0.809619238,1)
(0.811623246,1)
(0.813627255,1)
(0.815631263,1)
(0.817635271,1)
(0.819639279,1)
(0.821643287,1)
(0.823647295,1)
(0.825651303,1)
(0.827655311,1)
(0.829659319,1)
(0.831663327,1)
(0.833667335,1)
(0.835671343,1)
(0.837675351,1)
(0.839679359,1)
(0.841683367,1)
(0.843687375,1)
(0.845691383,1)
(0.847695391,1)
(0.849699399,1)
(0.851703407,1)
(0.853707415,1)
(0.855711423,1)
(0.857715431,1)
(0.859719439,1)
(0.861723447,1)
(0.863727455,1)
(0.865731463,1)
(0.867735471,1)
(0.869739479,1)
(0.871743487,1)
(0.873747495,1)
(0.875751503,1)
(0.877755511,1)
(0.879759519,1)
(0.881763527,1)
(0.883767535,1)
(0.885771543,1)
(0.887775551,1)
(0.889779559,1)
(0.891783567,1)
(0.893787575,1)
(0.895791583,1)
(0.897795591,1)
(0.899799599,1)
(0.901803607,1)
(0.903807615,1)
(0.905811623,1)
(0.907815631,1)
(0.909819639,1)
(0.911823647,1)
(0.913827655,1)
(0.915831663,1)
(0.917835671,1)
(0.919839679,1)
(0.921843687,1)
(0.923847695,1)
(0.925851703,1)
(0.927855711,1)
(0.929859719,1)
(0.931863727,1)
(0.933867735,1)
(0.935871743,1)
(0.937875752,1)
(0.93987976,1)
(0.941883768,1)
(0.943887776,1)
(0.945891784,1)
(0.947895792,1)
(0.9498998,1)
(0.951903808,1)
(0.953907816,1)
(0.955911824,1)
(0.957915832,1)
(0.95991984,1)
(0.961923848,1)
(0.963927856,1)
(0.965931864,1)
(0.967935872,1)
(0.96993988,1)
(0.971943888,1)
(0.973947896,1)
(0.975951904,1)
(0.977955912,1)
(0.97995992,1)
(0.981963928,1)
(0.983967936,1)
(0.985971944,1)
(0.987975952,1)
(0.98997996,1)
(0.991983968,1)
(0.993987976,1)
(0.995991984,1)
(0.997995992,1)
(1,1)
};

\addlegendentry{50 agents}
\addplot[
color=ProcessBlue,
mark=square,
mark size=0.5pt
]
coordinates {
(0.0009765625,50)
(0.002004008,50)
(0.004008016,50)
(0.006012024,50)
(0.008016032,50)
(0.01002004,50)
(0.012024048,50)
(0.014028056,50)
(0.016032064,50)
(0.018036072,50)
(0.02004008,50)
(0.022044088,17)
(0.024048096,17)
(0.026052104,14)
(0.028056112,13)
(0.03006012,13)
(0.032064128,13)
(0.034068136,13)
(0.036072144,10)
(0.038076152,10)
(0.04008016,10)
(0.042084168,10)
(0.044088176,9)
(0.046092184,9)
(0.048096192,9)
(0.0501002,9)
(0.052104208,8)
(0.054108216,7)
(0.056112224,7)
(0.058116232,7)
(0.06012024,7)
(0.062124248,7)
(0.064128257,7)
(0.066132265,5)
(0.068136273,6)
(0.070140281,6)
(0.072144289,6)
(0.074148297,6)
(0.076152305,6)
(0.078156313,6)
(0.080160321,5)
(0.082164329,5)
(0.084168337,5)
(0.086172345,5)
(0.088176353,5)
(0.090180361,5)
(0.092184369,5)
(0.094188377,5)
(0.096192385,5)
(0.098196393,5)
(0.100200401,5)
(0.102204409,4)
(0.104208417,4)
(0.106212425,4)
(0.108216433,4)
(0.110220441,4)
(0.112224449,4)
(0.114228457,4)
(0.116232465,4)
(0.118236473,4)
(0.120240481,4)
(0.122244489,3)
(0.124248497,3)
(0.126252505,3)
(0.128256513,3)
(0.130260521,3)
(0.132264529,3)
(0.134268537,3)
(0.136272545,3)
(0.138276553,3)
(0.140280561,3)
(0.142284569,3)
(0.144288577,3)
(0.146292585,3)
(0.148296593,3)
(0.150300601,3)
(0.152304609,3)
(0.154308617,3)
(0.156312625,3)
(0.158316633,3)
(0.160320641,3)
(0.162324649,3)
(0.164328657,3)
(0.166332665,3)
(0.168336673,3)
(0.170340681,3)
(0.172344689,3)
(0.174348697,3)
(0.176352705,3)
(0.178356713,1)
(0.180360721,3)
(0.182364729,3)
(0.184368737,3)
(0.186372745,2)
(0.188376754,2)
(0.190380762,1)
(0.19238477,2)
(0.194388778,2)
(0.196392786,2)
(0.198396794,2)
(0.200400802,2)
(0.20240481,2)
(0.204408818,2)
(0.206412826,2)
(0.208416834,2)
(0.210420842,1)
(0.21242485,1)
(0.214428858,2)
(0.216432866,2)
(0.218436874,2)
(0.220440882,2)
(0.22244489,2)
(0.224448898,2)
(0.226452906,2)
(0.228456914,2)
(0.230460922,2)
(0.23246493,1)
(0.234468938,1)
(0.236472946,2)
(0.238476954,1)
(0.240480962,1)
(0.24248497,1)
(0.244488978,1)
(0.246492986,1)
(0.248496994,1)
(0.250501002,1)
(0.25250501,1)
(0.254509018,1)
(0.256513026,1)
(0.258517034,1)
(0.260521042,1)
(0.26252505,1)
(0.264529058,1)
(0.266533066,1)
(0.268537074,1)
(0.270541082,1)
(0.27254509,1)
(0.274549098,1)
(0.276553106,1)
(0.278557114,1)
(0.280561122,1)
(0.28256513,1)
(0.284569138,1)
(0.286573146,1)
(0.288577154,1)
(0.290581162,1)
(0.29258517,1)
(0.294589178,1)
(0.296593186,1)
(0.298597194,1)
(0.300601202,1)
(0.30260521,1)
(0.304609218,1)
(0.306613226,1)
(0.308617234,1)
(0.310621242,1)
(0.312625251,1)
(0.314629259,1)
(0.316633267,1)
(0.318637275,1)
(0.320641283,1)
(0.322645291,1)
(0.324649299,1)
(0.326653307,1)
(0.328657315,1)
(0.330661323,1)
(0.332665331,1)
(0.334669339,1)
(0.336673347,1)
(0.338677355,1)
(0.340681363,1)
(0.342685371,1)
(0.344689379,1)
(0.346693387,1)
(0.348697395,1)
(0.350701403,1)
(0.352705411,1)
(0.354709419,1)
(0.356713427,1)
(0.358717435,1)
(0.360721443,1)
(0.362725451,1)
(0.364729459,1)
(0.366733467,1)
(0.368737475,1)
(0.370741483,1)
(0.372745491,1)
(0.374749499,1)
(0.376753507,1)
(0.378757515,1)
(0.380761523,1)
(0.382765531,1)
(0.384769539,1)
(0.386773547,1)
(0.388777555,1)
(0.390781563,1)
(0.392785571,1)
(0.394789579,1)
(0.396793587,1)
(0.398797595,1)
(0.400801603,1)
(0.402805611,1)
(0.404809619,1)
(0.406813627,1)
(0.408817635,1)
(0.410821643,1)
(0.412825651,1)
(0.414829659,1)
(0.416833667,1)
(0.418837675,1)
(0.420841683,1)
(0.422845691,1)
(0.424849699,1)
(0.426853707,1)
(0.428857715,1)
(0.430861723,1)
(0.432865731,1)
(0.434869739,1)
(0.436873747,1)
(0.438877756,1)
(0.440881764,1)
(0.442885772,1)
(0.44488978,1)
(0.446893788,1)
(0.448897796,1)
(0.450901804,1)
(0.452905812,1)
(0.45490982,1)
(0.456913828,1)
(0.458917836,1)
(0.460921844,1)
(0.462925852,1)
(0.46492986,1)
(0.466933868,1)
(0.468937876,1)
(0.470941884,1)
(0.472945892,1)
(0.4749499,1)
(0.476953908,1)
(0.478957916,1)
(0.480961924,1)
(0.482965932,1)
(0.48496994,1)
(0.486973948,1)
(0.488977956,1)
(0.490981964,1)
(0.492985972,1)
(0.49498998,1)
(0.496993988,1)
(0.498997996,1)
(0.501002004,1)
(0.503006012,1)
(0.50501002,1)
(0.507014028,1)
(0.509018036,1)
(0.511022044,1)
(0.513026052,1)
(0.51503006,1)
(0.517034068,1)
(0.519038076,1)
(0.521042084,1)
(0.523046092,1)
(0.5250501,1)
(0.527054108,1)
(0.529058116,1)
(0.531062124,1)
(0.533066132,1)
(0.53507014,1)
(0.537074148,1)
(0.539078156,1)
(0.541082164,1)
(0.543086172,1)
(0.54509018,1)
(0.547094188,1)
(0.549098196,1)
(0.551102204,1)
(0.553106212,1)
(0.55511022,1)
(0.557114228,1)
(0.559118236,1)
(0.561122244,1)
(0.563126253,1)
(0.565130261,1)
(0.567134269,1)
(0.569138277,1)
(0.571142285,1)
(0.573146293,1)
(0.575150301,1)
(0.577154309,1)
(0.579158317,1)
(0.581162325,1)
(0.583166333,1)
(0.585170341,1)
(0.587174349,1)
(0.589178357,1)
(0.591182365,1)
(0.593186373,1)
(0.595190381,1)
(0.597194389,1)
(0.599198397,1)
(0.601202405,1)
(0.603206413,1)
(0.605210421,1)
(0.607214429,1)
(0.609218437,1)
(0.611222445,1)
(0.613226453,1)
(0.615230461,1)
(0.617234469,1)
(0.619238477,1)
(0.621242485,1)
(0.623246493,1)
(0.625250501,1)
(0.627254509,1)
(0.629258517,1)
(0.631262525,1)
(0.633266533,1)
(0.635270541,1)
(0.637274549,1)
(0.639278557,1)
(0.641282565,1)
(0.643286573,1)
(0.645290581,1)
(0.647294589,1)
(0.649298597,1)
(0.651302605,1)
(0.653306613,1)
(0.655310621,1)
(0.657314629,1)
(0.659318637,1)
(0.661322645,1)
(0.663326653,1)
(0.665330661,1)
(0.667334669,1)
(0.669338677,1)
(0.671342685,1)
(0.673346693,1)
(0.675350701,1)
(0.677354709,1)
(0.679358717,1)
(0.681362725,1)
(0.683366733,1)
(0.685370741,1)
(0.687374749,1)
(0.689378758,1)
(0.691382766,1)
(0.693386774,1)
(0.695390782,1)
(0.69739479,1)
(0.699398798,1)
(0.701402806,1)
(0.703406814,1)
(0.705410822,1)
(0.70741483,1)
(0.709418838,1)
(0.711422846,1)
(0.713426854,1)
(0.715430862,1)
(0.71743487,1)
(0.719438878,1)
(0.721442886,1)
(0.723446894,1)
(0.725450902,1)
(0.72745491,1)
(0.729458918,1)
(0.731462926,1)
(0.733466934,1)
(0.735470942,1)
(0.73747495,1)
(0.739478958,1)
(0.741482966,1)
(0.743486974,1)
(0.745490982,1)
(0.74749499,1)
(0.749498998,1)
(0.751503006,1)
(0.753507014,1)
(0.755511022,1)
(0.75751503,1)
(0.759519038,1)
(0.761523046,1)
(0.763527054,1)
(0.765531062,1)
(0.76753507,1)
(0.769539078,1)
(0.771543086,1)
(0.773547094,1)
(0.775551102,1)
(0.77755511,1)
(0.779559118,1)
(0.781563126,1)
(0.783567134,1)
(0.785571142,1)
(0.78757515,1)
(0.789579158,1)
(0.791583166,1)
(0.793587174,1)
(0.795591182,1)
(0.79759519,1)
(0.799599198,1)
(0.801603206,1)
(0.803607214,1)
(0.805611222,1)
(0.80761523,1)
(0.809619238,1)
(0.811623246,1)
(0.813627255,1)
(0.815631263,1)
(0.817635271,1)
(0.819639279,1)
(0.821643287,1)
(0.823647295,1)
(0.825651303,1)
(0.827655311,1)
(0.829659319,1)
(0.831663327,1)
(0.833667335,1)
(0.835671343,1)
(0.837675351,1)
(0.839679359,1)
(0.841683367,1)
(0.843687375,1)
(0.845691383,1)
(0.847695391,1)
(0.849699399,1)
(0.851703407,1)
(0.853707415,1)
(0.855711423,1)
(0.857715431,1)
(0.859719439,1)
(0.861723447,1)
(0.863727455,1)
(0.865731463,1)
(0.867735471,1)
(0.869739479,1)
(0.871743487,1)
(0.873747495,1)
(0.875751503,1)
(0.877755511,1)
(0.879759519,1)
(0.881763527,1)
(0.883767535,1)
(0.885771543,1)
(0.887775551,1)
(0.889779559,1)
(0.891783567,1)
(0.893787575,1)
(0.895791583,1)
(0.897795591,1)
(0.899799599,1)
(0.901803607,1)
(0.903807615,1)
(0.905811623,1)
(0.907815631,1)
(0.909819639,1)
(0.911823647,1)
(0.913827655,1)
(0.915831663,1)
(0.917835671,1)
(0.919839679,1)
(0.921843687,1)
(0.923847695,1)
(0.925851703,1)
(0.927855711,1)
(0.929859719,1)
(0.931863727,1)
(0.933867735,1)
(0.935871743,1)
(0.937875752,1)
(0.93987976,1)
(0.941883768,1)
(0.943887776,1)
(0.945891784,1)
(0.947895792,1)
(0.9498998,1)
(0.951903808,1)
(0.953907816,1)
(0.955911824,1)
(0.957915832,1)
(0.95991984,1)
(0.961923848,1)
(0.963927856,1)
(0.965931864,1)
(0.967935872,1)
(0.96993988,1)
(0.971943888,1)
(0.973947896,1)
(0.975951904,1)
(0.977955912,1)
(0.97995992,1)
(0.981963928,1)
(0.983967936,1)
(0.985971944,1)
(0.987975952,1)
(0.98997996,1)
(0.991983968,1)
(0.993987976,1)
(0.995991984,1)
(0.997995992,1)
(1,1)
};

\addlegendentry{25 agents}
\addplot[
color=blue,
mark=+,
mark size=0.8pt
]
coordinates {
(0.0009765625,25)
(0.002004008,25)
(0.004008016,25)
(0.006012024,25)
(0.008016032,25)
(0.01002004,25)
(0.012024048,25)
(0.014028056,25)
(0.016032064,25)
(0.018036072,25)
(0.02004008,25)
(0.022044088,25)
(0.024048096,25)
(0.026052104,25)
(0.028056112,25)
(0.03006012,25)
(0.032064128,25)
(0.034068136,25)
(0.036072144,25)
(0.038076152,25)
(0.04008016,25)
(0.042084168,9)
(0.044088176,9)
(0.046092184,9)
(0.048096192,9)
(0.0501002,7)
(0.052104208,9)
(0.054108216,7)
(0.056112224,7)
(0.058116232,7)
(0.06012024,7)
(0.062124248,7)
(0.064128257,7)
(0.066132265,7)
(0.068136273,7)
(0.070140281,7)
(0.072144289,5)
(0.074148297,5)
(0.076152305,5)
(0.078156313,5)
(0.080160321,5)
(0.082164329,5)
(0.084168337,5)
(0.086172345,5)
(0.088176353,5)
(0.090180361,5)
(0.092184369,5)
(0.094188377,3)
(0.096192385,3)
(0.098196393,3)
(0.100200401,3)
(0.102204409,3)
(0.104208417,3)
(0.106212425,3)
(0.108216433,3)
(0.110220441,3)
(0.112224449,3)
(0.114228457,3)
(0.116232465,3)
(0.118236473,3)
(0.120240481,3)
(0.122244489,3)
(0.124248497,3)
(0.126252505,3)
(0.128256513,3)
(0.130260521,3)
(0.132264529,3)
(0.134268537,3)
(0.136272545,3)
(0.138276553,3)
(0.140280561,3)
(0.142284569,3)
(0.144288577,3)
(0.146292585,3)
(0.148296593,3)
(0.150300601,3)
(0.152304609,3)
(0.154308617,3)
(0.156312625,3)
(0.158316633,3)
(0.160320641,3)
(0.162324649,3)
(0.164328657,3)
(0.166332665,3)
(0.168336673,3)
(0.170340681,3)
(0.172344689,3)
(0.174348697,3)
(0.176352705,3)
(0.178356713,1)
(0.180360721,3)
(0.182364729,3)
(0.184368737,1)
(0.186372745,3)
(0.188376754,1)
(0.190380762,1)
(0.19238477,1)
(0.194388778,1)
(0.196392786,1)
(0.198396794,1)
(0.200400802,1)
(0.20240481,1)
(0.204408818,1)
(0.206412826,1)
(0.208416834,1)
(0.210420842,1)
(0.21242485,1)
(0.214428858,1)
(0.216432866,1)
(0.218436874,1)
(0.220440882,1)
(0.22244489,1)
(0.224448898,1)
(0.226452906,1)
(0.228456914,1)
(0.230460922,1)
(0.23246493,1)
(0.234468938,1)
(0.236472946,1)
(0.238476954,1)
(0.240480962,1)
(0.24248497,1)
(0.244488978,1)
(0.246492986,1)
(0.248496994,1)
(0.250501002,1)
(0.25250501,1)
(0.254509018,1)
(0.256513026,1)
(0.258517034,1)
(0.260521042,1)
(0.26252505,1)
(0.264529058,1)
(0.266533066,1)
(0.268537074,1)
(0.270541082,1)
(0.27254509,1)
(0.274549098,1)
(0.276553106,1)
(0.278557114,1)
(0.280561122,1)
(0.28256513,1)
(0.284569138,1)
(0.286573146,1)
(0.288577154,1)
(0.290581162,1)
(0.29258517,1)
(0.294589178,1)
(0.296593186,1)
(0.298597194,1)
(0.300601202,1)
(0.30260521,1)
(0.304609218,1)
(0.306613226,1)
(0.308617234,1)
(0.310621242,1)
(0.312625251,1)
(0.314629259,1)
(0.316633267,1)
(0.318637275,1)
(0.320641283,1)
(0.322645291,1)
(0.324649299,1)
(0.326653307,1)
(0.328657315,1)
(0.330661323,1)
(0.332665331,1)
(0.334669339,1)
(0.336673347,1)
(0.338677355,1)
(0.340681363,1)
(0.342685371,1)
(0.344689379,1)
(0.346693387,1)
(0.348697395,1)
(0.350701403,1)
(0.352705411,1)
(0.354709419,1)
(0.356713427,1)
(0.358717435,1)
(0.360721443,1)
(0.362725451,1)
(0.364729459,1)
(0.366733467,1)
(0.368737475,1)
(0.370741483,1)
(0.372745491,1)
(0.374749499,1)
(0.376753507,1)
(0.378757515,1)
(0.380761523,1)
(0.382765531,1)
(0.384769539,1)
(0.386773547,1)
(0.388777555,1)
(0.390781563,1)
(0.392785571,1)
(0.394789579,1)
(0.396793587,1)
(0.398797595,1)
(0.400801603,1)
(0.402805611,1)
(0.404809619,1)
(0.406813627,1)
(0.408817635,1)
(0.410821643,1)
(0.412825651,1)
(0.414829659,1)
(0.416833667,1)
(0.418837675,1)
(0.420841683,1)
(0.422845691,1)
(0.424849699,1)
(0.426853707,1)
(0.428857715,1)
(0.430861723,1)
(0.432865731,1)
(0.434869739,1)
(0.436873747,1)
(0.438877756,1)
(0.440881764,1)
(0.442885772,1)
(0.44488978,1)
(0.446893788,1)
(0.448897796,1)
(0.450901804,1)
(0.452905812,1)
(0.45490982,1)
(0.456913828,1)
(0.458917836,1)
(0.460921844,1)
(0.462925852,1)
(0.46492986,1)
(0.466933868,1)
(0.468937876,1)
(0.470941884,1)
(0.472945892,1)
(0.4749499,1)
(0.476953908,1)
(0.478957916,1)
(0.480961924,1)
(0.482965932,1)
(0.48496994,1)
(0.486973948,1)
(0.488977956,1)
(0.490981964,1)
(0.492985972,1)
(0.49498998,1)
(0.496993988,1)
(0.498997996,1)
(0.501002004,1)
(0.503006012,1)
(0.50501002,1)
(0.507014028,1)
(0.509018036,1)
(0.511022044,1)
(0.513026052,1)
(0.51503006,1)
(0.517034068,1)
(0.519038076,1)
(0.521042084,1)
(0.523046092,1)
(0.5250501,1)
(0.527054108,1)
(0.529058116,1)
(0.531062124,1)
(0.533066132,1)
(0.53507014,1)
(0.537074148,1)
(0.539078156,1)
(0.541082164,1)
(0.543086172,1)
(0.54509018,1)
(0.547094188,1)
(0.549098196,1)
(0.551102204,1)
(0.553106212,1)
(0.55511022,1)
(0.557114228,1)
(0.559118236,1)
(0.561122244,1)
(0.563126253,1)
(0.565130261,1)
(0.567134269,1)
(0.569138277,1)
(0.571142285,1)
(0.573146293,1)
(0.575150301,1)
(0.577154309,1)
(0.579158317,1)
(0.581162325,1)
(0.583166333,1)
(0.585170341,1)
(0.587174349,1)
(0.589178357,1)
(0.591182365,1)
(0.593186373,1)
(0.595190381,1)
(0.597194389,1)
(0.599198397,1)
(0.601202405,1)
(0.603206413,1)
(0.605210421,1)
(0.607214429,1)
(0.609218437,1)
(0.611222445,1)
(0.613226453,1)
(0.615230461,1)
(0.617234469,1)
(0.619238477,1)
(0.621242485,1)
(0.623246493,1)
(0.625250501,1)
(0.627254509,1)
(0.629258517,1)
(0.631262525,1)
(0.633266533,1)
(0.635270541,1)
(0.637274549,1)
(0.639278557,1)
(0.641282565,1)
(0.643286573,1)
(0.645290581,1)
(0.647294589,1)
(0.649298597,1)
(0.651302605,1)
(0.653306613,1)
(0.655310621,1)
(0.657314629,1)
(0.659318637,1)
(0.661322645,1)
(0.663326653,1)
(0.665330661,1)
(0.667334669,1)
(0.669338677,1)
(0.671342685,1)
(0.673346693,1)
(0.675350701,1)
(0.677354709,1)
(0.679358717,1)
(0.681362725,1)
(0.683366733,1)
(0.685370741,1)
(0.687374749,1)
(0.689378758,1)
(0.691382766,1)
(0.693386774,1)
(0.695390782,1)
(0.69739479,1)
(0.699398798,1)
(0.701402806,1)
(0.703406814,1)
(0.705410822,1)
(0.70741483,1)
(0.709418838,1)
(0.711422846,1)
(0.713426854,1)
(0.715430862,1)
(0.71743487,1)
(0.719438878,1)
(0.721442886,1)
(0.723446894,1)
(0.725450902,1)
(0.72745491,1)
(0.729458918,1)
(0.731462926,1)
(0.733466934,1)
(0.735470942,1)
(0.73747495,1)
(0.739478958,1)
(0.741482966,1)
(0.743486974,1)
(0.745490982,1)
(0.74749499,1)
(0.749498998,1)
(0.751503006,1)
(0.753507014,1)
(0.755511022,1)
(0.75751503,1)
(0.759519038,1)
(0.761523046,1)
(0.763527054,1)
(0.765531062,1)
(0.76753507,1)
(0.769539078,1)
(0.771543086,1)
(0.773547094,1)
(0.775551102,1)
(0.77755511,1)
(0.779559118,1)
(0.781563126,1)
(0.783567134,1)
(0.785571142,1)
(0.78757515,1)
(0.789579158,1)
(0.791583166,1)
(0.793587174,1)
(0.795591182,1)
(0.79759519,1)
(0.799599198,1)
(0.801603206,1)
(0.803607214,1)
(0.805611222,1)
(0.80761523,1)
(0.809619238,1)
(0.811623246,1)
(0.813627255,1)
(0.815631263,1)
(0.817635271,1)
(0.819639279,1)
(0.821643287,1)
(0.823647295,1)
(0.825651303,1)
(0.827655311,1)
(0.829659319,1)
(0.831663327,1)
(0.833667335,1)
(0.835671343,1)
(0.837675351,1)
(0.839679359,1)
(0.841683367,1)
(0.843687375,1)
(0.845691383,1)
(0.847695391,1)
(0.849699399,1)
(0.851703407,1)
(0.853707415,1)
(0.855711423,1)
(0.857715431,1)
(0.859719439,1)
(0.861723447,1)
(0.863727455,1)
(0.865731463,1)
(0.867735471,1)
(0.869739479,1)
(0.871743487,1)
(0.873747495,1)
(0.875751503,1)
(0.877755511,1)
(0.879759519,1)
(0.881763527,1)
(0.883767535,1)
(0.885771543,1)
(0.887775551,1)
(0.889779559,1)
(0.891783567,1)
(0.893787575,1)
(0.895791583,1)
(0.897795591,1)
(0.899799599,1)
(0.901803607,1)
(0.903807615,1)
(0.905811623,1)
(0.907815631,1)
(0.909819639,1)
(0.911823647,1)
(0.913827655,1)
(0.915831663,1)
(0.917835671,1)
(0.919839679,1)
(0.921843687,1)
(0.923847695,1)
(0.925851703,1)
(0.927855711,1)
(0.929859719,1)
(0.931863727,1)
(0.933867735,1)
(0.935871743,1)
(0.937875752,1)
(0.93987976,1)
(0.941883768,1)
(0.943887776,1)
(0.945891784,1)
(0.947895792,1)
(0.9498998,1)
(0.951903808,1)
(0.953907816,1)
(0.955911824,1)
(0.957915832,1)
(0.95991984,1)
(0.961923848,1)
(0.963927856,1)
(0.965931864,1)
(0.967935872,1)
(0.96993988,1)
(0.971943888,1)
(0.973947896,1)
(0.975951904,1)
(0.977955912,1)
(0.97995992,1)
(0.981963928,1)
(0.983967936,1)
(0.985971944,1)
(0.987975952,1)
(0.98997996,1)
(0.991983968,1)
(0.993987976,1)
(0.995991984,1)
(0.997995992,1)
(1,1)
};

\addlegendentry{10 agents}
\addplot[
color=orange,
mark=o,
mark size=0.5pt
]
coordinates {(0.0009765625,10)
(0.002004008,10)
(0.004008016,10)
(0.006012024,10)
(0.008016032,10)
(0.01002004,10)
(0.012024048,10)
(0.014028056,10)
(0.016032064,10)
(0.018036072,10)
(0.02004008,10)
(0.022044088,10)
(0.024048096,10)
(0.026052104,10)
(0.028056112,10)
(0.03006012,10)
(0.032064128,10)
(0.034068136,10)
(0.036072144,10)
(0.038076152,10)
(0.04008016,10)
(0.042084168,10)
(0.044088176,10)
(0.046092184,10)
(0.048096192,10)
(0.0501002,10)
(0.052104208,10)
(0.054108216,10)
(0.056112224,10)
(0.058116232,10)
(0.06012024,10)
(0.062124248,10)
(0.064128257,10)
(0.066132265,10)
(0.068136273,10)
(0.070140281,10)
(0.072144289,10)
(0.074148297,10)
(0.076152305,10)
(0.078156313,10)
(0.080160321,10)
(0.082164329,10)
(0.084168337,10)
(0.086172345,10)
(0.088176353,10)
(0.090180361,10)
(0.092184369,10)
(0.094188377,10)
(0.096192385,10)
(0.098196393,10)
(0.100200401,10)
(0.102204409,10)
(0.104208417,10)
(0.106212425,10)
(0.108216433,10)
(0.110220441,10)
(0.112224449,3)
(0.114228457,3)
(0.116232465,3)
(0.118236473,3)
(0.120240481,3)
(0.122244489,3)
(0.124248497,3)
(0.126252505,3)
(0.128256513,3)
(0.130260521,3)
(0.132264529,3)
(0.134268537,3)
(0.136272545,3)
(0.138276553,3)
(0.140280561,3)
(0.142284569,3)
(0.144288577,3)
(0.146292585,3)
(0.148296593,3)
(0.150300601,3)
(0.152304609,3)
(0.154308617,3)
(0.156312625,3)
(0.158316633,3)
(0.160320641,3)
(0.162324649,3)
(0.164328657,3)
(0.166332665,3)
(0.168336673,3)
(0.170340681,3)
(0.172344689,3)
(0.174348697,3)
(0.176352705,3)
(0.178356713,3)
(0.180360721,3)
(0.182364729,3)
(0.184368737,3)
(0.186372745,3)
(0.188376754,3)
(0.190380762,2)
(0.19238477,2)
(0.194388778,2)
(0.196392786,2)
(0.198396794,2)
(0.200400802,2)
(0.20240481,2)
(0.204408818,2)
(0.206412826,2)
(0.208416834,2)
(0.210420842,2)
(0.21242485,2)
(0.214428858,2)
(0.216432866,2)
(0.218436874,2)
(0.220440882,2)
(0.22244489,2)
(0.224448898,2)
(0.226452906,2)
(0.228456914,2)
(0.230460922,2)
(0.23246493,2)
(0.234468938,2)
(0.236472946,2)
(0.238476954,2)
(0.240480962,2)
(0.24248497,2)
(0.244488978,2)
(0.246492986,2)
(0.248496994,2)
(0.250501002,2)
(0.25250501,2)
(0.254509018,2)
(0.256513026,2)
(0.258517034,2)
(0.260521042,2)
(0.26252505,2)
(0.264529058,2)
(0.266533066,2)
(0.268537074,2)
(0.270541082,2)
(0.27254509,1)
(0.274549098,1)
(0.276553106,1)
(0.278557114,2)
(0.280561122,2)
(0.28256513,2)
(0.284569138,2)
(0.286573146,2)
(0.288577154,2)
(0.290581162,2)
(0.29258517,2)
(0.294589178,2)
(0.296593186,1)
(0.298597194,1)
(0.300601202,1)
(0.30260521,1)
(0.304609218,1)
(0.306613226,1)
(0.308617234,1)
(0.310621242,1)
(0.312625251,1)
(0.314629259,1)
(0.316633267,1)
(0.318637275,1)
(0.320641283,1)
(0.322645291,1)
(0.324649299,1)
(0.326653307,1)
(0.328657315,1)
(0.330661323,1)
(0.332665331,1)
(0.334669339,1)
(0.336673347,1)
(0.338677355,1)
(0.340681363,1)
(0.342685371,1)
(0.344689379,1)
(0.346693387,1)
(0.348697395,1)
(0.350701403,1)
(0.352705411,1)
(0.354709419,1)
(0.356713427,1)
(0.358717435,1)
(0.360721443,1)
(0.362725451,1)
(0.364729459,1)
(0.366733467,1)
(0.368737475,1)
(0.370741483,1)
(0.372745491,1)
(0.374749499,1)
(0.376753507,1)
(0.378757515,1)
(0.380761523,1)
(0.382765531,1)
(0.384769539,1)
(0.386773547,1)
(0.388777555,1)
(0.390781563,1)
(0.392785571,1)
(0.394789579,1)
(0.396793587,1)
(0.398797595,1)
(0.400801603,1)
(0.402805611,1)
(0.404809619,1)
(0.406813627,1)
(0.408817635,1)
(0.410821643,1)
(0.412825651,1)
(0.414829659,1)
(0.416833667,1)
(0.418837675,1)
(0.420841683,1)
(0.422845691,1)
(0.424849699,1)
(0.426853707,1)
(0.428857715,1)
(0.430861723,1)
(0.432865731,1)
(0.434869739,1)
(0.436873747,1)
(0.438877756,1)
(0.440881764,1)
(0.442885772,1)
(0.44488978,1)
(0.446893788,1)
(0.448897796,1)
(0.450901804,1)
(0.452905812,1)
(0.45490982,1)
(0.456913828,1)
(0.458917836,1)
(0.460921844,1)
(0.462925852,1)
(0.46492986,1)
(0.466933868,1)
(0.468937876,1)
(0.470941884,1)
(0.472945892,1)
(0.4749499,1)
(0.476953908,1)
(0.478957916,1)
(0.480961924,1)
(0.482965932,1)
(0.48496994,1)
(0.486973948,1)
(0.488977956,1)
(0.490981964,1)
(0.492985972,1)
(0.49498998,1)
(0.496993988,1)
(0.498997996,1)
(0.501002004,1)
(0.503006012,1)
(0.50501002,1)
(0.507014028,1)
(0.509018036,1)
(0.511022044,1)
(0.513026052,1)
(0.51503006,1)
(0.517034068,1)
(0.519038076,1)
(0.521042084,1)
(0.523046092,1)
(0.5250501,1)
(0.527054108,1)
(0.529058116,1)
(0.531062124,1)
(0.533066132,1)
(0.53507014,1)
(0.537074148,1)
(0.539078156,1)
(0.541082164,1)
(0.543086172,1)
(0.54509018,1)
(0.547094188,1)
(0.549098196,1)
(0.551102204,1)
(0.553106212,1)
(0.55511022,1)
(0.557114228,1)
(0.559118236,1)
(0.561122244,1)
(0.563126253,1)
(0.565130261,1)
(0.567134269,1)
(0.569138277,1)
(0.571142285,1)
(0.573146293,1)
(0.575150301,1)
(0.577154309,1)
(0.579158317,1)
(0.581162325,1)
(0.583166333,1)
(0.585170341,1)
(0.587174349,1)
(0.589178357,1)
(0.591182365,1)
(0.593186373,1)
(0.595190381,1)
(0.597194389,1)
(0.599198397,1)
(0.601202405,1)
(0.603206413,1)
(0.605210421,1)
(0.607214429,1)
(0.609218437,1)
(0.611222445,1)
(0.613226453,1)
(0.615230461,1)
(0.617234469,1)
(0.619238477,1)
(0.621242485,1)
(0.623246493,1)
(0.625250501,1)
(0.627254509,1)
(0.629258517,1)
(0.631262525,1)
(0.633266533,1)
(0.635270541,1)
(0.637274549,1)
(0.639278557,1)
(0.641282565,1)
(0.643286573,1)
(0.645290581,1)
(0.647294589,1)
(0.649298597,1)
(0.651302605,1)
(0.653306613,1)
(0.655310621,1)
(0.657314629,1)
(0.659318637,1)
(0.661322645,1)
(0.663326653,1)
(0.665330661,1)
(0.667334669,1)
(0.669338677,1)
(0.671342685,1)
(0.673346693,1)
(0.675350701,1)
(0.677354709,1)
(0.679358717,1)
(0.681362725,1)
(0.683366733,1)
(0.685370741,1)
(0.687374749,1)
(0.689378758,1)
(0.691382766,1)
(0.693386774,1)
(0.695390782,1)
(0.69739479,1)
(0.699398798,1)
(0.701402806,1)
(0.703406814,1)
(0.705410822,1)
(0.70741483,1)
(0.709418838,1)
(0.711422846,1)
(0.713426854,1)
(0.715430862,1)
(0.71743487,1)
(0.719438878,1)
(0.721442886,1)
(0.723446894,1)
(0.725450902,1)
(0.72745491,1)
(0.729458918,1)
(0.731462926,1)
(0.733466934,1)
(0.735470942,1)
(0.73747495,1)
(0.739478958,1)
(0.741482966,1)
(0.743486974,1)
(0.745490982,1)
(0.74749499,1)
(0.749498998,1)
(0.751503006,1)
(0.753507014,1)
(0.755511022,1)
(0.75751503,1)
(0.759519038,1)
(0.761523046,1)
(0.763527054,1)
(0.765531062,1)
(0.76753507,1)
(0.769539078,1)
(0.771543086,1)
(0.773547094,1)
(0.775551102,1)
(0.77755511,1)
(0.779559118,1)
(0.781563126,1)
(0.783567134,1)
(0.785571142,1)
(0.78757515,1)
(0.789579158,1)
(0.791583166,1)
(0.793587174,1)
(0.795591182,1)
(0.79759519,1)
(0.799599198,1)
(0.801603206,1)
(0.803607214,1)
(0.805611222,1)
(0.80761523,1)
(0.809619238,1)
(0.811623246,1)
(0.813627255,1)
(0.815631263,1)
(0.817635271,1)
(0.819639279,1)
(0.821643287,1)
(0.823647295,1)
(0.825651303,1)
(0.827655311,1)
(0.829659319,1)
(0.831663327,1)
(0.833667335,1)
(0.835671343,1)
(0.837675351,1)
(0.839679359,1)
(0.841683367,1)
(0.843687375,1)
(0.845691383,1)
(0.847695391,1)
(0.849699399,1)
(0.851703407,1)
(0.853707415,1)
(0.855711423,1)
(0.857715431,1)
(0.859719439,1)
(0.861723447,1)
(0.863727455,1)
(0.865731463,1)
(0.867735471,1)
(0.869739479,1)
(0.871743487,1)
(0.873747495,1)
(0.875751503,1)
(0.877755511,1)
(0.879759519,1)
(0.881763527,1)
(0.883767535,1)
(0.885771543,1)
(0.887775551,1)
(0.889779559,1)
(0.891783567,1)
(0.893787575,1)
(0.895791583,1)
(0.897795591,1)
(0.899799599,1)
(0.901803607,1)
(0.903807615,1)
(0.905811623,1)
(0.907815631,1)
(0.909819639,1)
(0.911823647,1)
(0.913827655,1)
(0.915831663,1)
(0.917835671,1)
(0.919839679,1)
(0.921843687,1)
(0.923847695,1)
(0.925851703,1)
(0.927855711,1)
(0.929859719,1)
(0.931863727,1)
(0.933867735,1)
(0.935871743,1)
(0.937875752,1)
(0.93987976,1)
(0.941883768,1)
(0.943887776,1)
(0.945891784,1)
(0.947895792,1)
(0.9498998,1)
(0.951903808,1)
(0.953907816,1)
(0.955911824,1)
(0.957915832,1)
(0.95991984,1)
(0.961923848,1)
(0.963927856,1)
(0.965931864,1)
(0.967935872,1)
(0.96993988,1)
(0.971943888,1)
(0.973947896,1)
(0.975951904,1)
(0.977955912,1)
(0.97995992,1)
(0.981963928,1)
(0.983967936,1)
(0.985971944,1)
(0.987975952,1)
(0.98997996,1)
(0.991983968,1)
(0.993987976,1)
(0.995991984,1)
(0.997995992,1)
(1,1)	
};

\end{axis}
\end{tikzpicture}

%% file: Figures/closeOpenExample.tex
\begin{tikzpicture}
\begin{axis}[
xlabel={Opinion ($x(t)$)},
ylabel={Time instance ($t$)},
xmin=0, xmax=1,
ymin=0, ymax=12,
xtick={0,0.1,0.2,0.3,0.4,0.5,0.6,0.7,0.8,0.9,1.0},
ytick={0,4,8,12},
xticklabel style={rotate=60},
legend style={at={(1,1)},anchor=north east,font=\footnotesize},
ymajorgrids=true,
grid style=dashed
]

\addplot[
color=blue,
mark=o,
mark size=1pt
]
table {
0.3 0
0.4975 1
0.5178703703703703 2
0.5090946502057613 3
0.5042192501143119 4
0.501510694507951 5
0.5000059413933061 6
0.49916996744072567 7
0.49870553746706986 8
0.49844752081503885 9
0.49830417823057715 10
0.49822454346143175 11
0.49818030192301765 12
0.4981557232905654 13
0.49814206849475856 14
0.49813448249708814 15
0.4981302680539378 16
0.49812792669663214 17
0.4981266259425734 18
};

\addplot[
color=red,
mark=diamond,
mark size=1.2pt
]
table {
0.35 0
0.365 1
0.365 2
0.365 3
0.365 4
0.365 5
0.365 6
0.365 7
0.365 8
0.365 9
0.365 10
0.365 11
0.365 12
0.365 13
0.365 14
0.365 15
0.365 16
0.365 17
0.365 18
};

\addplot[
color=red,
mark=diamond,
mark size=1.2pt
]
table {
0.38 0
0.365 1
0.365 2
0.365 3
0.365 4
0.365 5
0.365 6
0.365 7
0.365 8
0.365 9
0.365 10
0.365 11
0.365 12
0.365 13
0.365 14
0.365 15
0.365 16
0.365 17
0.365 18
};

\addplot[
color=blue,
mark=o,
mark size=1pt
]
table {
0.45 0
0.5311111111111111 1
0.5178703703703703 2
0.5090946502057613 3
0.5042192501143119 4
0.501510694507951 5
0.5000059413933061 6
0.49916996744072567 7
0.49870553746706986 8
0.49844752081503885 9
0.49830417823057715 10
0.49822454346143175 11
0.49818030192301765 12
0.4981557232905654 13
0.49814206849475856 14
0.49813448249708814 15
0.4981302680539378 16
0.49812792669663214 17
0.4981266259425734 18
};

\addplot[
color=blue,
mark=o,
mark size=1pt
]
table {
0.55 0
0.5311111111111111 1
0.5178703703703703 2
0.5090946502057613 3
0.5042192501143119 4
0.501510694507951 5
0.5000059413933061 6
0.49916996744072567 7
0.49870553746706986 8
0.49844752081503885 9
0.49830417823057715 10
0.49822454346143175 11
0.49818030192301765 12
0.4981557232905654 13
0.49814206849475856 14
0.49813448249708814 15
0.4981302680539378 16
0.49812792669663214 17
0.4981266259425734 18
};

\addplot[
color=red,
mark=diamond,
mark size=1.2pt
]
table {
0.58 0
0.565 1
0.5774999999999999 2
0.5774999999999999 3
0.5774999999999999 4
0.5774999999999999 5
0.5774999999999999 6
0.5774999999999999 7
0.5774999999999999 8
0.5774999999999999 9
0.5774999999999999 10
0.5774999999999999 11
0.5774999999999999 12
0.5774999999999999 13
0.5774999999999999 14
0.5774999999999999 15
0.5774999999999999 16
0.5774999999999999 17
0.5774999999999999 18
};

\addplot[
color=red,
mark=diamond,
mark size=1.2pt
]
table {
0.67 0
0.685 1
0.685 2
0.685 3
0.685 4
0.685 5
0.685 6
0.685 7
0.685 8
0.685 9
0.685 10
0.685 11
0.685 12
0.685 13
0.685 14
0.685 15
0.685 16
0.685 17
0.685 18
};

\addplot[
color=blue,
mark=o,
mark size=1pt
]
table {
0.7 0
0.5311111111111111 1
0.5178703703703703 2
0.5090946502057613 3
0.5042192501143119 4
0.501510694507951 5
0.5000059413933061 6
0.49916996744072567 7
0.49870553746706986 8
0.49844752081503885 9
0.49830417823057715 10
0.49822454346143175 11
0.49818030192301765 12
0.4981557232905654 13
0.49814206849475856 14
0.49813448249708814 15
0.4981302680539378 16
0.49812792669663214 17
0.4981266259425734 18
};

\addplot[
color=blue,
mark=o,
mark size=1pt
]
table {
0.8 0
0.59 1
0.5178703703703703 2
0.5090946502057613 3
0.5042192501143119 4
0.501510694507951 5
0.5000059413933061 6
0.49916996744072567 7
0.49870553746706986 8
0.49844752081503885 9
0.49830417823057715 10
0.49822454346143175 11
0.49818030192301765 12
0.4981557232905654 13
0.49814206849475856 14
0.49813448249708814 15
0.4981302680539378 16
0.49812792669663214 17
0.4981266259425734 18
};

\end{axis}
\end{tikzpicture}

%% file: Figures/closeToModerate_numEpochs.tex
\begin{tikzpicture}
\begin{axis}[
xlabel={Fraction of close-minded transformed to moderate-minded},
ylabel={Time to equilibrium ($t_{eqm}$)},
xmin=0, xmax=1.05,
ymin=0, ymax=1000,
xtick={0,0.1,0.2,0.3,0.4,0.5,0.6,0.7,0.8,0.9,1.0},
ytick={0,200,400,600,800,1000},
xticklabel style={rotate=60},
legend style={at={(0,1)},anchor=north west,font=\footnotesize},
ymajorgrids=true,
grid style=dashed
]

\addlegendentry{Run \#1}
\addplot[
color=red,
mark=diamond,
mark size=0.6pt
]
coordinates {
(0,15)
(0.01010101,23)
(0.03030303,19)
(0.050505051,18)
(0.070707071,20)
(0.090909091,18)
(0.111111111,39)
(0.131313131,37)
(0.151515152,20)
(0.171717172,38)
(0.191919192,37)
(0.212121212,28)
(0.232323232,30)
(0.252525253,29)
(0.272727273,31)
(0.292929293,30)
(0.313131313,30)
(0.333333333,29)
(0.353535354,41)
(0.373737374,37)
(0.393939394,41)
(0.414141414,39)
(0.434343434,41)
(0.454545455,50)
(0.474747475,45)
(0.494949495,40)
(0.515151515,42)
(0.535353535,42)
(0.555555556,57)
(0.575757576,49)
(0.595959596,50)
(0.616161616,52)
(0.636363636,56)
(0.656565657,278)
(0.676767677,265)
(0.696969697,244)
(0.717171717,284)
(0.737373737,285)
(0.757575758,314)
(0.777777778,344)
(0.797979798,405)
(0.818181818,375)
(0.838383838,368)
(0.858585859,353)
(0.878787879,539)
(0.898989899,708)
(0.919191919,728)
(0.939393939,229)
(0.95959596,248)
(0.97979798,186)
(1,9)
};

\addlegendentry{Run \#2}
\addplot[
color=LimeGreen,
mark=triangle,
mark size=0.5pt
]
coordinates {
(0,15)
(0.01010101,19)
(0.03030303,41)
(0.050505051,43)
(0.070707071,26)
(0.090909091,31)
(0.111111111,23)
(0.131313131,33)
(0.151515152,29)
(0.171717172,24)
(0.191919192,25)
(0.212121212,49)
(0.232323232,35)
(0.252525253,24)
(0.272727273,29)
(0.292929293,28)
(0.313131313,25)
(0.333333333,31)
(0.353535354,22)
(0.373737374,25)
(0.393939394,36)
(0.414141414,43)
(0.434343434,71)
(0.454545455,34)
(0.474747475,37)
(0.494949495,60)
(0.515151515,84)
(0.535353535,83)
(0.555555556,56)
(0.575757576,99)
(0.595959596,78)
(0.616161616,85)
(0.636363636,99)
(0.656565657,131)
(0.676767677,111)
(0.696969697,118)
(0.717171717,124)
(0.737373737,371)
(0.757575758,382)
(0.777777778,409)
(0.797979798,382)
(0.818181818,133)
(0.838383838,429)
(0.858585859,148)
(0.878787879,213)
(0.898989899,168)
(0.919191919,253)
(0.939393939,860)
(0.95959596,1001)
(0.97979798,806)
(1,9)
};

\addlegendentry{Run \#3}
\addplot[
color=ProcessBlue,
mark=square,
mark size=0.5pt
]
coordinates {
(0,15)
(0.01010101,19)
(0.03030303,22)
(0.050505051,21)
(0.070707071,22)
(0.090909091,22)
(0.111111111,50)
(0.131313131,28)
(0.151515152,28)
(0.171717172,30)
(0.191919192,42)
(0.212121212,29)
(0.232323232,34)
(0.252525253,29)
(0.272727273,33)
(0.292929293,34)
(0.313131313,29)
(0.333333333,33)
(0.353535354,37)
(0.373737374,36)
(0.393939394,32)
(0.414141414,34)
(0.434343434,31)
(0.454545455,29)
(0.474747475,29)
(0.494949495,30)
(0.515151515,37)
(0.535353535,35)
(0.555555556,41)
(0.575757576,48)
(0.595959596,62)
(0.616161616,67)
(0.636363636,63)
(0.656565657,59)
(0.676767677,50)
(0.696969697,62)
(0.717171717,84)
(0.737373737,71)
(0.757575758,71)
(0.777777778,93)
(0.797979798,99)
(0.818181818,76)
(0.838383838,91)
(0.858585859,104)
(0.878787879,123)
(0.898989899,140)
(0.919191919,177)
(0.939393939,177)
(0.95959596,229)
(0.97979798,386)
(1,9)
};

\addlegendentry{Run \#4}
\addplot[
color=blue,
mark=+,
mark size=0.8pt
]
coordinates {
(0,15)
(0.01010101,17)
(0.03030303,18)
(0.050505051,21)
(0.070707071,22)
(0.090909091,25)
(0.111111111,22)
(0.131313131,44)
(0.151515152,47)
(0.171717172,33)
(0.191919192,29)
(0.212121212,27)
(0.232323232,27)
(0.252525253,27)
(0.272727273,27)
(0.292929293,26)
(0.313131313,38)
(0.333333333,44)
(0.353535354,37)
(0.373737374,50)
(0.393939394,45)
(0.414141414,44)
(0.434343434,51)
(0.454545455,49)
(0.474747475,47)
(0.494949495,53)
(0.515151515,54)
(0.535353535,59)
(0.555555556,59)
(0.575757576,68)
(0.595959596,54)
(0.616161616,57)
(0.636363636,74)
(0.656565657,41)
(0.676767677,44)
(0.696969697,51)
(0.717171717,57)
(0.737373737,65)
(0.757575758,71)
(0.777777778,65)
(0.797979798,74)
(0.818181818,96)
(0.838383838,102)
(0.858585859,148)
(0.878787879,193)
(0.898989899,215)
(0.919191919,379)
(0.939393939,226)
(0.95959596,225)
(0.97979798,1001)
(1,9)
};

\addlegendentry{Run \#5}
\addplot[
color=orange,
mark=o,
mark size=0.5pt
]
coordinates {
(0,15)
(0.01010101,27)
(0.03030303,30)
(0.050505051,29)
(0.070707071,23)
(0.090909091,79)
(0.111111111,80)
(0.131313131,25)
(0.151515152,35)
(0.171717172,33)
(0.191919192,42)
(0.212121212,32)
(0.232323232,57)
(0.252525253,29)
(0.272727273,29)
(0.292929293,30)
(0.313131313,25)
(0.333333333,30)
(0.353535354,29)
(0.373737374,25)
(0.393939394,28)
(0.414141414,30)
(0.434343434,30)
(0.454545455,30)
(0.474747475,32)
(0.494949495,33)
(0.515151515,48)
(0.535353535,51)
(0.555555556,49)
(0.575757576,54)
(0.595959596,58)
(0.616161616,42)
(0.636363636,65)
(0.656565657,66)
(0.676767677,64)
(0.696969697,66)
(0.717171717,80)
(0.737373737,91)
(0.757575758,85)
(0.777777778,83)
(0.797979798,79)
(0.818181818,130)
(0.838383838,562)
(0.858585859,193)
(0.878787879,113)
(0.898989899,131)
(0.919191919,541)
(0.939393939,286)
(0.95959596,844)
(0.97979798,1001)
(1,9)
};

\addlegendentry{Average}
\addplot[
color=black,
line width=2pt
]
coordinates {
(0,15)
(0.01010101,21)
(0.03030303,26)
(0.050505051,26.4)
(0.070707071,22.6)
(0.090909091,35)
(0.111111111,42.8)
(0.131313131,33.4)
(0.151515152,31.8)
(0.171717172,31.6)
(0.191919192,35)
(0.212121212,33)
(0.232323232,36.6)
(0.252525253,27.6)
(0.272727273,29.8)
(0.292929293,29.6)
(0.313131313,29.4)
(0.333333333,33.4)
(0.353535354,33.2)
(0.373737374,34.6)
(0.393939394,36.4)
(0.414141414,38)
(0.434343434,44.8)
(0.454545455,38.4)
(0.474747475,38)
(0.494949495,43.2)
(0.515151515,53)
(0.535353535,54)
(0.555555556,52.4)
(0.575757576,63.6)
(0.595959596,60.4)
(0.616161616,60.6)
(0.636363636,71.4)
(0.656565657,115)
(0.676767677,106.8)
(0.696969697,108.2)
(0.717171717,125.8)
(0.737373737,176.6)
(0.757575758,184.6)
(0.777777778,198.8)
(0.797979798,207.8)
(0.818181818,162)
(0.838383838,310.4)
(0.858585859,189.2)
(0.878787879,236.2)
(0.898989899,272.4)
(0.919191919,415.6)
(0.939393939,355.6)
(0.95959596,509.4)
(0.97979798,676)
(1,9)
};

\end{axis}
\end{tikzpicture}

%% file: Figures/closeToModerate_numClusters.tex
\begin{tikzpicture}
\begin{axis}[
xlabel={Fraction of close-minded transformed to moderate-minded},
ylabel={\#(clusters) at equilibrium ($C_{eqm}$)},
xmin=0, xmax=1.05,
ymin=0, ymax=45,
xtick={0,0.1,0.2,0.3,0.4,0.5,0.6,0.7,0.8,0.9,1.0},
ytick={0,15,30,45},
xticklabel style={rotate=60},
legend style={at={(0,0)},anchor=south west,font=\footnotesize},
ymajorgrids=true,
grid style=dashed
]

\addlegendentry{Run \#1}
\addplot[
color=red,
mark=diamond,
mark size=0.6pt
]
coordinates {
(0,42)
(0.01010101,41)
(0.03030303,40)
(0.050505051,42)
(0.070707071,44)
(0.090909091,45)
(0.111111111,44)
(0.131313131,43)
(0.151515152,43)
(0.171717172,44)
(0.191919192,42)
(0.212121212,41)
(0.232323232,43)
(0.252525253,43)
(0.272727273,41)
(0.292929293,44)
(0.313131313,43)
(0.333333333,42)
(0.353535354,40)
(0.373737374,40)
(0.393939394,44)
(0.414141414,40)
(0.434343434,40)
(0.454545455,38)
(0.474747475,39)
(0.494949495,35)
(0.515151515,34)
(0.535353535,35)
(0.555555556,35)
(0.575757576,36)
(0.595959596,34)
(0.616161616,32)
(0.636363636,30)
(0.656565657,17)
(0.676767677,16)
(0.696969697,17)
(0.717171717,16)
(0.737373737,16)
(0.757575758,14)
(0.777777778,14)
(0.797979798,15)
(0.818181818,12)
(0.838383838,7)
(0.858585859,6)
(0.878787879,9)
(0.898989899,9)
(0.919191919,8)
(0.939393939,11)
(0.95959596,8)
(0.97979798,5)
(1,1)
};

\addlegendentry{Run \#2}
\addplot[
color=LimeGreen,
mark=triangle,
mark size=0.5pt
]
coordinates {
(0,42)
(0.01010101,41)
(0.03030303,41)
(0.050505051,41)
(0.070707071,40)
(0.090909091,42)
(0.111111111,40)
(0.131313131,41)
(0.151515152,40)
(0.171717172,41)
(0.191919192,40)
(0.212121212,37)
(0.232323232,37)
(0.252525253,38)
(0.272727273,37)
(0.292929293,37)
(0.313131313,38)
(0.333333333,36)
(0.353535354,38)
(0.373737374,37)
(0.393939394,37)
(0.414141414,36)
(0.434343434,34)
(0.454545455,36)
(0.474747475,32)
(0.494949495,28)
(0.515151515,30)
(0.535353535,27)
(0.555555556,32)
(0.575757576,25)
(0.595959596,29)
(0.616161616,30)
(0.636363636,29)
(0.656565657,28)
(0.676767677,28)
(0.696969697,26)
(0.717171717,25)
(0.737373737,19)
(0.757575758,18)
(0.777777778,18)
(0.797979798,18)
(0.818181818,23)
(0.838383838,13)
(0.858585859,19)
(0.878787879,17)
(0.898989899,14)
(0.919191919,11)
(0.939393939,6)
(0.95959596,5)
(0.97979798,5)
(1,1)
};

\addlegendentry{Run \#3}
\addplot[
color=ProcessBlue,
mark=square,
mark size=0.5pt
]
coordinates {
(0,42)
(0.01010101,41)
(0.03030303,42)
(0.050505051,43)
(0.070707071,42)
(0.090909091,42)
(0.111111111,39)
(0.131313131,42)
(0.151515152,40)
(0.171717172,42)
(0.191919192,40)
(0.212121212,40)
(0.232323232,40)
(0.252525253,37)
(0.272727273,37)
(0.292929293,40)
(0.313131313,40)
(0.333333333,40)
(0.353535354,41)
(0.373737374,39)
(0.393939394,41)
(0.414141414,42)
(0.434343434,40)
(0.454545455,40)
(0.474747475,40)
(0.494949495,39)
(0.515151515,35)
(0.535353535,37)
(0.555555556,35)
(0.575757576,31)
(0.595959596,33)
(0.616161616,33)
(0.636363636,30)
(0.656565657,30)
(0.676767677,30)
(0.696969697,28)
(0.717171717,28)
(0.737373737,25)
(0.757575758,23)
(0.777777778,21)
(0.797979798,21)
(0.818181818,19)
(0.838383838,19)
(0.858585859,17)
(0.878787879,17)
(0.898989899,17)
(0.919191919,13)
(0.939393939,10)
(0.95959596,7)
(0.97979798,5)
(1,1)
};

\addlegendentry{Run \#4}
\addplot[
color=blue,
mark=+,
mark size=0.8pt
]
coordinates {
(0,42)
(0.01010101,41)
(0.03030303,43)
(0.050505051,45)
(0.070707071,44)
(0.090909091,42)
(0.111111111,44)
(0.131313131,42)
(0.151515152,40)
(0.171717172,42)
(0.191919192,40)
(0.212121212,37)
(0.232323232,36)
(0.252525253,37)
(0.272727273,38)
(0.292929293,40)
(0.313131313,35)
(0.333333333,32)
(0.353535354,34)
(0.373737374,33)
(0.393939394,33)
(0.414141414,34)
(0.434343434,32)
(0.454545455,33)
(0.474747475,31)
(0.494949495,31)
(0.515151515,28)
(0.535353535,31)
(0.555555556,32)
(0.575757576,31)
(0.595959596,32)
(0.616161616,31)
(0.636363636,29)
(0.656565657,31)
(0.676767677,31)
(0.696969697,29)
(0.717171717,28)
(0.737373737,25)
(0.757575758,22)
(0.777777778,21)
(0.797979798,23)
(0.818181818,22)
(0.838383838,19)
(0.858585859,16)
(0.878787879,15)
(0.898989899,11)
(0.919191919,9)
(0.939393939,11)
(0.95959596,8)
(0.97979798,3)
(1,1)
};

\addlegendentry{Run \#5}
\addplot[
color=orange,
mark=o,
mark size=0.5pt
]
coordinates {
(0,42)
(0.01010101,41)
(0.03030303,44)
(0.050505051,43)
(0.070707071,43)
(0.090909091,41)
(0.111111111,40)
(0.131313131,38)
(0.151515152,39)
(0.171717172,38)
(0.191919192,36)
(0.212121212,37)
(0.232323232,36)
(0.252525253,38)
(0.272727273,36)
(0.292929293,36)
(0.313131313,36)
(0.333333333,33)
(0.353535354,34)
(0.373737374,34)
(0.393939394,36)
(0.414141414,35)
(0.434343434,35)
(0.454545455,34)
(0.474747475,35)
(0.494949495,36)
(0.515151515,33)
(0.535353535,32)
(0.555555556,32)
(0.575757576,31)
(0.595959596,30)
(0.616161616,31)
(0.636363636,32)
(0.656565657,31)
(0.676767677,31)
(0.696969697,28)
(0.717171717,26)
(0.737373737,24)
(0.757575758,26)
(0.777777778,26)
(0.797979798,26)
(0.818181818,22)
(0.838383838,16)
(0.858585859,21)
(0.878787879,21)
(0.898989899,18)
(0.919191919,10)
(0.939393939,11)
(0.95959596,3)
(0.97979798,3)
(1,1)
};

\addlegendentry{Average}
\addplot[
color=black,
line width=2pt
]
coordinates {
(0,42)
(0.01010101,41)
(0.03030303,42)
(0.050505051,42.8)
(0.070707071,42.6)
(0.090909091,42.4)
(0.111111111,41.4)
(0.131313131,41.2)
(0.151515152,40.4)
(0.171717172,41.4)
(0.191919192,39.6)
(0.212121212,38.4)
(0.232323232,38.4)
(0.252525253,38.6)
(0.272727273,37.8)
(0.292929293,39.4)
(0.313131313,38.4)
(0.333333333,36.6)
(0.353535354,37.4)
(0.373737374,36.6)
(0.393939394,38.2)
(0.414141414,37.4)
(0.434343434,36.2)
(0.454545455,36.2)
(0.474747475,35.4)
(0.494949495,33.8)
(0.515151515,32)
(0.535353535,32.4)
(0.555555556,33.2)
(0.575757576,30.8)
(0.595959596,31.6)
(0.616161616,31.4)
(0.636363636,30)
(0.656565657,27.4)
(0.676767677,27.2)
(0.696969697,25.6)
(0.717171717,24.6)
(0.737373737,21.8)
(0.757575758,20.6)
(0.777777778,20)
(0.797979798,20.6)
(0.818181818,19.6)
(0.838383838,14.8)
(0.858585859,15.8)
(0.878787879,15.8)
(0.898989899,13.8)
(0.919191919,10.2)
(0.939393939,9.8)
(0.95959596,6.2)
(0.97979798,4.2)
(1,1)
};

\end{axis}
\end{tikzpicture}

%% file: Figures/openToModerate_numEpochs.tex
\begin{tikzpicture}
\begin{axis}[
xlabel={Fraction of open-minded transformed to moderate-minded},
ylabel={Time to equilibrium ($t_{eqm}$)},
xmin=0, xmax=1.05,
ymin=0, ymax=400,
xtick={0,0.1,0.2,0.3,0.4,0.5,0.6,0.7,0.8,0.9,1.0},
ytick={0,100,200,300,400},
xticklabel style={rotate=60},
legend style={at={(0,1)},anchor=north west,font=\footnotesize},
ymajorgrids=true,
grid style=dashed
]

\addlegendentry{Run \#1}
\addplot[
color=red,
mark=diamond,
mark size=0.6pt
]
coordinates {
(0,50)
(0.01010101,50)
(0.03030303,53)
(0.050505051,52)
(0.070707071,52)
(0.090909091,53)
(0.111111111,53)
(0.131313131,54)
(0.151515152,55)
(0.171717172,55)
(0.191919192,56)
(0.212121212,57)
(0.232323232,58)
(0.252525253,59)
(0.272727273,59)
(0.292929293,59)
(0.313131313,60)
(0.333333333,61)
(0.353535354,62)
(0.373737374,63)
(0.393939394,85)
(0.414141414,65)
(0.434343434,66)
(0.454545455,67)
(0.474747475,136)
(0.494949495,67)
(0.515151515,69)
(0.535353535,70)
(0.555555556,70)
(0.575757576,70)
(0.595959596,69)
(0.616161616,75)
(0.636363636,127)
(0.656565657,132)
(0.676767677,256)
(0.696969697,195)
(0.717171717,197)
(0.737373737,196)
(0.757575758,175)
(0.777777778,321)
(0.797979798,294)
(0.818181818,284)
(0.838383838,269)
(0.858585859,264)
(0.878787879,254)
(0.898989899,250)
(0.919191919,243)
(0.939393939,242)
(0.95959596,239)
(0.97979798,246)
(1,125)
};

\addlegendentry{Run \#2}
\addplot[
color=LimeGreen,
mark=triangle,
mark size=0.5pt
]
coordinates {
(0,50)
(0.01010101,50)
(0.03030303,50)
(0.050505051,52)
(0.070707071,53)
(0.090909091,55)
(0.111111111,54)
(0.131313131,67)
(0.151515152,68)
(0.171717172,60)
(0.191919192,60)
(0.212121212,71)
(0.232323232,104)
(0.252525253,108)
(0.272727273,64)
(0.292929293,69)
(0.313131313,70)
(0.333333333,71)
(0.353535354,65)
(0.373737374,67)
(0.393939394,76)
(0.414141414,93)
(0.434343434,101)
(0.454545455,117)
(0.474747475,87)
(0.494949495,111)
(0.515151515,125)
(0.535353535,120)
(0.555555556,120)
(0.575757576,120)
(0.595959596,131)
(0.616161616,74)
(0.636363636,75)
(0.656565657,75)
(0.676767677,76)
(0.696969697,199)
(0.717171717,193)
(0.737373737,214)
(0.757575758,330)
(0.777777778,312)
(0.797979798,292)
(0.818181818,279)
(0.838383838,273)
(0.858585859,260)
(0.878787879,256)
(0.898989899,253)
(0.919191919,245)
(0.939393939,244)
(0.95959596,246)
(0.97979798,248)
(1,125)
};

\addlegendentry{Run \#3}
\addplot[
color=ProcessBlue,
mark=square,
mark size=0.5pt
]
coordinates {
(0,50)
(0.01010101,50)
(0.03030303,51)
(0.050505051,51)
(0.070707071,52)
(0.090909091,53)
(0.111111111,55)
(0.131313131,54)
(0.151515152,56)
(0.171717172,56)
(0.191919192,56)
(0.212121212,55)
(0.232323232,55)
(0.252525253,56)
(0.272727273,57)
(0.292929293,57)
(0.313131313,57)
(0.333333333,58)
(0.353535354,59)
(0.373737374,60)
(0.393939394,61)
(0.414141414,63)
(0.434343434,64)
(0.454545455,64)
(0.474747475,66)
(0.494949495,68)
(0.515151515,68)
(0.535353535,68)
(0.555555556,69)
(0.575757576,71)
(0.595959596,65)
(0.616161616,71)
(0.636363636,71)
(0.656565657,73)
(0.676767677,96)
(0.696969697,78)
(0.717171717,75)
(0.737373737,70)
(0.757575758,83)
(0.777777778,113)
(0.797979798,84)
(0.818181818,83)
(0.838383838,268)
(0.858585859,266)
(0.878787879,262)
(0.898989899,254)
(0.919191919,242)
(0.939393939,241)
(0.95959596,238)
(0.97979798,245)
(1,125)
};

\addlegendentry{Run \#4}
\addplot[
color=blue,
mark=+,
mark size=0.8pt
]
coordinates {
(0,50)
(0.01010101,50)
(0.03030303,51)
(0.050505051,51)
(0.070707071,52)
(0.090909091,52)
(0.111111111,53)
(0.131313131,54)
(0.151515152,55)
(0.171717172,55)
(0.191919192,56)
(0.212121212,57)
(0.232323232,59)
(0.252525253,59)
(0.272727273,60)
(0.292929293,77)
(0.313131313,62)
(0.333333333,62)
(0.353535354,64)
(0.373737374,83)
(0.393939394,87)
(0.414141414,87)
(0.434343434,113)
(0.454545455,96)
(0.474747475,93)
(0.494949495,121)
(0.515151515,118)
(0.535353535,114)
(0.555555556,116)
(0.575757576,118)
(0.595959596,122)
(0.616161616,70)
(0.636363636,73)
(0.656565657,74)
(0.676767677,74)
(0.696969697,77)
(0.717171717,68)
(0.737373737,68)
(0.757575758,81)
(0.777777778,303)
(0.797979798,288)
(0.818181818,277)
(0.838383838,98)
(0.858585859,261)
(0.878787879,254)
(0.898989899,249)
(0.919191919,246)
(0.939393939,240)
(0.95959596,239)
(0.97979798,248)
(1,125)
};

\addlegendentry{Run \#5}
\addplot[
color=orange,
mark=o,
mark size=0.5pt
]
coordinates {
(0,50)
(0.01010101,50)
(0.03030303,51)
(0.050505051,52)
(0.070707071,49)
(0.090909091,49)
(0.111111111,53)
(0.131313131,55)
(0.151515152,49)
(0.171717172,50)
(0.191919192,55)
(0.212121212,50)
(0.232323232,58)
(0.252525253,59)
(0.272727273,60)
(0.292929293,75)
(0.313131313,61)
(0.333333333,62)
(0.353535354,63)
(0.373737374,63)
(0.393939394,64)
(0.414141414,65)
(0.434343434,65)
(0.454545455,66)
(0.474747475,68)
(0.494949495,92)
(0.515151515,120)
(0.535353535,140)
(0.555555556,139)
(0.575757576,128)
(0.595959596,128)
(0.616161616,133)
(0.636363636,132)
(0.656565657,135)
(0.676767677,214)
(0.696969697,194)
(0.717171717,216)
(0.737373737,390)
(0.757575758,345)
(0.777777778,315)
(0.797979798,294)
(0.818181818,281)
(0.838383838,270)
(0.858585859,262)
(0.878787879,255)
(0.898989899,252)
(0.919191919,244)
(0.939393939,240)
(0.95959596,243)
(0.97979798,239)
(1,125)
};

\addlegendentry{Average}
\addplot[
color=black,
line width=2pt
]
coordinates {
(0,50)
(0.01010101,50)
(0.03030303,51.2)
(0.050505051,51.6)
(0.070707071,51.6)
(0.090909091,52.4)
(0.111111111,53.6)
(0.131313131,56.8)
(0.151515152,56.6)
(0.171717172,55.2)
(0.191919192,56.6)
(0.212121212,58)
(0.232323232,66.8)
(0.252525253,68.2)
(0.272727273,60)
(0.292929293,67.4)
(0.313131313,62)
(0.333333333,62.8)
(0.353535354,62.6)
(0.373737374,67.2)
(0.393939394,74.6)
(0.414141414,74.6)
(0.434343434,81.8)
(0.454545455,82)
(0.474747475,90)
(0.494949495,91.8)
(0.515151515,100)
(0.535353535,102.4)
(0.555555556,102.8)
(0.575757576,101.4)
(0.595959596,103)
(0.616161616,84.6)
(0.636363636,95.6)
(0.656565657,97.8)
(0.676767677,143.2)
(0.696969697,148.6)
(0.717171717,149.8)
(0.737373737,187.6)
(0.757575758,202.8)
(0.777777778,272.8)
(0.797979798,250.4)
(0.818181818,240.8)
(0.838383838,235.6)
(0.858585859,262.6)
(0.878787879,256.2)
(0.898989899,251.6)
(0.919191919,244)
(0.939393939,241.4)
(0.95959596,241)
(0.97979798,245.2)
(1,125)
};

\end{axis}
\end{tikzpicture}

%% file: Figures/openToModerate_numClusters.tex
\begin{tikzpicture}
\begin{axis}[
xlabel={Fraction of open-minded transformed to moderate-minded},
ylabel={\#(clusters) at equilibrium ($C_{eqm}$)},
xmin=0, xmax=1.05,
ymin=0, ymax=45,
xtick={0,0.1,0.2,0.3,0.4,0.5,0.6,0.7,0.8,0.9,1.0},
ytick={0,15,30,45},
xticklabel style={rotate=60},
legend style={at={(0,0)},anchor=south west,font=\footnotesize},
ymajorgrids=true,
grid style=dashed
]

\addlegendentry{Run \#1}
\addplot[
color=red,
mark=diamond,
mark size=0.6pt
]
coordinates {
(0,39)
(0.01010101,43)
(0.03030303,40)
(0.050505051,40)
(0.070707071,39)
(0.090909091,41)
(0.111111111,43)
(0.131313131,40)
(0.151515152,40)
(0.171717172,39)
(0.191919192,41)
(0.212121212,42)
(0.232323232,39)
(0.252525253,39)
(0.272727273,40)
(0.292929293,37)
(0.313131313,37)
(0.333333333,38)
(0.353535354,39)
(0.373737374,37)
(0.393939394,37)
(0.414141414,35)
(0.434343434,35)
(0.454545455,36)
(0.474747475,37)
(0.494949495,38)
(0.515151515,35)
(0.535353535,37)
(0.555555556,37)
(0.575757576,37)
(0.595959596,38)
(0.616161616,35)
(0.636363636,36)
(0.656565657,35)
(0.676767677,36)
(0.696969697,34)
(0.717171717,34)
(0.737373737,35)
(0.757575758,35)
(0.777777778,31)
(0.797979798,33)
(0.818181818,36)
(0.838383838,34)
(0.858585859,33)
(0.878787879,35)
(0.898989899,31)
(0.919191919,32)
(0.939393939,30)
(0.95959596,32)
(0.97979798,26)
(1,28)
};

\addlegendentry{Run \#2}
\addplot[
color=LimeGreen,
mark=triangle,
mark size=0.5pt
]
coordinates {
(0,39)
(0.01010101,42)
(0.03030303,40)
(0.050505051,38)
(0.070707071,40)
(0.090909091,39)
(0.111111111,41)
(0.131313131,38)
(0.151515152,39)
(0.171717172,38)
(0.191919192,37)
(0.212121212,37)
(0.232323232,37)
(0.252525253,41)
(0.272727273,38)
(0.292929293,37)
(0.313131313,36)
(0.333333333,37)
(0.353535354,37)
(0.373737374,35)
(0.393939394,38)
(0.414141414,40)
(0.434343434,38)
(0.454545455,36)
(0.474747475,37)
(0.494949495,37)
(0.515151515,42)
(0.535353535,40)
(0.555555556,41)
(0.575757576,42)
(0.595959596,42)
(0.616161616,42)
(0.636363636,39)
(0.656565657,39)
(0.676767677,38)
(0.696969697,36)
(0.717171717,36)
(0.737373737,34)
(0.757575758,36)
(0.777777778,35)
(0.797979798,33)
(0.818181818,33)
(0.838383838,32)
(0.858585859,34)
(0.878787879,32)
(0.898989899,33)
(0.919191919,32)
(0.939393939,32)
(0.95959596,32)
(0.97979798,31)
(1,28)
};

\addlegendentry{Run \#3}
\addplot[
color=ProcessBlue,
mark=square,
mark size=0.5pt
]
coordinates {
(0,39)
(0.01010101,40)
(0.03030303,40)
(0.050505051,40)
(0.070707071,40)
(0.090909091,41)
(0.111111111,40)
(0.131313131,41)
(0.151515152,42)
(0.171717172,40)
(0.191919192,38)
(0.212121212,41)
(0.232323232,40)
(0.252525253,37)
(0.272727273,39)
(0.292929293,38)
(0.313131313,40)
(0.333333333,40)
(0.353535354,39)
(0.373737374,37)
(0.393939394,36)
(0.414141414,38)
(0.434343434,37)
(0.454545455,38)
(0.474747475,35)
(0.494949495,35)
(0.515151515,36)
(0.535353535,36)
(0.555555556,38)
(0.575757576,40)
(0.595959596,40)
(0.616161616,39)
(0.636363636,38)
(0.656565657,41)
(0.676767677,39)
(0.696969697,37)
(0.717171717,39)
(0.737373737,35)
(0.757575758,36)
(0.777777778,34)
(0.797979798,37)
(0.818181818,38)
(0.838383838,36)
(0.858585859,35)
(0.878787879,31)
(0.898989899,37)
(0.919191919,37)
(0.939393939,30)
(0.95959596,33)
(0.97979798,41)
(1,38)
};

\addlegendentry{Run \#4}
\addplot[
color=blue,
mark=+,
mark size=0.8pt
]
coordinates {
(0,39)
(0.01010101,39)
(0.03030303,41)
(0.050505051,41)
(0.070707071,42)
(0.090909091,42)
(0.111111111,41)
(0.131313131,39)
(0.151515152,42)
(0.171717172,41)
(0.191919192,40)
(0.212121212,40)
(0.232323232,41)
(0.252525253,40)
(0.272727273,40)
(0.292929293,39)
(0.313131313,39)
(0.333333333,40)
(0.353535354,38)
(0.373737374,38)
(0.393939394,38)
(0.414141414,38)
(0.434343434,37)
(0.454545455,38)
(0.474747475,37)
(0.494949495,38)
(0.515151515,37)
(0.535353535,38)
(0.555555556,36)
(0.575757576,37)
(0.595959596,37)
(0.616161616,38)
(0.636363636,39)
(0.656565657,39)
(0.676767677,37)
(0.696969697,36)
(0.717171717,38)
(0.737373737,36)
(0.757575758,37)
(0.777777778,34)
(0.797979798,35)
(0.818181818,34)
(0.838383838,36)
(0.858585859,37)
(0.878787879,36)
(0.898989899,31)
(0.919191919,31)
(0.939393939,32)
(0.95959596,32)
(0.97979798,29)
(1,38)
};

\addlegendentry{Run \#5}
\addplot[
color=orange,
mark=o,
mark size=0.5pt
]
coordinates {
(0,39)
(0.01010101,41)
(0.03030303,41)
(0.050505051,40)
(0.070707071,40)
(0.090909091,41)
(0.111111111,39)
(0.131313131,38)
(0.151515152,41)
(0.171717172,38)
(0.191919192,40)
(0.212121212,42)
(0.232323232,39)
(0.252525253,38)
(0.272727273,40)
(0.292929293,35)
(0.313131313,39)
(0.333333333,41)
(0.353535354,40)
(0.373737374,40)
(0.393939394,39)
(0.414141414,38)
(0.434343434,41)
(0.454545455,39)
(0.474747475,41)
(0.494949495,39)
(0.515151515,38)
(0.535353535,39)
(0.555555556,37)
(0.575757576,38)
(0.595959596,38)
(0.616161616,37)
(0.636363636,37)
(0.656565657,37)
(0.676767677,37)
(0.696969697,37)
(0.717171717,34)
(0.737373737,34)
(0.757575758,33)
(0.777777778,34)
(0.797979798,33)
(0.818181818,33)
(0.838383838,35)
(0.858585859,35)
(0.878787879,35)
(0.898989899,33)
(0.919191919,31)
(0.939393939,32)
(0.95959596,31)
(0.97979798,30)
(1,31)
};

\addlegendentry{Average}
\addplot[
color=black,
line width=2pt
]
coordinates {
(0,39)
(0.01010101,41)
(0.03030303,40.4)
(0.050505051,39.8)
(0.070707071,40.2)
(0.090909091,40.8)
(0.111111111,40.8)
(0.131313131,39.2)
(0.151515152,40.8)
(0.171717172,39.2)
(0.191919192,39.2)
(0.212121212,40.4)
(0.232323232,39.2)
(0.252525253,39)
(0.272727273,39.4)
(0.292929293,37.2)
(0.313131313,38.2)
(0.333333333,39.2)
(0.353535354,38.6)
(0.373737374,37.4)
(0.393939394,37.6)
(0.414141414,37.8)
(0.434343434,37.6)
(0.454545455,37.4)
(0.474747475,37.4)
(0.494949495,37.4)
(0.515151515,37.6)
(0.535353535,38)
(0.555555556,37.8)
(0.575757576,38.8)
(0.595959596,39)
(0.616161616,38.2)
(0.636363636,37.8)
(0.656565657,38.2)
(0.676767677,37.4)
(0.696969697,36)
(0.717171717,36.2)
(0.737373737,34.8)
(0.757575758,35.4)
(0.777777778,33.6)
(0.797979798,34.2)
(0.818181818,34.8)
(0.838383838,34.6)
(0.858585859,34.8)
(0.878787879,33.8)
(0.898989899,33)
(0.919191919,32.6)
(0.939393939,31.2)
(0.95959596,32)
(0.97979798,31.4)
(1,32.6)
};

\end{axis}
\end{tikzpicture}

%% file: Figures/newModerate_numEpochs.tex
\begin{tikzpicture}
\begin{axis}[
xlabel={Fraction of new moderate-minded agents w.r.t population size},
ylabel={Time to equilibrium ($t_{eqm}$)},
xmin=0, xmax=1.05,
ymin=0, ymax=300,
xtick={0,0.1,0.2,0.3,0.4,0.5,0.6,0.7,0.8,0.9,1.0},
ytick={0,100,200,300},
xticklabel style={rotate=60},
legend style={at={(0,1)},anchor=north west,font=\footnotesize},
ymajorgrids=true,
grid style=dashed
]

\addlegendentry{Run \#1}
\addplot[
color=red,
mark=diamond,
mark size=0.6pt
]
coordinates {
(0,24)
(0.01010101,23)
(0.03030303,24)
(0.050505051,29)
(0.070707071,37)
(0.090909091,29)
(0.111111111,28)
(0.131313131,41)
(0.151515152,29)
(0.171717172,46)
(0.191919192,29)
(0.212121212,29)
(0.232323232,47)
(0.252525253,48)
(0.272727273,50)
(0.292929293,52)
(0.313131313,53)
(0.333333333,54)
(0.353535354,55)
(0.373737374,56)
(0.393939394,56)
(0.414141414,58)
(0.434343434,59)
(0.454545455,59)
(0.474747475,62)
(0.494949495,65)
(0.515151515,68)
(0.535353535,71)
(0.555555556,74)
(0.575757576,78)
(0.595959596,74)
(0.616161616,75)
(0.636363636,75)
(0.656565657,94)
(0.676767677,81)
(0.696969697,81)
(0.717171717,83)
(0.737373737,85)
(0.757575758,85)
(0.777777778,92)
(0.797979798,158)
(0.818181818,96)
(0.838383838,161)
(0.858585859,100)
(0.878787879,102)
(0.898989899,154)
(0.919191919,204)
(0.939393939,195)
(0.95959596,163)
(0.97979798,231)
(1,234)
};

\addlegendentry{Run \#2}
\addplot[
color=LimeGreen,
mark=triangle,
mark size=0.5pt
]
coordinates {
(0,24)
(0.01010101,24)
(0.03030303,25)
(0.050505051,28)
(0.070707071,27)
(0.090909091,33)
(0.111111111,28)
(0.131313131,37)
(0.151515152,36)
(0.171717172,42)
(0.191919192,43)
(0.212121212,44)
(0.232323232,39)
(0.252525253,42)
(0.272727273,46)
(0.292929293,48)
(0.313131313,50)
(0.333333333,52)
(0.353535354,53)
(0.373737374,53)
(0.393939394,56)
(0.414141414,59)
(0.434343434,57)
(0.454545455,58)
(0.474747475,59)
(0.494949495,60)
(0.515151515,66)
(0.535353535,73)
(0.555555556,75)
(0.575757576,70)
(0.595959596,70)
(0.616161616,73)
(0.636363636,79)
(0.656565657,80)
(0.676767677,170)
(0.696969697,146)
(0.717171717,146)
(0.737373737,147)
(0.757575758,144)
(0.777777778,144)
(0.797979798,146)
(0.818181818,142)
(0.838383838,198)
(0.858585859,135)
(0.878787879,196)
(0.898989899,211)
(0.919191919,257)
(0.939393939,248)
(0.95959596,224)
(0.97979798,245)
(1,300)
};

\addlegendentry{Run \#3}
\addplot[
color=ProcessBlue,
mark=square,
mark size=0.5pt
]
coordinates {
(0,24)
(0.01010101,24)
(0.03030303,25)
(0.050505051,26)
(0.070707071,28)
(0.090909091,28)
(0.111111111,38)
(0.131313131,30)
(0.151515152,32)
(0.171717172,40)
(0.191919192,42)
(0.212121212,42)
(0.232323232,45)
(0.252525253,47)
(0.272727273,47)
(0.292929293,69)
(0.313131313,50)
(0.333333333,54)
(0.353535354,55)
(0.373737374,57)
(0.393939394,59)
(0.414141414,61)
(0.434343434,61)
(0.454545455,64)
(0.474747475,63)
(0.494949495,63)
(0.515151515,69)
(0.535353535,70)
(0.555555556,71)
(0.575757576,74)
(0.595959596,73)
(0.616161616,76)
(0.636363636,77)
(0.656565657,79)
(0.676767677,111)
(0.696969697,136)
(0.717171717,95)
(0.737373737,174)
(0.757575758,142)
(0.777777778,92)
(0.797979798,94)
(0.818181818,94)
(0.838383838,94)
(0.858585859,109)
(0.878787879,196)
(0.898989899,74)
(0.919191919,109)
(0.939393939,281)
(0.95959596,263)
(0.97979798,115)
(1,110)
};

\addlegendentry{Run \#4}
\addplot[
color=blue,
mark=+,
mark size=0.8pt
]
coordinates {
(0,24)
(0.01010101,24)
(0.03030303,25)
(0.050505051,32)
(0.070707071,30)
(0.090909091,30)
(0.111111111,30)
(0.131313131,37)
(0.151515152,31)
(0.171717172,32)
(0.191919192,42)
(0.212121212,46)
(0.232323232,45)
(0.252525253,47)
(0.272727273,51)
(0.292929293,51)
(0.313131313,53)
(0.333333333,55)
(0.353535354,56)
(0.373737374,55)
(0.393939394,54)
(0.414141414,62)
(0.434343434,65)
(0.454545455,63)
(0.474747475,63)
(0.494949495,67)
(0.515151515,64)
(0.535353535,66)
(0.555555556,82)
(0.575757576,54)
(0.595959596,72)
(0.616161616,83)
(0.636363636,84)
(0.656565657,57)
(0.676767677,131)
(0.696969697,124)
(0.717171717,83)
(0.737373737,139)
(0.757575758,120)
(0.777777778,157)
(0.797979798,149)
(0.818181818,150)
(0.838383838,148)
(0.858585859,100)
(0.878787879,158)
(0.898989899,151)
(0.919191919,156)
(0.939393939,156)
(0.95959596,162)
(0.97979798,249)
(1,76)
};

\addlegendentry{Run \#5}
\addplot[
color=orange,
mark=o,
mark size=0.5pt
]
coordinates {
(0,24)
(0.01010101,24)
(0.03030303,26)
(0.050505051,25)
(0.070707071,27)
(0.090909091,33)
(0.111111111,35)
(0.131313131,40)
(0.151515152,36)
(0.171717172,39)
(0.191919192,44)
(0.212121212,36)
(0.232323232,64)
(0.252525253,51)
(0.272727273,51)
(0.292929293,52)
(0.313131313,53)
(0.333333333,55)
(0.353535354,58)
(0.373737374,65)
(0.393939394,70)
(0.414141414,81)
(0.434343434,64)
(0.454545455,67)
(0.474747475,68)
(0.494949495,72)
(0.515151515,66)
(0.535353535,74)
(0.555555556,110)
(0.575757576,110)
(0.595959596,123)
(0.616161616,109)
(0.636363636,134)
(0.656565657,135)
(0.676767677,140)
(0.696969697,85)
(0.717171717,141)
(0.737373737,145)
(0.757575758,148)
(0.777777778,146)
(0.797979798,145)
(0.818181818,149)
(0.838383838,148)
(0.858585859,151)
(0.878787879,148)
(0.898989899,193)
(0.919191919,246)
(0.939393939,245)
(0.95959596,238)
(0.97979798,229)
(1,233)
};

\addlegendentry{Average}
\addplot[
color=black,
line width=2pt
]
coordinates {
(0,24)
(0.01010101,23.8)
(0.03030303,25)
(0.050505051,28)
(0.070707071,29.8)
(0.090909091,30.6)
(0.111111111,31.8)
(0.131313131,37)
(0.151515152,32.8)
(0.171717172,39.8)
(0.191919192,40)
(0.212121212,39.4)
(0.232323232,48)
(0.252525253,47)
(0.272727273,49)
(0.292929293,54.4)
(0.313131313,51.8)
(0.333333333,54)
(0.353535354,55.4)
(0.373737374,57.2)
(0.393939394,59)
(0.414141414,64.2)
(0.434343434,61.2)
(0.454545455,62.2)
(0.474747475,63)
(0.494949495,65.4)
(0.515151515,66.6)
(0.535353535,70.8)
(0.555555556,82.4)
(0.575757576,77.2)
(0.595959596,82.4)
(0.616161616,83.2)
(0.636363636,89.8)
(0.656565657,89)
(0.676767677,126.6)
(0.696969697,114.4)
(0.717171717,109.6)
(0.737373737,138)
(0.757575758,127.8)
(0.777777778,126.2)
(0.797979798,138.4)
(0.818181818,126.2)
(0.838383838,149.8)
(0.858585859,119)
(0.878787879,160)
(0.898989899,156.6)
(0.919191919,194.4)
(0.939393939,225)
(0.95959596,210)
(0.97979798,213.8)
(1,190.6)
};

\end{axis}
\end{tikzpicture}

%% file: Figures/newModerate_numClusters.tex
\begin{tikzpicture}
\begin{axis}[
xlabel={Fraction of new moderate-minded agents w.r.t population size},
ylabel={\#(clusters) at equilibrium ($C_{eqm}$)},
xmin=0, xmax=1.05,
ymin=0, ymax=45,
xtick={0,0.1,0.2,0.3,0.4,0.5,0.6,0.7,0.8,0.9,1.0},
ytick={0,15,30,45},
xticklabel style={rotate=60},
legend style={at={(0,0)},anchor=south west,font=\footnotesize},
ymajorgrids=true,
grid style=dashed
]

\addlegendentry{Run \#1}
\addplot[
color=red,
mark=diamond,
mark size=0.6pt
]
coordinates {
(0,39)
(0.01010101,43)
(0.02020202,44)
(0.03030303,40)
(0.04040404,42)
(0.050505051,40)
(0.060606061,38)
(0.070707071,39)
(0.080808081,41)
(0.090909091,41)
(0.101010101,42)
(0.111111111,43)
(0.121212121,41)
(0.131313131,40)
(0.141414141,40)
(0.151515152,40)
(0.161616162,39)
(0.171717172,39)
(0.181818182,40)
(0.191919192,41)
(0.202020202,40)
(0.212121212,42)
(0.222222222,41)
(0.232323232,39)
(0.242424242,38)
(0.252525253,39)
(0.262626263,39)
(0.272727273,40)
(0.282828283,38)
(0.292929293,37)
(0.303030303,37)
(0.313131313,37)
(0.323232323,38)
(0.333333333,38)
(0.343434343,39)
(0.353535354,39)
(0.363636364,38)
(0.373737374,37)
(0.383838384,38)
(0.393939394,37)
(0.404040404,34)
(0.414141414,35)
(0.424242424,33)
(0.434343434,35)
(0.444444444,37)
(0.454545455,36)
(0.464646465,37)
(0.474747475,37)
(0.484848485,37)
(0.494949495,38)
(0.505050505,35)
(0.515151515,35)
(0.525252525,37)
(0.535353535,37)
(0.545454545,36)
(0.555555556,37)
(0.565656566,36)
(0.575757576,37)
(0.585858586,38)
(0.595959596,38)
(0.606060606,38)
(0.616161616,35)
(0.626262626,38)
(0.636363636,36)
(0.646464646,35)
(0.656565657,35)
(0.666666667,34)
(0.676767677,36)
(0.686868687,35)
(0.696969697,34)
(0.707070707,34)
(0.717171717,34)
(0.727272727,34)
(0.737373737,35)
(0.747474747,33)
(0.757575758,35)
(0.767676768,34)
(0.777777778,31)
(0.787878788,35)
(0.797979798,33)
(0.808080808,38)
(0.818181818,36)
(0.828282828,37)
(0.838383838,34)
(0.848484848,34)
(0.858585859,33)
(0.868686869,36)
(0.878787879,35)
(0.888888889,31)
(0.898989899,31)
(0.909090909,30)
(0.919191919,32)
(0.929292929,30)
(0.939393939,30)
(0.949494949,30)
(0.95959596,32)
(0.96969697,28)
(0.97979798,26)
(0.98989899,27)
(1,28)
};

\addlegendentry{Run \#2}
\addplot[
color=LimeGreen,
mark=triangle,
mark size=0.5pt
]
coordinates {
(0,39)
(0.01010101,42)
(0.02020202,42)
(0.03030303,40)
(0.04040404,38)
(0.050505051,38)
(0.060606061,40)
(0.070707071,40)
(0.080808081,40)
(0.090909091,39)
(0.101010101,41)
(0.111111111,41)
(0.121212121,38)
(0.131313131,38)
(0.141414141,39)
(0.151515152,39)
(0.161616162,38)
(0.171717172,38)
(0.181818182,38)
(0.191919192,37)
(0.202020202,38)
(0.212121212,37)
(0.222222222,37)
(0.232323232,37)
(0.242424242,39)
(0.252525253,41)
(0.262626263,40)
(0.272727273,38)
(0.282828283,39)
(0.292929293,37)
(0.303030303,37)
(0.313131313,36)
(0.323232323,37)
(0.333333333,37)
(0.343434343,36)
(0.353535354,37)
(0.363636364,34)
(0.373737374,35)
(0.383838384,38)
(0.393939394,38)
(0.404040404,37)
(0.414141414,40)
(0.424242424,38)
(0.434343434,38)
(0.444444444,36)
(0.454545455,36)
(0.464646465,38)
(0.474747475,37)
(0.484848485,38)
(0.494949495,37)
(0.505050505,38)
(0.515151515,42)
(0.525252525,40)
(0.535353535,40)
(0.545454545,40)
(0.555555556,41)
(0.565656566,41)
(0.575757576,42)
(0.585858586,43)
(0.595959596,42)
(0.606060606,41)
(0.616161616,42)
(0.626262626,39)
(0.636363636,39)
(0.646464646,38)
(0.656565657,39)
(0.666666667,36)
(0.676767677,38)
(0.686868687,38)
(0.696969697,36)
(0.707070707,37)
(0.717171717,36)
(0.727272727,35)
(0.737373737,34)
(0.747474747,36)
(0.757575758,36)
(0.767676768,35)
(0.777777778,35)
(0.787878788,34)
(0.797979798,33)
(0.808080808,36)
(0.818181818,33)
(0.828282828,34)
(0.838383838,32)
(0.848484848,34)
(0.858585859,34)
(0.868686869,33)
(0.878787879,32)
(0.888888889,34)
(0.898989899,33)
(0.909090909,32)
(0.919191919,32)
(0.929292929,31)
(0.939393939,32)
(0.949494949,33)
(0.95959596,32)
(0.96969697,32)
(0.97979798,31)
(0.98989899,30)
(1,28)
};

\addlegendentry{Run \#3}
\addplot[
color=ProcessBlue,
mark=square,
mark size=0.5pt
]
coordinates {
(0,39)
(0.01010101,40)
(0.02020202,41)
(0.03030303,40)
(0.04040404,40)
(0.050505051,40)
(0.060606061,39)
(0.070707071,40)
(0.080808081,39)
(0.090909091,41)
(0.101010101,41)
(0.111111111,40)
(0.121212121,41)
(0.131313131,41)
(0.141414141,43)
(0.151515152,42)
(0.161616162,40)
(0.171717172,40)
(0.181818182,38)
(0.191919192,38)
(0.202020202,40)
(0.212121212,41)
(0.222222222,39)
(0.232323232,40)
(0.242424242,39)
(0.252525253,37)
(0.262626263,39)
(0.272727273,39)
(0.282828283,38)
(0.292929293,38)
(0.303030303,42)
(0.313131313,40)
(0.323232323,40)
(0.333333333,40)
(0.343434343,37)
(0.353535354,39)
(0.363636364,39)
(0.373737374,37)
(0.383838384,37)
(0.393939394,36)
(0.404040404,34)
(0.414141414,38)
(0.424242424,40)
(0.434343434,37)
(0.444444444,38)
(0.454545455,38)
(0.464646465,35)
(0.474747475,35)
(0.484848485,35)
(0.494949495,35)
(0.505050505,37)
(0.515151515,36)
(0.525252525,36)
(0.535353535,36)
(0.545454545,39)
(0.555555556,38)
(0.565656566,39)
(0.575757576,40)
(0.585858586,41)
(0.595959596,40)
(0.606060606,39)
(0.616161616,39)
(0.626262626,39)
(0.636363636,38)
(0.646464646,40)
(0.656565657,41)
(0.666666667,40)
(0.676767677,39)
(0.686868687,37)
(0.696969697,37)
(0.707070707,37)
(0.717171717,39)
(0.727272727,35)
(0.737373737,35)
(0.747474747,37)
(0.757575758,36)
(0.767676768,36)
(0.777777778,34)
(0.787878788,36)
(0.797979798,37)
(0.808080808,40)
(0.818181818,38)
(0.828282828,36)
(0.838383838,36)
(0.848484848,38)
(0.858585859,35)
(0.868686869,37)
(0.878787879,31)
(0.888888889,35)
(0.898989899,37)
(0.909090909,37)
(0.919191919,37)
(0.929292929,37)
(0.939393939,30)
(0.949494949,31)
(0.95959596,33)
(0.96969697,38)
(0.97979798,41)
(0.98989899,39)
(1,38)
};

\addlegendentry{Run \#4}
\addplot[
color=blue,
mark=+,
mark size=0.8pt
]
coordinates {
(0,39)
(0.01010101,39)
(0.02020202,38)
(0.03030303,41)
(0.04040404,41)
(0.050505051,41)
(0.060606061,41)
(0.070707071,42)
(0.080808081,42)
(0.090909091,42)
(0.101010101,42)
(0.111111111,41)
(0.121212121,40)
(0.131313131,39)
(0.141414141,41)
(0.151515152,42)
(0.161616162,39)
(0.171717172,41)
(0.181818182,40)
(0.191919192,40)
(0.202020202,40)
(0.212121212,40)
(0.222222222,40)
(0.232323232,41)
(0.242424242,40)
(0.252525253,40)
(0.262626263,40)
(0.272727273,40)
(0.282828283,39)
(0.292929293,39)
(0.303030303,39)
(0.313131313,39)
(0.323232323,40)
(0.333333333,40)
(0.343434343,38)
(0.353535354,38)
(0.363636364,37)
(0.373737374,38)
(0.383838384,38)
(0.393939394,38)
(0.404040404,39)
(0.414141414,38)
(0.424242424,38)
(0.434343434,37)
(0.444444444,38)
(0.454545455,38)
(0.464646465,37)
(0.474747475,37)
(0.484848485,38)
(0.494949495,38)
(0.505050505,36)
(0.515151515,37)
(0.525252525,38)
(0.535353535,38)
(0.545454545,37)
(0.555555556,36)
(0.565656566,37)
(0.575757576,37)
(0.585858586,37)
(0.595959596,37)
(0.606060606,38)
(0.616161616,38)
(0.626262626,38)
(0.636363636,39)
(0.646464646,37)
(0.656565657,39)
(0.666666667,37)
(0.676767677,37)
(0.686868687,37)
(0.696969697,36)
(0.707070707,37)
(0.717171717,38)
(0.727272727,36)
(0.737373737,36)
(0.747474747,37)
(0.757575758,37)
(0.767676768,36)
(0.777777778,34)
(0.787878788,34)
(0.797979798,35)
(0.808080808,36)
(0.818181818,34)
(0.828282828,35)
(0.838383838,36)
(0.848484848,35)
(0.858585859,37)
(0.868686869,37)
(0.878787879,36)
(0.888888889,37)
(0.898989899,31)
(0.909090909,38)
(0.919191919,31)
(0.929292929,33)
(0.939393939,32)
(0.949494949,32)
(0.95959596,32)
(0.96969697,37)
(0.97979798,29)
(0.98989899,29)
(1,38)
};

\addlegendentry{Run \#5}
\addplot[
color=orange,
mark=o,
mark size=0.5pt
]
coordinates {
(0,39)
(0.01010101,41)
(0.02020202,42)
(0.03030303,41)
(0.04040404,41)
(0.050505051,40)
(0.060606061,41)
(0.070707071,40)
(0.080808081,40)
(0.090909091,41)
(0.101010101,39)
(0.111111111,39)
(0.121212121,39)
(0.131313131,38)
(0.141414141,39)
(0.151515152,41)
(0.161616162,41)
(0.171717172,38)
(0.181818182,40)
(0.191919192,40)
(0.202020202,39)
(0.212121212,42)
(0.222222222,39)
(0.232323232,39)
(0.242424242,40)
(0.252525253,38)
(0.262626263,37)
(0.272727273,40)
(0.282828283,37)
(0.292929293,35)
(0.303030303,39)
(0.313131313,39)
(0.323232323,39)
(0.333333333,41)
(0.343434343,38)
(0.353535354,40)
(0.363636364,41)
(0.373737374,40)
(0.383838384,41)
(0.393939394,39)
(0.404040404,40)
(0.414141414,38)
(0.424242424,40)
(0.434343434,41)
(0.444444444,41)
(0.454545455,39)
(0.464646465,41)
(0.474747475,41)
(0.484848485,39)
(0.494949495,39)
(0.505050505,39)
(0.515151515,38)
(0.525252525,40)
(0.535353535,39)
(0.545454545,40)
(0.555555556,37)
(0.565656566,40)
(0.575757576,38)
(0.585858586,38)
(0.595959596,38)
(0.606060606,40)
(0.616161616,37)
(0.626262626,37)
(0.636363636,37)
(0.646464646,37)
(0.656565657,37)
(0.666666667,36)
(0.676767677,37)
(0.686868687,38)
(0.696969697,37)
(0.707070707,35)
(0.717171717,34)
(0.727272727,35)
(0.737373737,34)
(0.747474747,32)
(0.757575758,33)
(0.767676768,34)
(0.777777778,34)
(0.787878788,33)
(0.797979798,33)
(0.808080808,33)
(0.818181818,33)
(0.828282828,34)
(0.838383838,35)
(0.848484848,34)
(0.858585859,35)
(0.868686869,34)
(0.878787879,35)
(0.888888889,34)
(0.898989899,33)
(0.909090909,34)
(0.919191919,31)
(0.929292929,32)
(0.939393939,32)
(0.949494949,31)
(0.95959596,31)
(0.96969697,31)
(0.97979798,30)
(0.98989899,31)
(1,31)
};

\addlegendentry{Average}
\addplot[
color=black,
line width=2pt
]
coordinates {
(0,39)
(0.01010101,41)
(0.02020202,41.4)
(0.03030303,40.4)
(0.04040404,40.4)
(0.050505051,39.8)
(0.060606061,39.8)
(0.070707071,40.2)
(0.080808081,40.4)
(0.090909091,40.8)
(0.101010101,41)
(0.111111111,40.8)
(0.121212121,39.8)
(0.131313131,39.2)
(0.141414141,40.4)
(0.151515152,40.8)
(0.161616162,39.4)
(0.171717172,39.2)
(0.181818182,39.2)
(0.191919192,39.2)
(0.202020202,39.4)
(0.212121212,40.4)
(0.222222222,39.2)
(0.232323232,39.2)
(0.242424242,39.2)
(0.252525253,39)
(0.262626263,39)
(0.272727273,39.4)
(0.282828283,38.2)
(0.292929293,37.2)
(0.303030303,38.8)
(0.313131313,38.2)
(0.323232323,38.8)
(0.333333333,39.2)
(0.343434343,37.6)
(0.353535354,38.6)
(0.363636364,37.8)
(0.373737374,37.4)
(0.383838384,38.4)
(0.393939394,37.6)
(0.404040404,36.8)
(0.414141414,37.8)
(0.424242424,37.8)
(0.434343434,37.6)
(0.444444444,38)
(0.454545455,37.4)
(0.464646465,37.6)
(0.474747475,37.4)
(0.484848485,37.4)
(0.494949495,37.4)
(0.505050505,37)
(0.515151515,37.6)
(0.525252525,38.2)
(0.535353535,38)
(0.545454545,38.4)
(0.555555556,37.8)
(0.565656566,38.6)
(0.575757576,38.8)
(0.585858586,39.4)
(0.595959596,39)
(0.606060606,39.2)
(0.616161616,38.2)
(0.626262626,38.2)
(0.636363636,37.8)
(0.646464646,37.4)
(0.656565657,38.2)
(0.666666667,36.6)
(0.676767677,37.4)
(0.686868687,37)
(0.696969697,36)
(0.707070707,36)
(0.717171717,36.2)
(0.727272727,35)
(0.737373737,34.8)
(0.747474747,35)
(0.757575758,35.4)
(0.767676768,35)
(0.777777778,33.6)
(0.787878788,34.4)
(0.797979798,34.2)
(0.808080808,36.6)
(0.818181818,34.8)
(0.828282828,35.2)
(0.838383838,34.6)
(0.848484848,35)
(0.858585859,34.8)
(0.868686869,35.4)
(0.878787879,33.8)
(0.888888889,34.2)
(0.898989899,33)
(0.909090909,34.2)
(0.919191919,32.6)
(0.929292929,32.6)
(0.939393939,31.2)
(0.949494949,31.4)
(0.95959596,32)
(0.96969697,33.2)
(0.97979798,31.4)
(0.98989899,31.2)
(1,32.6)
};

\end{axis}
\end{tikzpicture}

%% file: Figures/optimalPlcmnt_numEpochs.tex
\begin{tikzpicture}
\begin{axis}[
xlabel={Fraction of new moderate-minded agents w.r.t population size},
ylabel={Time to equilibrium ($t_{eqm}$)},
xmin=0, xmax=1.05,
ymin=0, ymax=800,
xtick={0,0.1,0.2,0.3,0.4,0.5,0.6,0.7,0.8,0.9,1.0},
ytick={0,200,400,600,800},
xticklabel style={rotate=60},
legend style={at={(0,1)},anchor=north west,font=\footnotesize},
ymajorgrids=true,
grid style=dashed
]

\addlegendentry{Average of Random Placement}
\addplot[
color=black,
line width=1pt,
mark=o,
mark size=1pt
]
coordinates {
(0,24)
(0.01010101,23.8)
(0.02020202,25.2)
(0.03030303,25)
(0.04040404,25.4)
(0.050505051,28)
(0.060606061,27.2)
(0.070707071,29.8)
(0.080808081,28.8)
(0.090909091,30.6)
(0.101010101,34.6)
(0.111111111,31.8)
(0.121212121,32.8)
(0.131313131,37)
(0.141414141,35.6)
(0.151515152,32.8)
(0.161616162,38.4)
(0.171717172,39.8)
(0.181818182,37.2)
(0.191919192,40)
(0.202020202,39.4)
(0.212121212,39.4)
(0.222222222,41.8)
(0.232323232,48)
(0.242424242,43.2)
(0.252525253,47)
(0.262626263,48.8)
(0.272727273,49)
(0.282828283,49.4)
(0.292929293,54.4)
(0.303030303,50.8)
(0.313131313,51.8)
(0.323232323,53)
(0.333333333,54)
(0.343434343,56.6)
(0.353535354,55.4)
(0.363636364,54.8)
(0.373737374,57.2)
(0.383838384,58.2)
(0.393939394,59)
(0.404040404,60.2)
(0.414141414,64.2)
(0.424242424,60.4)
(0.434343434,61.2)
(0.444444444,61)
(0.454545455,62.2)
(0.464646465,63.8)
(0.474747475,63)
(0.484848485,65.6)
(0.494949495,65.4)
(0.505050505,67.2)
(0.515151515,66.6)
(0.525252525,65.8)
(0.535353535,70.8)
(0.545454545,77.8)
(0.555555556,82.4)
(0.565656566,72.6)
(0.575757576,77.2)
(0.585858586,77.4)
(0.595959596,82.4)
(0.606060606,73.8)
(0.616161616,83.2)
(0.626262626,85.8)
(0.636363636,89.8)
(0.646464646,92.4)
(0.656565657,89)
(0.666666667,100)
(0.676767677,126.6)
(0.686868687,111.2)
(0.696969697,114.4)
(0.707070707,124.8)
(0.717171717,109.6)
(0.727272727,139.4)
(0.737373737,138)
(0.747474747,129)
(0.757575758,127.8)
(0.767676768,118.2)
(0.777777778,126.2)
(0.787878788,131.2)
(0.797979798,138.4)
(0.808080808,121.8)
(0.818181818,126.2)
(0.828282828,122.6)
(0.838383838,149.8)
(0.848484848,139.8)
(0.858585859,119)
(0.868686869,129.4)
(0.878787879,160)
(0.888888889,134.2)
(0.898989899,156.6)
(0.909090909,160.2)
(0.919191919,194.4)
(0.929292929,189.4)
(0.939393939,225)
(0.949494949,219.4)
(0.95959596,210)
(0.96969697,178.4)
(0.97979798,213.8)
(0.98989899,197.6)
(1,190.6)
};

\addlegendentry{Intelligent Placement}
\addplot[
color=red,
line width=1pt,
mark=square,
mark size=1pt
]
coordinates {
(0,24)
(0.01010101,24)
(0.02020202,24)
(0.03030303,24)
(0.04040404,24)
(0.050505051,24)
(0.060606061,24)
(0.070707071,24)
(0.080808081,24)
(0.090909091,24)
(0.101010101,24)
(0.111111111,24)
(0.121212121,24)
(0.131313131,24)
(0.141414141,24)
(0.151515152,24)
(0.161616162,24)
(0.171717172,24)
(0.181818182,24)
(0.191919192,24)
(0.202020202,24)
(0.212121212,24)
(0.222222222,24)
(0.232323232,24)
(0.242424242,24)
(0.252525253,24)
(0.262626263,24)
(0.272727273,24)
(0.282828283,129)
(0.292929293,129)
(0.303030303,129)
(0.313131313,129)
(0.323232323,129)
(0.333333333,129)
(0.343434343,129)
(0.353535354,129)
(0.363636364,129)
(0.373737374,129)
(0.383838384,129)
(0.393939394,129)
(0.404040404,129)
(0.414141414,129)
(0.424242424,129)
(0.434343434,129)
(0.444444444,129)
(0.454545455,129)
(0.464646465,129)
(0.474747475,129)
(0.484848485,129)
(0.494949495,129)
(0.505050505,519)
(0.515151515,519)
(0.525252525,519)
(0.535353535,519)
(0.545454545,564)
(0.555555556,564)
(0.565656566,564)
(0.575757576,564)
(0.585858586,564)
(0.595959596,564)
(0.606060606,564)
(0.616161616,564)
(0.626262626,564)
(0.636363636,564)
(0.646464646,564)
(0.656565657,564)
(0.666666667,564)
(0.676767677,564)
(0.686868687,501)
(0.696969697,501)
(0.707070707,501)
(0.717171717,501)
(0.727272727,585)
(0.737373737,585)
(0.747474747,585)
(0.757575758,585)
(0.767676768,585)
(0.777777778,585)
(0.787878788,585)
(0.797979798,585)
(0.808080808,585)
(0.818181818,585)
(0.828282828,585)
(0.838383838,585)
(0.848484848,585)
(0.858585859,562)
(0.868686869,570)
(0.878787879,570)
(0.888888889,570)
(0.898989899,570)
(0.909090909,589)
(0.919191919,589)
(0.929292929,589)
(0.939393939,589)
(0.949494949,589)
(0.95959596,589)
(0.96969697,589)
(0.97979798,589)
(0.98989899,589)
(1,589)
};

\end{axis}
\end{tikzpicture}

%% file: Figures/optimalPlcmnt_numClusters.tex
\begin{tikzpicture}
\begin{axis}[
xlabel={Fraction of new moderate-minded agents w.r.t population size},
ylabel={\#(clusters) at equilibrium ($C_{eqm}$)},
xmin=0, xmax=1.05,
ymin=0, ymax=60,
xtick={0,0.1,0.2,0.3,0.4,0.5,0.6,0.7,0.8,0.9,1.0},
ytick={0,15,30,45,60},
xticklabel style={rotate=60},
legend style={at={(0,1)},anchor=north west,font=\footnotesize},
ymajorgrids=true,
grid style=dashed
]

\addlegendentry{Average of Random Placement}
\addplot[
color=black,
line width=1pt,
mark=o,
mark size=1pt
]
coordinates {
(0,39)
(0.01010101,41)
(0.02020202,41.4)
(0.03030303,40.4)
(0.04040404,40.4)
(0.050505051,39.8)
(0.060606061,39.8)
(0.070707071,40.2)
(0.080808081,40.4)
(0.090909091,40.8)
(0.101010101,41)
(0.111111111,40.8)
(0.121212121,39.8)
(0.131313131,39.2)
(0.141414141,40.4)
(0.151515152,40.8)
(0.161616162,39.4)
(0.171717172,39.2)
(0.181818182,39.2)
(0.191919192,39.2)
(0.202020202,39.4)
(0.212121212,40.4)
(0.222222222,39.2)
(0.232323232,39.2)
(0.242424242,39.2)
(0.252525253,39)
(0.262626263,39)
(0.272727273,39.4)
(0.282828283,38.2)
(0.292929293,37.2)
(0.303030303,38.8)
(0.313131313,38.2)
(0.323232323,38.8)
(0.333333333,39.2)
(0.343434343,37.6)
(0.353535354,38.6)
(0.363636364,37.8)
(0.373737374,37.4)
(0.383838384,38.4)
(0.393939394,37.6)
(0.404040404,36.8)
(0.414141414,37.8)
(0.424242424,37.8)
(0.434343434,37.6)
(0.444444444,38)
(0.454545455,37.4)
(0.464646465,37.6)
(0.474747475,37.4)
(0.484848485,37.4)
(0.494949495,37.4)
(0.505050505,37)
(0.515151515,37.6)
(0.525252525,38.2)
(0.535353535,38)
(0.545454545,38.4)
(0.555555556,37.8)
(0.565656566,38.6)
(0.575757576,38.8)
(0.585858586,39.4)
(0.595959596,39)
(0.606060606,39.2)
(0.616161616,38.2)
(0.626262626,38.2)
(0.636363636,37.8)
(0.646464646,37.4)
(0.656565657,38.2)
(0.666666667,36.6)
(0.676767677,37.4)
(0.686868687,37)
(0.696969697,36)
(0.707070707,36)
(0.717171717,36.2)
(0.727272727,35)
(0.737373737,34.8)
(0.747474747,35)
(0.757575758,35.4)
(0.767676768,35)
(0.777777778,33.6)
(0.787878788,34.4)
(0.797979798,34.2)
(0.808080808,36.6)
(0.818181818,34.8)
(0.828282828,35.2)
(0.838383838,34.6)
(0.848484848,35)
(0.858585859,34.8)
(0.868686869,35.4)
(0.878787879,33.8)
(0.888888889,34.2)
(0.898989899,33)
(0.909090909,34.2)
(0.919191919,32.6)
(0.929292929,32.6)
(0.939393939,31.2)
(0.949494949,31.4)
(0.95959596,32)
(0.96969697,33.2)
(0.97979798,31.4)
(0.98989899,31.2)
(1,32.6)
};

\addlegendentry{Intelligent Placement}
\addplot[
color=red,
line width=1pt,
mark=square,
mark size=1pt
]
coordinates {
(0,39)
(0.01010101,39)
(0.02020202,39)
(0.03030303,39)
(0.04040404,39)
(0.050505051,39)
(0.060606061,39)
(0.070707071,39)
(0.080808081,39)
(0.090909091,39)
(0.101010101,39)
(0.111111111,39)
(0.121212121,39)
(0.131313131,39)
(0.141414141,39)
(0.151515152,39)
(0.161616162,39)
(0.171717172,39)
(0.181818182,39)
(0.191919192,39)
(0.202020202,39)
(0.212121212,39)
(0.222222222,39)
(0.232323232,39)
(0.242424242,39)
(0.252525253,39)
(0.262626263,39)
(0.272727273,39)
(0.282828283,34)
(0.292929293,34)
(0.303030303,34)
(0.313131313,34)
(0.323232323,34)
(0.333333333,34)
(0.343434343,34)
(0.353535354,34)
(0.363636364,34)
(0.373737374,34)
(0.383838384,34)
(0.393939394,34)
(0.404040404,34)
(0.414141414,34)
(0.424242424,34)
(0.434343434,34)
(0.444444444,34)
(0.454545455,34)
(0.464646465,34)
(0.474747475,34)
(0.484848485,34)
(0.494949495,34)
(0.505050505,10)
(0.515151515,10)
(0.525252525,10)
(0.535353535,10)
(0.545454545,9)
(0.555555556,9)
(0.565656566,9)
(0.575757576,9)
(0.585858586,9)
(0.595959596,9)
(0.606060606,9)
(0.616161616,9)
(0.626262626,9)
(0.636363636,9)
(0.646464646,9)
(0.656565657,9)
(0.666666667,9)
(0.676767677,9)
(0.686868687,7)
(0.696969697,7)
(0.707070707,7)
(0.717171717,7)
(0.727272727,10)
(0.737373737,10)
(0.747474747,10)
(0.757575758,10)
(0.767676768,10)
(0.777777778,10)
(0.787878788,10)
(0.797979798,10)
(0.808080808,10)
(0.818181818,10)
(0.828282828,10)
(0.838383838,10)
(0.848484848,10)
(0.858585859,7)
(0.868686869,7)
(0.878787879,7)
(0.888888889,7)
(0.898989899,7)
(0.909090909,5)
(0.919191919,5)
(0.929292929,5)
(0.939393939,5)
(0.949494949,5)
(0.95959596,5)
(0.96969697,5)
(0.97979798,5)
(0.98989899,5)
(1,5)
};

\end{axis}
\end{tikzpicture}